\documentclass[sigconf, 10pt, nonacm]{acmart}
\usepackage{blindtext}
\usepackage[utf8]{inputenc}
\usepackage{epsfig,xspace,url,xcolor}
\usepackage{amsmath}
\usepackage{mathtools}
\usepackage{amsthm}
\usepackage{tikz}
\usetikzlibrary{positioning,calc}
\usepackage{caption}

\usepackage{graphicx}
\usepackage[inline]{enumitem}
\usepackage{tikz}
\usepackage{thmtools}
\usepackage{thm-restate}
\usepackage{scalerel}
\usepackage[ruled, lined, linesnumbered, commentsnumbered, noend]{algorithm2e}
\usepackage[noend]{algpseudocode}
\usepackage{graphicx}
\usepackage{blindtext}
\usepackage{url}
\usepackage{subcaption}
\usepackage{mathtools}
\usepackage{xcolor}
\usepackage{color, colortbl}
\usepackage{tablefootnote}
\usetikzlibrary{positioning, shapes.geometric, shadows}
\usepackage[subtle]{savetrees}
\usepackage{graphbox}
\usepackage{url}

\usepackage[normalem]{ulem}
\usepackage{wrapfig}
\usepackage{wrapstuff}
\usepackage{breakurl}
\renewcommand\footnotetextcopyrightpermission[1]{} % removes footnote with conference info
\let\ACMmaketitle=\maketitle
\renewcommand{\maketitle}{\begingroup\let\footnote=\thanks \ACMmaketitle\endgroup}
\let\height\relax

\newcommand{\name}{\textsc{Harvest}\xspace}

\usepackage[framemethod=TikZ]{mdframed}
\usepackage{enumitem}
\colorlet{rowHighlight}{green!10}
\definecolor{darkgreen}{rgb}{0.01, 0.75, 0.24}
\definecolor{darkcyan}{rgb}{0.0, 0.55, 0.55}
\definecolor{mycolor}{rgb}{0.01, 0.75, 0.24}
\definecolor{takeawaycolor}{rgb}{0.75, 0, 0}
\definecolor{mygray}{gray}{0.9}
\definecolor{myred}{rgb}{0.75, 0, 0}

\definecolor{alg1base}{rgb}{0.957, 1, 0.961}
\colorlet{alg1}{alg1base!90!white}
\definecolor{alg2base}{rgb}{0.957, 0.988, 1}
\colorlet{alg2}{alg2base!90!white}
\definecolor{alg3base}{rgb}{0.824, 0.914, 1}
\colorlet{alg3}{alg3base!90!white}

\mathchardef\hyphen="2D

\newtheoremstyle{sltheorem}
{}                % Space above
{}                % Space below
{\slshape}        % Theorem body font % (default is "\upshape")
{}                % Indent amount
{\bfseries}       % Theorem head font % (default is \mdseries)
{.}               % Punctuation after theorem head % default: no punctuation
{ }               % Space after theorem head
{}                % Theorem head spec
\theoremstyle{sltheorem}

\newtheorem{theorem}{Theorem}
\newtheorem{observation}{Observation}

\newtheorem{Lemma}{Lemma}

\newmdenv[
  backgroundcolor=gray!20,
  linecolor=gray,
  linewidth=1pt,
  roundcorner=5pt,
  skipabove=6pt,
  skipbelow=6pt,
  innertopmargin=1pt,
  innerbottommargin=4pt,
  innerleftmargin=6pt,   % tighter left margin inside the box
  innerrightmargin=6pt   % tighter right margin inside the box
]{graybox}

\newcommand{\myitem}[1]{\vspace{0.6mm}\noindent\textbf{#1}}

\newcommand{\remove}[1]{}

\newcommand{\first}{\emph{(i)}\xspace}
\newcommand{\second}{\emph{(ii)}\xspace}
\newcommand{\third}{\emph{(iii)}\xspace}
\newcommand{\fourth}{\emph{(iv)}\xspace}

\usepackage{xcolor}
\definecolor{light-gray}{gray}{0.9}
\usepackage{tcolorbox}

\newcommand{\cref}[1]{Chapter~\ref{#1}}

\newcommand{\dpv}{\ensuremath{DP}\xspace}
\newcommand{\nextv}{\ensuremath{next}\xspace}

\DeclareMathOperator*{\argmin}{arg\,min}

\newcommand*{\affaddr}[1]{#1}
\newcommand*{\affmark}[1][*]{\textsuperscript{#1}}
\DeclareSymbolFont{extraup}{U}{zavm}{m}{n}
\DeclareMathSymbol{\varheart}{\mathalpha}{extraup}{86}
\DeclareMathSymbol{\vardiamond}{\mathalpha}{extraup}{87}
\title{Harvest: Adaptive Photonic Switching Schedules\\for Collective Communication in Scale-up Domains
}

\author{%
{\rm Mahir Rahman\affmark[$\blacklozenge$]\ \ Samuel Joseph\affmark[$\blacklozenge$]\ \ Nihar Kodkani\affmark[$\blacklozenge$]\ \ Behnaz Arzani\affmark[$\spadesuit$]\ \ Vamsi Addanki\affmark[$\blacklozenge$]
}\\
\affaddr{{\rm \affmark[$\blacklozenge$]}\textit{Purdue University}}\ \ \
\affaddr{{\rm \affmark[$\spadesuit$]}\textit{Microsoft Research}}
}

\renewcommand{\shortauthors}{Rahman et al.}

\begin{document}

\begin{abstract}

As chip-to-chip silicon photonics gain traction for their bandwidth and energy efficiency, their circuit-switched nature raises a fundamental question for collective communication: \emph{when} and \emph{how} should the interconnect be reconfigured to realize these benefits?
Establishing direct optical paths can reduce congestion and propagation delay, but each reconfiguration incurs non-negligible overhead, making naive per-step reconfiguration impractical.

We present \name, a systematic approach for synthesizing topology reconfiguration schedules that minimize collective completion time in photonic interconnects.
Given a collective communication algorithm and its \emph{fixed} communication schedule, \name determines how the interconnect should evolve over the course of the collective, explicitly balancing reconfiguration delay against congestion and propagation delay. We reduce the synthesis problem into a dynamic program with an underlying topology optimization subproblem and show that the approach applies to arbitrary collective communication algorithms. Furthermore, we exploit the algorithmic structure of a well-known AllReduce algorithm (Recursive Doubling) to synthesize \emph{optimal} reconfiguration schedules without using any optimizers.
By parameterizing the formulation using reconfiguration delay, \name naturally adapts to various photonic technologies. Using packet-level and flow-level evaluations, as well as hardware emulation on commercial GPUs, we show that the schedules synthesized by \name significantly reduce collective completion time across multiple collective algorithms compared to static interconnects and reconfigure-every-step baselines.

\end{abstract}

\maketitle
\pagestyle{plain}
\thispagestyle{plain}

\section{Introduction}
\label{sec:intro}

The explosive growth of AI/ML workloads~\cite{NEURIPS2020_1457c0d6,shoeybi2019megatron,10.1145/3458817.3476209,10.1145/3579371.3589350,295551,10.1145/3651890.3672265,10.1145/3651890.3672233}, together with the increasing scale of distributed computing infrastructure~\cite{10.1145/3696348.3696893,10.1145/3651890.3672265,10.1145/3651890.3672233}, has led to rapidly rising demands on network bandwidth and energy efficiency. The performance of these workloads critically depends on collective communication among GPUs, such as AllReduce and All-to-All~\cite{295653,10.1145/3437801.3441620,285084,10.1145/3651890.3672249}. Modern hyperscale systems consist of large numbers of multi-GPU servers interconnected by packet-switched networks that support GPU-to-GPU communication~\cite{nvidia2023superpod}.
Despite their widespread adoption, these interconnects face fundamental limitations. Electrical links are power-hungry and generate significant heat~\cite{10946778}, raising concerns for scalability and sustainability. At the same time, intra-server GPU interconnects rely on CMOS-based technologies whose bandwidth has not kept pace with GPU compute growth~\cite{10.1145/3579371.3589350,10.1145/3387514.3406221,10.1145/3098822.3098838,10.1145/3651890.3672273}. As the slowdown of Moore's Law continues to widen the gap between computation and communication~\cite{10.1145/3387514.3406221}, these constraints become especially pronounced in scale-up systems, where limited interconnect bandwidth, such as PCIe, increasingly bottlenecks collective performance.

Silicon photonics promise substantially higher bandwidth and improved energy efficiency~\cite{9007742,10.1145/3696348.3696856,10.1145/3452296.3472900,Ding:25}, making them an attractive alternative.
At the same time, their circuit-switched nature introduces new challenges for collective communication due to non-negligible reconfiguration delays. A static photonic topology avoids reconfiguration overhead, but inevitably suffers from congestion caused by multi-hop forwarding between GPUs. Conversely, a dynamically reconfigurable topology can, in principle, eliminate congestion by establishing direct optical paths between communicating GPUs, but only at the cost of reconfiguration delay. Balancing this fundamental tradeoff is essential to realizing the practical benefits of silicon photonics.

Prior work on optical circuit-switched networks largely falls into three categories.
One line of work advocates one-shot or infrequent reconfiguration when reconfiguration overhead is high~\cite{285119}, effectively treating the topology as static during execution.
Another line assumes that reconfiguration overhead is negligible and relies on periodic or demand-aware reconfiguration~\cite{10.1145/3651890.3672248,10.1145/3519935.3520020,10.1145/3579312,10.1145/2934872.2934911,10.1145/2486001.2486007,10.1145/3651890.3672222}, often using Birkhoff--von Neumann (BvN) decompositions~\cite{birkhoff1946three} of aggregate traffic matrices.
A third category explicitly incorporates reconfiguration delay into the optimization objective~\cite{10.1145/2896377.2901479,1230204,10.1145/2716281.2836126}, but still adopts a traffic-matrix abstraction and restricts routing to single-hop paths within each topology choice\footnote{These works permit multi-hop forwarding only across reconfiguration events~\cite{10.1145/2896377.2901479}, or rely on auxiliary electrical interconnects when reconfiguration delays are high~\cite{10.1145/2716281.2836126}.}.

While appropriate for bulk or steady-state traffic, these abstractions fundamentally ignore the step-wise structure of collective communication.
Unlike bulk traffic, collective communication is \emph{staged}: progress unfolds over a sequence of steps, and the communication pattern at each step has \emph{known} dependencies on prior ones.
This structure creates an opportunity to plan reconfigurations across multiple steps, exploiting future knowledge to amortize reconfiguration delay against reductions in congestion and propagation delay.
However, existing approaches, whether they ignore reconfiguration overhead, restrict reconfiguration during execution, or optimize over aggregate traffic matrices, fail to capture these step-level dependencies.
As a result, they forgo a significant opportunity to reduce collective completion time.
This gap in the literature raises a natural question:

\smallskip
\emph{To what extent can reconfigurability be exploited to reduce collective completion time?}
\smallskip

Answering this question requires deciding \emph{when} and \emph{how} to reconfigure a photonic interconnect \emph{during} a collective, while balancing reconfiguration delay against congestion and propagation delay.
Our approach is grounded in two observations that together provide a principled way to reason about this trade-off.

First, unlike prior work that applies Birkhoff--von Neumann (BvN) decompositions to \emph{aggregate} traffic matrices, many collective communication algorithms admit a natural BvN representation at the level of individual communication steps~\cite{birkhoff1946three}.
Each step corresponds to a matching, and the collective as a whole can be viewed as a weighted sequence of such matchings.
This structure arises directly from the staged nature of collective primitives, which are traditionally designed around point-to-point communication patterns\footnote{We generalize beyond the point-to-point communication model later in the paper (\S\ref{sec:cct-model}).}~\cite{https://doi.org/10.1002/cpe.1206}. Second, this representation connects naturally to performance analysis.

Based on these observations, we can express the completion time of each step through the lens of maximum concurrent flow~\cite{10.1145/77600.77620},which permits multi-hop forwarding under a chosen topology. Interestingly, this representation uncovers the classic $\alpha$--$\beta$ cost model for collective communication, explicitly accounting for network congestion.

The model can then be extended for reconfigurable interconnects by explicitly accounting for reconfiguration overhead, where each topology change contributes an additional delay $\alpha_r$ to the overall completion time. This new formulation quantifies collective completion time in a way that explicitly captures reconfiguration overhead.

Finally, we cast this formulation as an optimization problem that synthesizes circuit-switching schedules which adapt to the underlying reconfiguration delay and determine \emph{when} and \emph{how} the interconnect should reconfigure to minimize total collective completion time.

We present \name, a framework that synthesizes optimal photonic switching schedules for any \emph{given} collective communication algorithm.
The key insight underlying \name is that topology synthesis exhibits a natural recurrence over contiguous ranges of communication steps, which enables a dynamic programming formulation.
Each subproblem selects an optimal topology for a sequence of steps executed without reconfiguration.
The resulting subproblem resembles degree-bounded, demand-aware network design, with a crucial distinction: communication demand is not available upfront, but is revealed sequentially, as each step of the collective depends on the completion of prior ones.
To capture this temporal structure, we formulate the subproblem as a Mixed-Integer Second-Order Conic Program (MISOCP) and integrate it with the dynamic program to construct a globally optimal reconfiguration schedule.
We synthesize topologies offline, computing the schedule once and reusing it across all executions of the collective.
Our framework applies to arbitrary collective communication algorithms and is suitable for scale-up domains with typical network sizes ranging from $8$ to $64$ GPUs.

We further apply \name to the recursive doubling AllReduce algorithm~\cite{10.1007/978-3-540-24685-5-1}, which exhibits additional structure that substantially simplifies schedule synthesis.
Exploiting this structure, we show that optimal topology reconfigurations can be computed with polylogarithmic complexity, and empirically within tens of microseconds, even for interconnects with up to $1024$ GPUs.

We evaluate \name using extensive packet-level simulations in Astra-Sim~\cite{10158106}, numerical evaluations using Gurobi~\cite{gurobi}, and testbed emulation.
Across a range of collective algorithms, including Swing~\cite{295653}, recursive doubling~\cite{kolmakov2020generalization}, Bine butterflies~\cite{10.1145/3712285.3759835}, binomial trees, and Bruck's algorithms~\cite{642949}, the topologies synthesized by \name reduce collective completion time by up to $\approx 2$x even compared to the best strategy among static topologies and BvN schedules that reconfigure at every step.
We further measure the synchronization overhead introduced by in-collective reconfigurations on a testbed with NVIDIA GPUs and find that these overheads are negligible, on the order of a few microseconds.

\smallskip
In summary, our main contributions are:
\begin{itemize}[label=\small{\textcolor{black}{$\blacksquare$}}, leftmargin=*]
	\item \name, a systematic approach for navigating the trade-off between congestion, propagation delay, and reconfiguration delay in photonic interconnects for collective communication. \name synthesizes optimal photonic switching schedules by combining dynamic programming with topology optimization and applies to arbitrary collective algorithms.
	\item Structural insights into recursive doubling AllReduce that enable optimal schedule synthesis within polylogarithmic complexity, eliminating the need for mixed-integer optimization in this special case.
	\item An extensive evaluation using Astra-Sim, flow-level simulations, and hardware emulation on commercial GPUs, demonstrating significant performance improvements over static and reconfigure-every-step baselines.
	\item Public release of all artifacts as open-source.
\end{itemize}

\emph{This work does not raise any ethical issues.}

\section{Background \& Motivation}
\label{sec:bottlenecks}

\begin{figure}[t]
  \centering
  \includegraphics[width=0.8\linewidth]{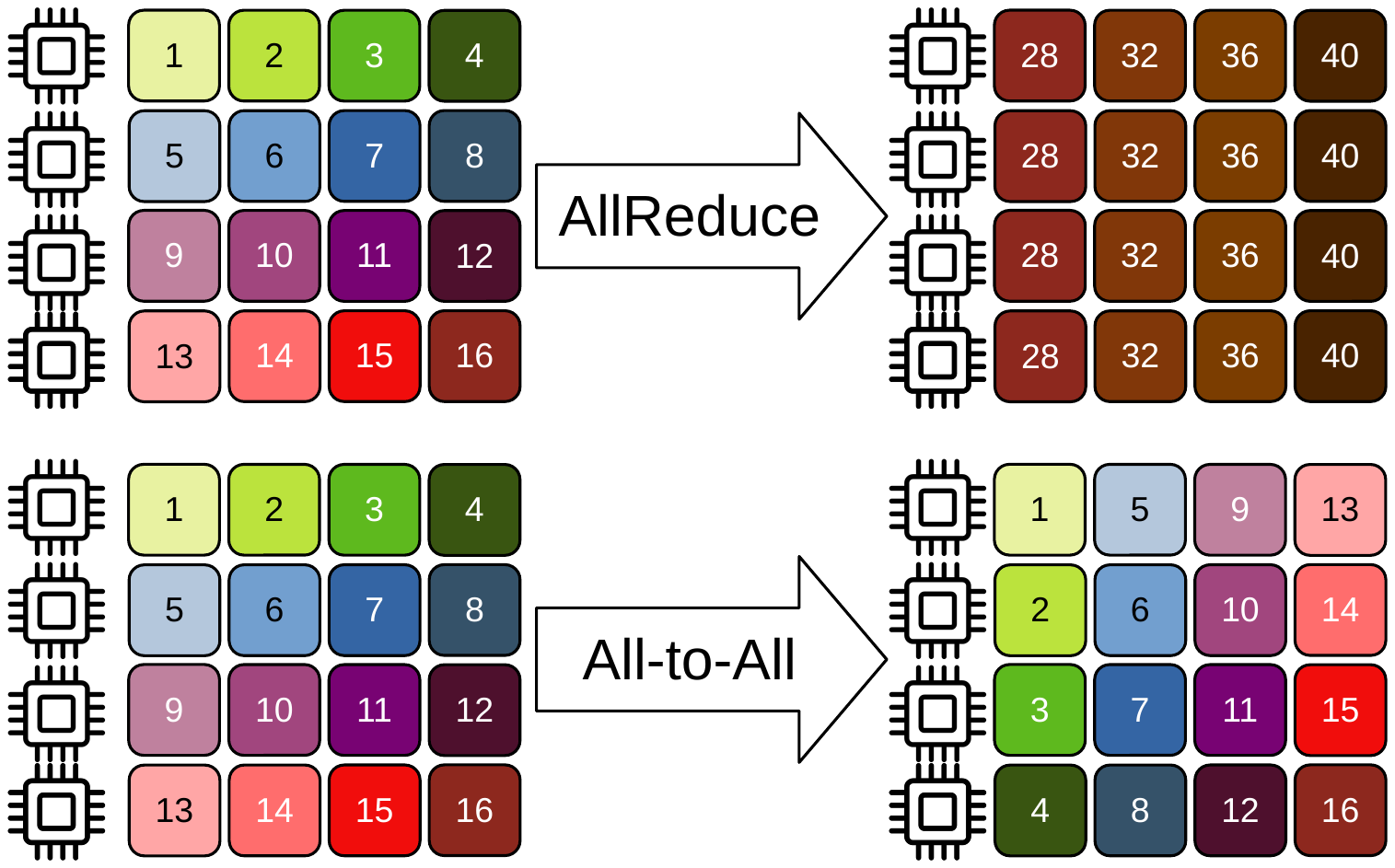}
  \vspace{-6mm}
  \caption{Collective communication primitives}
  \label{fig:allreduce-alltoall}
  \vspace{-4mm}
\end{figure}

Unlike traditional datacenter applications, the collective operations among GPUs in AI/ML workloads result in staged and highly structured communication patterns. Among these collectives, AllReduce and All-to-All are especially prevalent~\cite{shoeybi2019megatron,10.1145/3458817.3476209,10.1145/3651890.3672265,9355301,lepikhin2021gshard,NEURIPS2022_2f00ecd7,pmlr-v139-lewis21a,pmlr-v162-rajbhandari22a}.

All reduce, as the name suggests ``reduces'' (e.g., sum) the data each GPU holds and then distributes the results to all others. In All-to-All (Figure~\ref{fig:allreduce-alltoall}) each GPU $j$ delivers a portion of the data it holds (block $i$) to another GPU ($i$).
A wide range of algorithms exist for both of these collectives~\cite{295653,10.1007/978-3-540-24685-5-1,10.1145/2686882,642949}. Both primitives are bandwidth-intensive and latency-sensitive.

\myitem{Limits of topology-aware collectives:} Many prior work design collective algorithms for~\emph{specific} network topologies using the classic $\alpha$--$\beta$ cost model~\cite{10.1145/3437801.3441620,285084,10.1145/3651890.3672249}.
While these topology-aware algorithms can improve collective efficiency on fixed interconnects, they inherit the rigidity of the static networks on which they run.
Multi-step collectives~\cite{10.1007/978-3-540-24685-5-1}, repeatedly exchange data across different pairs of communication partners.
Under a fixed topology, some of these exchanges must traverse longer or congested paths, increasing both latency and bandwidth consumption~\cite{295653}.
As a result, static networks must provision for the worst-case demand over the entire lifetime of the collective, which often leads to underutilization when communication is sparse or staged.
Techniques such as pipelining or mirroring can partially mitigate these effects~\cite{295653} but they cannot fully overcome the limitations.

\myitem{Throughput modeling and BvN decompositions:}
The maximum concurrent flow framework~\cite{10.1145/77600.77620} is a standard tool for reasoning about network throughput and congestion~\cite{jyothi2016measuring,highthroughputSingla,10.1145/3452296.3472913,10.1145/3579312,10.1145/3519935.3520020}.
It connects naturally to Birkhoff--von Neumann (BvN) decompositions, which express an aggregate traffic matrix as a convex combination of matchings~\cite{10.1145/3452296.3472913}.
This abstraction underlies many approaches for synthesizing circuit-switching schedules in demand-aware networks~\cite{10.1145/2486001.2486007,285119,10.1145/3452296.3472900}.
However, BvN decompositions and traffic-matrix-based formulations inherently assume that all communication demand is available simultaneously.
Collective communication violates this assumption: collectives generate and consume data in a strict sequence, and later communication steps cannot begin until earlier steps complete.
As a result, static traffic-matrix decompositions fail to capture the temporal dependencies that are intrinsic to many collective algorithms.

\myitem{Programmable but costly reconfiguration:}
Reconfigurable photonic interconnects enable the network topology to adapt to the communication pattern of each collective step and can potentially reduce congestion and improve throughput~\cite{Ding:25,10.1145/3696348.3696856,10.1145/3748273.3749203}. But this flexibility has a cost. The reconfiguration delay is high in photonic interconnects which can negate any gains if we are not careful~\cite{Ding:25}.
Much of the existing literature either assumes that reconfiguration overheads are negligible or avoids reconfiguration altogether when they are not.

We advocate a more principled perspective that bridges the staged structure of collective algorithms, the limits imposed by network throughput, and the real cost of reconfiguration.

\begin{figure*}
\centering
\begin{minipage}{0.68\linewidth}
\begin{subfigure}{0.33\linewidth}
\centering
\begin{tikzpicture}[scale=1]
  \def\n{8}              % number of nodes
  \def\r{1.15cm}            % radius of the big circle
  \def\start{90}         % start angle (90 = top)

  % Big circle
  \draw[thick] (0,0) circle (\r);

  % Place nodes on the circumference
  \foreach \i in {1,...,\n} {
    \pgfmathsetmacro{\ang}{\start - 360*(\i-1)/\n}
    \node[circle,draw,fill=white,inner sep=1pt,font=\small] (N\i) at (\ang:\r) {\i};
  }

  % Arrows inside the circle: i -> i+1 (wrap 8 -> 1)
  \foreach \i [evaluate=\i as \j using {int(mod(\i,\n)+1)}] in {1,...,\n} {
    \draw[->,thick,color=red,shorten >=1pt,shorten <=1pt] (N\i) -- (N\j);
  }
\end{tikzpicture}
\caption{Step $1$ - $G_1$\\Congestion $=1$}
\label{fig:step1-g1}
\end{subfigure}
\begin{subfigure}{0.33\linewidth}
\centering
\begin{tikzpicture}[scale=1]
  \def\n{8}              % number of nodes
  \def\r{1.15cm}            % radius of the big circle
  \def\start{90}         % start angle (90 = top)

  % Big circle
  \draw[thick] (0,0) circle (\r);

  % Place nodes on the circumference
  \foreach \i in {1,...,\n} {
    \pgfmathsetmacro{\ang}{\start - 360*(\i-1)/\n}
    \node[circle,draw,fill=white,inner sep=1pt,font=\small] (N\i) at (\ang:\r) {\i};
  }

  % Arrows inside the circle: i -> i+1 (wrap 8 -> 1)
  \foreach \i [evaluate=\i as \j using {int(mod(\i+1,\n)+1)}] in {1,...,\n} {
    \draw[->,thick,color=red,shorten >=1pt,shorten <=1pt] (N\i) -- (N\j);
  }
\end{tikzpicture}
\caption{Step $2$ - $G_1$\\Congestion $=2$}
\label{fig:step2-g1}
\end{subfigure}
\begin{subfigure}{0.33\linewidth}
\centering
\begin{tikzpicture}[scale=1]
  \def\n{8}              % number of nodes
  \def\r{1.15cm}            % radius of the big circle
  \def\start{90}         % start angle (90 = top)

  % Big circle
  \draw[thick] (0,0) circle (\r);

  % Place nodes on the circumference
  \foreach \i in {1,...,\n} {
    \pgfmathsetmacro{\ang}{\start - 360*(\i-1)/\n}
    \node[circle,draw,fill=white,inner sep=1pt,font=\small] (N\i) at (\ang:\r) {\i};
  }

  % Arrows inside the circle: i -> i+1 (wrap 8 -> 1)
  \foreach \i [evaluate=\i as \j using {int(mod(\i+3,\n)+1)}] in {1,...,\n} {
    \draw[->,thick,color=red,shorten >=1pt,shorten <=1pt] (N\i) -- (N\j);
  }
\end{tikzpicture}
\caption{Step $3$ - $G_3$\\Congestion $=4$}
\label{fig:step3-g1}
\end{subfigure}
\begin{subfigure}{0.33\linewidth}
\centering
\begin{tikzpicture}[scale=1]
  \def\n{4}        % number of nodes
  \def\r{0.75cm}    % radius
  \def\start{90}   % start angle (90 = top)

  % Draw the circle
  \draw[thick] (0,0) circle (\r);

  % Place nodes with custom INTEGER labels (in this exact order)
  % We count i=1..n while iterating labels directly.
  \foreach [count=\i from 1] \lab in {1,3,5,7} {
    \pgfmathsetmacro{\ang}{\start - 360*(\i-1)/\n}
    \node[circle,draw,fill=white,inner sep=1pt,font=\small]
         (N\i) at (\ang:\r) {\lab};
  }

  % (Optional) arrows in label order: 10->20->30->40->10
  \foreach \i [evaluate=\i as \j using {int(mod(\i,\n)+1)}] in {1,...,\n} {
    \draw[->,thick,shorten >=1pt,shorten <=1pt,color=red] (N\i) -- (N\j);
  }
\end{tikzpicture}
\begin{tikzpicture}[scale=1]
  \def\n{4}
  \def\r{0.75cm}
  \def\start{90}
  \draw[thick] (0,0) circle (\r);
  \foreach [count=\i from 1] \lab in {2,4,6,8} {
    \pgfmathsetmacro{\ang}{\start - 360*(\i-1)/\n}
    \node[circle,draw,fill=white,inner sep=1pt,font=\small]
         (N\i) at (\ang:\r) {\lab};
  }
  \foreach \i [evaluate=\i as \j using {int(mod(\i,\n)+1)}] in {1,...,\n} {
    \draw[->,thick,color=red,shorten >=1pt,shorten <=1pt] (N\i) -- (N\j);
  }
\end{tikzpicture}
\caption{Step 2 - $G_2$\\Congestion $=1$}
\label{fig:step2-g2}
\end{subfigure}
\begin{subfigure}{0.33\linewidth}
\centering
\begin{tikzpicture}[scale=1]
  \def\n{4}
  \def\r{0.75cm}
  \def\start{90}
  \draw[thick] (0,0) circle (\r);
  \foreach [count=\i from 1] \lab in {1,3,5,7} {
    \pgfmathsetmacro{\ang}{\start - 360*(\i-1)/\n}
    \node[circle,draw,fill=white,inner sep=1pt,font=\small]
         (N\i) at (\ang:\r) {\lab};
  }
  \foreach \i [evaluate=\i as \j using {int(mod(\i+1,\n)+1)}] in {1,...,\n} {
    \draw[->,thick,color=red,shorten >=1pt,shorten <=1pt] (N\i) -- (N\j);
  }
\end{tikzpicture}
\begin{tikzpicture}[scale=1]
  \def\n{4}
  \def\r{0.75cm}
  \def\start{90}
  \draw[thick] (0,0) circle (\r);
  \foreach [count=\i from 1] \lab in {2,4,6,8} {
    \pgfmathsetmacro{\ang}{\start - 360*(\i-1)/\n}
    \node[circle,draw,fill=white,inner sep=1pt,font=\small]
         (N\i) at (\ang:\r) {\lab};
  }
  \foreach \i [evaluate=\i as \j using {int(mod(\i+1,\n)+1)}] in {1,...,\n} {
    \draw[->,thick,color=red,shorten >=1pt,shorten <=1pt] (N\i) -- (N\j);
  }
\end{tikzpicture}
\caption{Step 3 - $G_2$\\Congestion $=2$}
\label{fig:step3-g2}
\end{subfigure}
\begin{subfigure}{0.33\linewidth}
\centering
\begin{tikzpicture}[scale=1]
  \def\n{2}
  \def\r{0.375cm}
  \def\start{90}
  \draw[thick] (0,0) circle (\r);
  \foreach [count=\i from 1] \lab in {1,5} {
    \pgfmathsetmacro{\ang}{\start - 360*(\i-1)/\n}
    \node[circle,draw,fill=white,inner sep=1pt,font=\small]
         (N\i) at (\ang:\r) {\lab};
  }
  \foreach \i [evaluate=\i as \j using {int(mod(\i,\n)+1)}] in {1,...,\n} {
    \draw[->,thick,color=red,shorten >=1pt,shorten <=1pt] (N\i) -- (N\j);
  }
\end{tikzpicture}
\begin{tikzpicture}[scale=1]
  \def\n{2}
  \def\r{0.375cm}
  \def\start{90}
  \draw[thick] (0,0) circle (\r);
  \foreach [count=\i from 1] \lab in {2,6} {
    \pgfmathsetmacro{\ang}{\start - 360*(\i-1)/\n}
    \node[circle,draw,fill=white,inner sep=1pt,font=\small]
         (N\i) at (\ang:\r) {\lab};
  }
  \foreach \i [evaluate=\i as \j using {int(mod(\i,\n)+1)}] in {1,...,\n} {
    \draw[->,thick,color=red,shorten >=1pt,shorten <=1pt] (N\i) -- (N\j);
  }
\end{tikzpicture}
\begin{tikzpicture}[scale=1]
  \def\n{2}
  \def\r{0.375cm}
  \def\start{90}
  \draw[thick] (0,0) circle (\r);
  \foreach [count=\i from 1] \lab in {3,7} {
    \pgfmathsetmacro{\ang}{\start - 360*(\i-1)/\n}
    \node[circle,draw,fill=white,inner sep=1pt,font=\small]
         (N\i) at (\ang:\r) {\lab};
  }
  \foreach \i [evaluate=\i as \j using {int(mod(\i,\n)+1)}] in {1,...,\n} {
    \draw[->,thick,color=red,shorten >=1pt,shorten <=1pt] (N\i) -- (N\j);
  }
\end{tikzpicture}
\begin{tikzpicture}[scale=1]
  \def\n{2}
  \def\r{0.375cm}
  \def\start{90}
  \draw[thick] (0,0) circle (\r);
  \foreach [count=\i from 1] \lab in {4,8} {
    \pgfmathsetmacro{\ang}{\start - 360*(\i-1)/\n}
    \node[circle,draw,fill=white,inner sep=1pt,font=\small]
         (N\i) at (\ang:\r) {\lab};
  }
  \foreach \i [evaluate=\i as \j using {int(mod(\i,\n)+1)}] in {1,...,\n} {
    \draw[->,thick,color=red,shorten >=1pt,shorten <=1pt] (N\i) -- (N\j);
  }
\end{tikzpicture}
\caption{Step 3 - $G_3$\\Congestion $=1$}
\label{fig:step3-g3}
\end{subfigure}
\caption{Various configurations of topology and communication patterns during recursive doubling allreduce algorithm. Black lines indicate the physical topology, and the red arrows indicate the communication between nodes.}
\label{fig:rd-configurations}
\end{minipage}\hfill
\begin{minipage}{0.31\linewidth}
\centering
\includegraphics[width=1\linewidth]{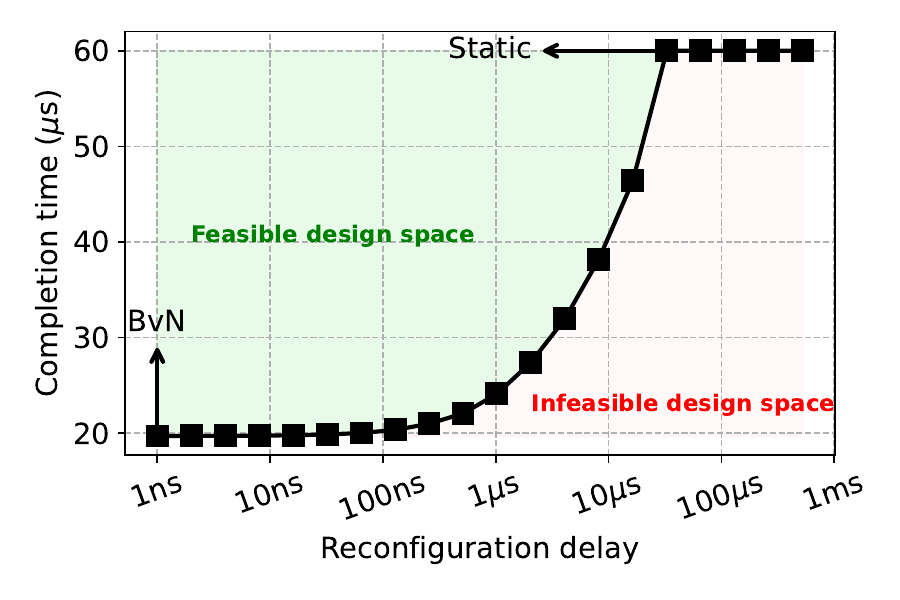}
\vspace{-4mm}
\caption{Reconfiguration delay--aware circuit-switching schedules for recursive doubling reveal the full design spectrum between BvN and static topologies. The black curve denotes a lower bound on completion time, beyond which no reconfiguration schedule can achieve further improvement.}
\vspace{-5mm}
\label{fig:intro}
\end{minipage}
\end{figure*}

\section{A Theory for Adaptive Scaleup Domains}
\label{sec:modeling}
We first describe our architecture (\S\ref{sec:architecture}), and motivate the case for reconfiguration delay-aware circuit-switching with an example (\S\ref{sec:example}). We then revisit modeling the completion time of collectives (\S\ref{sec:cct-model}), revealing an optimization opportunity to account for interconnect reconfiguration delays (\S\ref{sec:opportunity}). 

\subsection{Architecture and Assumptions}
\label{sec:architecture}

\myitem{Interconnect:}
We consider a scale-up domain with $n$ GPUs, each equipped with an electrical-to-optical transceiver (e.g., TeraPhy~\cite{9007742}) of capacity $c$.
All $n$ transceivers connect to a photonic interconnect with $n$ ports, which can establish direct optical paths between pairs of ports, thereby enabling GPU-to-GPU communication~\cite{torrijos2026industry}.
The interconnect is programmable and supports dynamic reconfiguration of optical paths on demand~\cite{10.1145/3696348.3696856}. Either a central controller controls the interconnect or it is passive (transceivers can tune the wavelength of the emitted light).
In the latter case, wavelength-selective switching within the photonic fabric establishes direct paths between ports without centralized control.
In both designs, reconfiguring the interconnect incurs a non-negligible delay, denoted by $\alpha_r$.
Several photonic technologies incur reconfiguration delays that depend on the number of ports involved in the reconfiguration~\cite{Ding:25}.

We assume the reconfiguration delay, $\alpha_r$, is constant (e.g., based on the total port count).
Our framework can be extended to capture port-dependent or topology-dependent reconfiguration delays.
We assume that all GPUs reside within a single scale-up domain and have fast access to shared memory, as in modern systems such as DGX-class servers~\cite{9895480}.
This enables GPUs to synchronize efficiently using a barrier before a collective step, perform reconfiguration synchronously if needed, and then proceed with communication. We assume collectives with $n$ GPUs. We can also selectively apply our framework to subsets of GPUs (where we reconfigure a subset of ports when necessary).

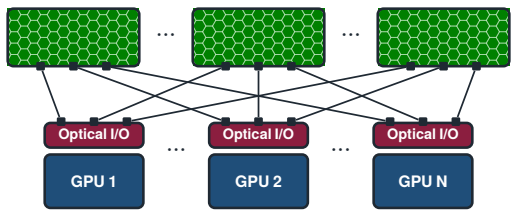
\begin{figure}[t]
\centering
\resizebox{0.8\columnwidth}{!}{
\begin{tikzpicture}
% ---- Color Palette ----
\definecolor{imgBlue}{HTML}{1F4E79}   % Original Light Blue
\definecolor{imgRed}{HTML}{8B1E3F}    % Original Salmon Red
\definecolor{imgDark}{HTML}{1F2933}   % Dark border color
\tikzset{
    % General "Button" Style matching the image blocks
    block/.style={
        draw=imgDark, 
        line width=1.8pt, 
        rounded corners=4pt, 
        font=\sffamily\bfseries,
        text=imgDark
    },
    % Specific Component Styles - Now at 50% saturation for a "deep whitish" look
    gpu/.style={
        block,
        fill=imgBlue,         % 50% Blue tint
        minimum width=2.4cm, 
        minimum height=1.3cm,
        font=\fontfamily{qhv}\selectfont\bfseries\large
    },
    opticalio/.style={
        block,
        fill=imgRed,          % 50% Red tint
        minimum width=2.4cm, 
        minimum height=0.6cm,
        font=\fontfamily{qhv}\selectfont\bfseries
    },
    % Links
    link/.style={
        draw=imgDark, 
        line width=1.2pt, 
        >=latex
    },
    port/.style={
        rectangle, 
        rounded corners=1pt, 
        fill=imgDark,
        minimum width=6pt, 
        minimum height=6pt, 
        inner sep=0pt
    }
}

% ---------------- GPUs (Bottom Layer) ----------------
\node[gpu,text=white] (gpu1) at (-4,0) {GPU 1};
\node[opticalio, text=white, above=3pt of gpu1] (io1) {Optical I/O};

\node[gpu,text=white] (gpu2) at (0,0) {GPU 2};
\node[opticalio, text=white, above=3pt of gpu2] (io2) {Optical I/O};

\node[gpu,text=white] (gpu3) at (4,0) {GPU N};
\node[opticalio, text=white, above=3pt of gpu3] (io3) {Optical I/O};

% Ellipsis dots
\node[font=\sffamily\bfseries\huge, color=imgDark] at (2,0.7) {$\cdots$};
\node[font=\sffamily\bfseries\huge, color=imgDark] at (-2,0.7) {$\cdots$};

% ---------------- Ports ----------------
\foreach \io in {io1,io2,io3}{
    \node[port] (\io p1) at ([xshift=-0.8cm] \io.north) {};
    \node[port] (\io p2) at (\io.north) {};
    \node[port] (\io p3) at ([xshift=0.8cm] \io.north) {};
}

% ---------------- Photonic tiles (Green with Honeycomb) ----------------
\def\hexSize{0.16} 
\foreach \x/\name in {-4.5/T1, 0/T2, 4.5/Tn}{
    \coordinate (A\name) at (\x-1.6,2.8);
    \coordinate (B\name) at (\x+1.6,4.2);
    
    % Background fill
     \fill[green!50!black] (A\name) rectangle (B\name);
    
    % Honeycomb pattern
    \begin{scope}
      \clip[rounded corners=4pt] (A\name) rectangle (B\name);
      \foreach \i in {-15,...,15}{ 
        \foreach \j in {-5,...,5}{
          \pgfmathsetmacro{\hx}{\hexSize*(3*\i/2)}
          \pgfmathsetmacro{\hy}{\hexSize*(sqrt(3)*(\j + \i/2))}
          \begin{scope}[shift={(\x+\hx,3.5+\hy)}] 
            \draw[green!8, opacity=0.45, line width=0.6pt]
              (0:\hexSize)--(60:\hexSize)--(120:\hexSize)--(180:\hexSize)--(240:\hexSize)--(300:\hexSize)--cycle;
          \end{scope}
        }
      }
    \end{scope}
    
    % Border
    \draw[imgDark, line width=1.8pt, rounded corners=4pt] (A\name) rectangle (B\name);
    
    % Ports
    \node[port] (p\name1) at (\x-0.8,2.8) {};
    \node[port] (p\name2) at (\x,2.8) {};
    \node[port] (p\name3) at (\x+0.8,2.8) {};
}

% Labels
% \node[font=\fontfamily{qhv}\selectfont\bfseries\Large, color=imgDark] at (0,4.6) {Photonic Interconnect Mesh};
\node[font=\sffamily\bfseries\huge, color=imgDark] at (2.25,3.5) {$\cdots$};
\node[font=\sffamily\bfseries\huge, color=imgDark] at (-2.25,3.5) {$\cdots$};

% ---------------- Links ----------------
\draw[link] (io1p1) -- (pT11); \draw[link] (io1p2) -- (pT21); \draw[link] (io1p3) -- (pTn1);
\draw[link] (io2p1) -- (pT12); \draw[link] (io2p2) -- (pT22); \draw[link] (io2p3) -- (pTn2);
\draw[link] (io3p1) -- (pT13); \draw[link] (io3p2) -- (pT23); \draw[link] (io3p3) -- (pTn3);

\end{tikzpicture}
}
\caption{GPUs with on-chip optical I/O (with one or more transceivers) connect to a photonic interconnect that establishes direct optical paths between them.}
\vspace{-6mm}
\end{figure}

\myitem{GPU forwarding:}
We assume that GPUs are equipped with an in-built router, similar to those used in Google's TPUs~\cite{10.1145/3579371.3589350}, and support cut-through forwarding at intermediate nodes.
In particular, GPUs can begin forwarding data before the entire message has been received.

\myitem{Communication steps:}
Throughout this paper, we adopt the standard notion of \emph{steps} used in prior work on collective algorithms~\cite{295653,https://doi.org/10.1002/cpe.1206,10.1145/2686882,doi:10.1177/1094342005051521,10.1145/3712285.3759835,10.1007/978-3-540-24685-5-1}.
A step denotes a communication phase during which each GPU exchanges data with a predetermined set of peers according to the collective algorithm.
Communication within a step may involve multi-hop forwarding through intermediate GPUs, consistent with prior works~\cite{10.1145/3579371.3589350,295653}.

\subsection{Example Walkthrough}
\label{sec:example}

Recursive doubling~\cite{10.1007/978-3-540-24685-5-1,295653,10.1145/2686882} is a bandwidth-optimal algorithm for AllReduce and completes in a logarithmic number of steps. For a network with $n$ nodes, recursive doubling proceeds in $\log_2(n)$ steps. In step $i$ (starting from $i = 1$), node $j$ communicates with node $(j + 2^{i-1}) \bmod n$, so the communication distance doubles at each step under a static topology.\footnote{Throughout this paper, we follow the cyclic variant of recursive doubling~\cite{kolmakov2020generalization}, which retains the same asymptotic properties as the pairwise-exchange formulation.} We first illustrate recursive doubling on a static one-dimensional interconnect, where a single ring topology supports communication across all steps (Figures~\ref{fig:step1-g1},~\ref{fig:step2-g1}, and~\ref{fig:step3-g1}). As the algorithm progresses, multiple communication pairs overlap on shared links in later steps.
In particular, congestion increases to $2$ in step~2 and to $4$ in step~3.
Higher congestion reduces the effective bandwidth available to each flow, increasing per-step transfer time and, consequently, the overall collective completion time.

\begin{table}[h]
  \begin{center}
    \begin{tabular}{|c|c|c|c|}
      \hline
      \textbf{Schedule}                                               & \textbf{Congestion} & \textbf{Propagation} & \textbf{Reconf.}
      \\
      \hline
      $G_1, G_1, G_1$ & $4$x & $4$x & $0$ \\
      $G_1, G_2, G_2$ & $2$x & $2$x & $1$x \\
      $G_1, G_2, G_3$ & $1$x & $1$x & $2$x \\
      \hline
    \end{tabular}
  \end{center}
  % \vspace{-5mm}
  \caption{Different sequence of topologies lead to significantly different congestion, propagation, and reconfiguration delays, that affect collective completion times.}
  \label{table:reconf-combinations}
  \vspace{-4mm}
\end{table}

We can reconfigure the topology before step~2 to establish direct optical paths between the GPUs that communicate in this step (Figure~\ref{fig:step2-g2}). At the final step, step~3, the interconnect faces a choice.
It can either remain in the same topology, which results in congestion of~$2$ (Figure~\ref{fig:step3-g2}), or reconfigure again to further reduce congestion (Figure~\ref{fig:step3-g3}).
The resulting reconfiguration options (Table~\ref{table:reconf-combinations}) highlight a fundamental trade-off: reducing congestion comes at the cost of incurring reconfiguration delay.

Congestion is not the only factor that increases completion time.
As communication distance grows, propagation delay also increases, further amplifying the impact of static or poorly chosen topologies (Table~\ref{fig:intro}).
Birkhoff--von Neumann schedules (e.g., $G_1$, $G_2$, $G_3$) that reconfigure before every communication step eliminate congestion and achieve the lowest completion time when reconfiguration delays are small.
As reconfiguration delays increase, static topologies become optimal, avoiding reconfiguration overhead at the cost of higher congestion and longer paths. Between these two extremes lies a broad design space in which carefully chosen, reconfiguration-aware schedules outperform both static interconnects and reconfigure-every-step baselines.
Our goal is to expose this spectrum and provide a systematic framework for reasoning about \emph{when} and \emph{how} reconfiguration should be used to minimize collective completion time.

\subsection{Modeling Collective Completion Time}
\label{sec:cct-model}

To reason about the impact of topology reconfigurations on collective completion time, we revisit the classic $\alpha$--$\beta$ cost model and extend it to explicitly account for network-level effects, such as propagation delay, congestion, and reconfiguration delay.
We model a collective communication algorithm that runs across $n$ GPUs as a sequence of $s$ communication steps.
In each step $i$, the collective exchanges a fixed amount of data, $m_i$, between pairs of GPUs. We use a communication matrix $\mathcal{M}_i$ to represent this.
An entry $\mathcal{M}_i(j,k)=1$ indicates GPU $j$ sends data to GPU $k$ during step $i$ (all other entries are zero). We describe a collective algorithm in this notation as a sequence $\langle \mathcal{M}_1, \mathcal{M}_2, \ldots, \mathcal{M}_s \rangle$ together with the data volumes $\langle m_1, m_2, \ldots, m_s \rangle$ asssociated with each step.

We use the \emph{aggregate demand matrix} $\mathcal{M}$ to capture the total communication across all steps, where each entry $\mathcal{M}(j,k)$ denotes the total volume of data sent from GPU $j$ to GPU $k$ over the entire collective. We sum the stepwise communication matrices weighted by their corresponding data volumes to compute this matrix:
\begin{equation}\label{eq:bvn-representation}
\mathcal{M}
= m_1 \cdot \mathcal{M}_1 + m_2 \cdot \mathcal{M}_2 + \ldots + m_s \cdot \mathcal{M}_s .
\end{equation}

\myitem{Point-to-point communication model:} Most prior work designs collective communication algorithms under the point-to-point communication model
~\cite{doi:10.1177/1094342005051521,https://doi.org/10.1002/cpe.1206,295653,10.1007/978-3-540-24685-5-1}, where in each step every node sends to and receives from at most one other node. As a result, the communication matrix $\mathcal{M}_i$ in each step is a permutation matrix. Many well-known algorithms, including ring, recursive doubling~\cite{10.1007/978-3-540-24685-5-1}, and swing~\cite{295653}\footnote{The Swing algorithm also admits a multi-port variant, which departs from the strict point-to-point model and captured by our one-to-many model.}, follow this abstraction. Under this model, Equation~\ref{eq:bvn-representation} corresponds by construction to a \emph{Birkhoff--von Neumann (BvN) decomposition} of $\mathcal{M}$, namely a convex combination
\footnote{We use the term combination loosely here. Specifically, $\sum_i m_i = m$, where~$m$ is the total data volume. Normalizing by $m$ yields a strict convex combination.} 
of permutation matrices. From this perspective, the steps of the collective algorithm correspond directly to the matchings in the decomposition, where each $m_i$ represents the data volume the algorithm transferrred during step $i$. 

\begin{graybox}
\begin{observation}[Relevance of BvN Decompositions]
Collective communication algorithms that proceed via a sequence of matchings naturally induce a BvN decomposition of their aggregate demand matrix.
\end{observation}
\end{graybox}

The converse does not hold. Not every BvN decomposition corresponds to a valid collective algorithm. BvN decompositions fail to capture the \emph{temporal structure} inherent in collective communication. In practical algorithms, the ordering of communication steps matters, and we cannot arbitrarily rearrange steps.
The data exchanged in step~$i$ is often generated as a consequence of the computation or communication performed in step $i-1$, which induces a strict sequence of dependencies.

We see a key limitation of aggregated demand matrices through these temporal and data-flow constraints. 
 While such matrices are useful in demand-aware network design~\cite{10.1145/2934872.2934911,10.1145/2716281.2836126,10.1145/2486001.2486007}, they implicitly assume that all traffic between source-destination pairs is simultaneously available. But this assumption does not hold for many collectives because communication steps must follow a fixed temporal order and we cannot assume data is always available to send at any given point in time. This is why we need to reason beyond static demand matrices and BvN decompositions when we design interconnects for collective communication.

\myitem{One-to-many communication model:}
Recent work proposes multi-port AllReduce algorithms~\cite{295653,10.1145/2686882} to better utilize network bandwidth when nodes are equipped with multiple links, such as in multi-dimensional Torus networks~\cite{10.1145/3579371.3589350}. Under this model, the communication matrix $\mathcal{M}_i$ is not necessarily a permutation matrix and is often composed of multiple permutations, for example one per dimension. We do not further decompose these matrices, and instead focus on the dependencies between successive matrices $\mathcal{M}_i$ and $\mathcal{M}_{i+1}$. We assume that each port participates in at most one point-to-point communication per step: the number of ports upper bounds the sum of each row and column of each $\mathcal{M}_i$. We can still use Equation~\ref{eq:bvn-representation} to express the communication model with this generalization~---~we extend each $\mathcal{M}_i$ to the one-to-many setting.

\myitem{All-to-All communication model:}
In many All-to-All implementations, multiple send and receive operations are grouped into a single logical step. Here, we can represent the communication with a single bulk demand, through an aggregate demand matrix. We can use a BvN decomposition to decompose the aggregate matrix and represent the communication via Equation~\ref{eq:bvn-representation} as a sequence of $s$ point-to-point steps. We sort the coefficients of the resulting decomposition such that $m_i$ is the $i^{th}$ largest coefficient, and $\mathcal{M}_i$ the corresponding point-to-point communication performed in the $i^{th}$ step.

The matrix decompositions induced by collective algorithms, as we show next, reveal a useful connection to both network throughput and the classic $\alpha$--$\beta$ cost model.

Consider a graph $G_i = (V, E_i)$, where $V$ is the set of $n\ $ GPUs and $E_i$ represents the photonic links between them during step $i$ of the collective. We can express the total completion time $t_c(1, s)$ of the collective communication algorithm from step $1$ through step $s$ (inclusive) as:
\begin{equation}
\begin{aligned}
t_c(1,s) ={}& DCT(m_1 \cdot \mathcal{M}_1,\ G_1)
            + DCT(m_2 \cdot \mathcal{M}_2,\ G_2) \\
           &\quad + \ldots
            + DCT(m_s \cdot \mathcal{M}_s,\ G_s) + \sum_{i=1}^{s} \underbrace{\alpha_r \cdot \mathbb{I}(E_i \not = E_{i-1})}_{\text{reconfiguration delay}}
\end{aligned}
\end{equation}
where $DCT(m_i \cdot \mathcal{M}_i,\ G)$ denotes the \emph{demand completion time} of step $i$, corresponding to a data volume $m_i$ and communication pattern $\mathcal{M}_i$, served by the underlying topology $G_i$. Each step incurs additional $\alpha_r$ delay if the topology differs from the previous step, representing a reconfiguration event.

$DCT(m_i \cdot \mathcal{M}_i,\ G_i)$ depends on the structure and capacity of the underlying graph $G_i$. We define the \emph{maximum concurrent flow} $\theta(G_i, \mathcal{M}_i)$ as the largest fraction of the communication matrix $\mathcal{M}_i$\footnote{Here, we consider that $\mathcal{M}_i$ is scaled proportionally to the node capacity.} the network can route simultaneously without exceeding any link capacities. $\theta(G_i, \mathcal{M}_i)$ is the achievable throughput for that step's communication pattern. This implies we can write demand completion time as:
\[
DCT(m_i \cdot \mathcal{M}_i,\ G_i) = \frac{m_i}{c} \cdot \frac{1}{\theta(G_i, \mathcal{M}_i)},
\vspace{-1mm}
\]
where $c$ is the link capacity. Here, $\frac{m_i}{c}$ represents the ideal transmission time assuming full throughput, while the factor $\frac{1}{\theta(G_i, \mathcal{M}_i)}$ accounts for congestion. By definition of the maximum concurrent flow, the effective capacity available for this communication is $c \cdot \theta(G_i, \mathcal{M}_i)$. Therefore, the actual transmission time scales inversely with the achievable throughput.

Each communication step $i$ incurs a fixed overhead $\alpha$, which captures startup latencies such as data preparation. In each step, we also have to account for the latency of the longest route that went through the most congested links in each step.
This latency is given by  $\delta \cdot \ell_i(G_i)$, where $\delta$ is the per-link propagation delay and $\ell_i(G_i)$ is the length of that longest path. The latency term is often neglected and absorbed into the constant $\alpha$. 
If the network offers capacity $c$ per node, we define $\beta = \frac{1}{c}$. We can write demand completion time for step-$i$ as:

\begin{align}\label{eq:dct}
DCT(m_i \mathcal{M}_i,\ G_i) =  
\underbrace{\alpha + \delta \cdot \ell_i(G_i)}_{\text{latency\ factor}} + 
\overbrace{\beta}^{\substack{bandwidth\\factor}}
 m_i  
\underbrace{\frac{1}{{\theta(G_i, \mathcal{M}_i)}}}_{\substack{congestion\\factor}}
\end{align}
We can now express the total completion time of the collective for all steps $s$ as:
\begin{align}
t_c(1,s) 
% &= \sum_{i=1}^{s} DCT(m_i \cdot \mathcal{M}_i,\ G_i) + \sum_{i=1}^{s} \alpha_r \cdot \mathbb{I}(E_i \not = E_{i-1}) \nonumber \\
& = \sum_{i=1}^{s} \left( \alpha + \delta\cdot \ell_i(G_i) +  \frac{\beta \cdot m_i}{\theta(G_i, \mathcal{M}_i)} + \alpha_r \cdot \mathbb{I}(E_i \not = E_{i-1}) \right)
% &= s \cdot \alpha + \sum_{i=1}^s \delta\cdot \ell_i(G_i) + \beta \cdot \sum_{i=1}^{s} m_i \cdot \frac{1}{\theta(G_i, \mathcal{M}_i)} + \sum_{i=1}^{s} \alpha_r \cdot \mathbb{I}(E_i \not = E_{i-1})
\end{align}
\noindent 
\begin{graybox}
\begin{observation}[Collective Completion Time as $\alpha$--$\beta$ Cost]
The classic $\alpha$--$\beta$ cost model emerges naturally when collective completion time is expressed in terms of network latency, bandwidth, propagation delay, and congestion, where congestion is captured by the inverse of maximum concurrent flow.
This perspective grounds the model in network throughput and exposes its dependence on both the interconnect topology and the staged structure of the collective.
\end{observation}
\end{graybox}

While the $\alpha$--$\beta$ model is widely used in practice, network throughput, propagation delay, and congestion are rarely made explicit in its formulation. A few exceptions relate congestion to communication distance or to the number of messages traversing a link in structured topologies~\cite{https://doi.org/10.1002/cpe.1206,10.1145/2686882,295653}, but these approaches are typically limited to specific communication patterns or architectures and often assume unsplittable flows. TE-CCL~\cite{10.1145/3651890.3672249} explored the relationship between the $\alpha$--$\beta$ model and multi-commodity flow in the context of collective algorithm synthesis by mapping collective communication patterns to demand matrices and interpreting the cost model through that formulation. In a similar vein, our approach explicitly links the $\alpha$--$\beta$ model to \emph{network throughput} via maximum concurrent flow in the context of optical interconnect configurations. This yields a more comprehensive understanding of performance that accounts for communication structure, congestion, propagation delay, and network topology. Our formulation applies to arbitrary topologies, making it broadly applicable beyond structured or hierarchical networks.

\subsection{An Optimization Opportunity}
\label{sec:opportunity}

We observe, the completion time of a collective communication algorithm is tied to the path lengths\footnote{Note that we assume GPUs are equipped with an in-built router (see \S\ref{sec:architecture}).}, congestion, and throughput of the underlying topology in each step. This is especially helpful in the context of circuit switching photonic interconnects: we can reduce congestion and path lengths to~$1$ (i.e., full throughput) and establish direct, high-bandwidth, optical paths that exactly match the communication pattern $\mathcal{M}_i$ for each step $i$.

But to realize these direct paths we need to reconfigure the interconnect which results in a reconfiguration delay $\alpha_r$. This creates a clear trade-off: we can reduce congestion if we reconfigure and improve throughput but this will increase latency; or we can maintain a static topology and avoid the reconfiguration cost but then we may suffer from higher congestion.

This tension opens up an opportunity for optimization: given a collective communication schedule, how should we schedule interconnect reconfigurations to minimize the total completion time for any given collective? For example, one might choose to maintain a static topology to avoid reconfiguration overhead but pay persistent congestion costs, or reconfigure before every step to eliminate it while incurring the maximum reconfiguration penalty. An effective circuit switching schedule must strike a balance, where we reconfigure only in steps when the throughput gain outweighs the cost.

\noindent This paper focuses on the \emph{topology synthesis} problem:

\myitem{Input:}
We are given a predefined collective communication algorithm i.e., a communication schedule. The schedule specifies, the amount of data each source-destination pair exchanges in each step of the collective. The input also includes topology constraints, such as link bandwidth, node degree (number of links), and fixed propagation latencies.

\myitem{Output:}
The output is a topology reconfiguration schedule that specifies a network topology for each communication step of the collective. When the topology differs between consecutive steps we incur a reconfiguration delay.

A related but orthogonal line of work studies the \emph{collective synthesis} problem, in which the input is a fixed topology and the output is a communication schedule~\cite{10.1145/3651890.3672249,285084,10.1145/3437801.3441620}. This setting is the converse of \emph{topology synthesis} problem we consider in this paper.

\section{\name}
% : Synthesizing Photonic Switching Schedules
% How we synthesize schedules.
\label{sec:harvest}
At its core, topology synthesis consists of two tightly coupled components. 
First, given a contiguous range of collective steps $a$ through $b$, we must decide \emph{how} to reconfigure the interconnect i.e., find a static topology that minimizes the completion time of those steps when no reconfiguration is allowed within the range. Second, we must partition the full sequence of collective steps into such ranges to determine~\emph{when} to reconfigure.
This partitioning induces a recurrence that jointly determines both the reconfiguration events and the corresponding topologies.

%We next formalize the topology optimization subproblem (\S\ref{sec:subproblem}) and then show how we can combine it with dynamic programming to synthesize globally optimal reconfiguration schedules (\S\ref{sec:recurrence}).

\subsection{Subproblem}\label{sec:subproblem}

Given a collective communication algorithm with $s$ steps, where each step $i$ is characterized by a communication pattern $\mathcal{M}_i$ and data volume $m_i$, our subproblem is to find an optimal topology $G_{a,b}$ that minimizes the completion time of steps $a$ through $b$ (inclusive), without any reconfigurations during this interval.
\vspace{-3mm}
\begin{equation}\label{eq:g-subproblem}
G_{a,b} = \arg\min_{G \in \mathcal{G}} \sum_{i=a}^{b} DCT(m_i \cdot \mathcal{M}_i,\ G) 
\vspace{-3mm}
\end{equation}
\begin{equation}\label{eq:tc-subproblem}
t_c(a,b) = \sum_{i=a}^{b} DCT(m_i \cdot \mathcal{M}_i,\ G_{a,b}) 
\end{equation}
\vspace{-2mm}

\noindent
Here, $\mathcal{G}$ denotes the set of all feasible topologies that satisfy the per-node degree constraints, and $t_c(a,b)$ denotes the completion time for steps $a$ through $b$. We use Equation~\ref{eq:dct} to compute the demand completion time $DCT(m_i \cdot \mathcal{M}_i,\ G_{a,b})$.

We solve this subproblem using a Mixed-Integer Second-Order Conic Program (MISOCP). We introduce integer decision variables $x_{u,v}$ that denote the number of directed edges from node $u$ to node $v$. These variables are subject to degree constraints $\sum_{u} x_{u,v} \le d$ and $\sum_{v} x_{u,v} \le d$, and ensure each node has degree at most $d$. The objective is to minimize the total demand completion time for the sequence of communication matrices $\langle \mathcal{M}_a, \ldots, \mathcal{M}_b \rangle$.

The constraints follow a standard maximum concurrent flow formulation, with the key distinction that edge capacities are decision variables rather than fixed inputs. The objective minimizes a weighted sum of the reciprocals of the concurrent flow values across steps (Equation~\ref{eq:dct}). This yields a second-order conic optimization problem with integer variables. Unlike a classical concurrent flow formulation, the presence of sequential dependencies across steps induces the conic structure. We present the complete formulation in \S~\ref{app:misocp}.

% Let $g = \lvert \mathcal{G} \rvert$ denote the number of candidate topologies. For $1$-D topologies, $g = n!$. Note that $g$ is treated as a constant in our formulation. Without additional structural insights into the communication pattern between steps $a$ through $b$, the complexity of finding the optimal topology remains $O(n!)$. As we show later in our evaluations, restricting the search space to $O(n)$ already provides substantial speedups in completion times in $1$-D interconnects where each GPU is equipped with a single I/O interface.
% We consider the following restricted set of topologies in our search space:
% \begin{itemize}[label=\small{\textcolor{black}{$\blacksquare$}}, leftmargin=*]
% \item The $s$ direct-connect topologies that match the communication pattern $\mathcal{M}_i$ at each step $i \in \{1,\ldots,s\}$. These are optimal for subproblems whose interval consists of a single step.
% \item Coprime shifted rings. Each node $u$ connects to node $(u+k) \bmod n$, where
% \[
% k \in \{1,2,\ldots,n-1 : \gcd(k,n)=1\}.
% \]
% The number of such shifts is Euler's totient function $\varphi(n)$~\cite{euler1763theoremata}, which satisfies $\varphi(n) \le n-1$ with equality when $n$ is prime. Each choice of $k$ yields a connected directed ring, which is useful for subproblems spanning many steps where connectivity is required.
% \end{itemize}
% These choices reduce the effective search space to $O\big(s+n\big)$. In the next section, we show that the search space can be reduced to $O(1)$ for recursive doubling AllReduce, yielding optimal schedules.

\begin{algorithm}[!h]
\DontPrintSemicolon
\SetAlgoLined
\KwIn{Setup latency $\alpha$, bandwidth $b$, $\beta=\tfrac{1}{b}$, propagation delay $\delta$, reconfiguration delay~$\alpha_r$, sequence of steps and corresponding data volumes $\langle m_1 \cdot \mathcal{M}_1,\ m_2 \cdot \mathcal{M}_2 ,\ \ldots ,\ m_s \cdot \mathcal{M}_s\rangle$.
% search space for topologies $\mathcal{G}$. 
}
\KwOut{Reconfiguration schedule}

% $s \gets \lceil \log_2 n \rceil$\;

\SetKwFunction{compTime}{\textbf{\textsc{\textcolor{myred}{CompletionTime}}}}
% \SetKwFunction{Comm}{\textbf{\textsc{\textcolor{myred}{TransmissionDelay}}}}
\SetKwProg{Fn}{Function}{:}{}

\Fn{\compTime{$a,b$}}{
  
  Find $G_{a,b}$ 
  \Comment{\textcolor{gray}{Equation~\ref{eq:g-subproblem}}}

  Find $t_c(a,\ b)$ \Comment{\textcolor{gray}{Equation~\ref{eq:tc-subproblem}}}

  \Return{$\left(G_{a,b},\ t_c(a,\ b)\right)$}
}

\textcolor{gray}{\textit{To find the optimal reconfigurations for a given number of reconfigurations $k$.}}

\SetKwFunction{Compute}{\textbf{\textsc{\textcolor{myred}{SynthesizeSchedule}}}}
\Fn{\Compute{$s,k$}}{
  Initialize $\mathrm{DP}[a][t]\gets \infty$, $\mathrm{next}[a][t]\gets \varnothing$, $\mathrm{topo}[a][t]\gets \varnothing$ for $a\in\{1,\dots,s+2\}, t\in\{0,\dots,k\}$\;
  \For{$a \gets 1$ \KwTo $s$}{
    $\mathrm{DP}[a][0] \gets \FuncSty{\textbf{\textsc{\textcolor{myred}{CompletionTime}}}}(a,s)$
  }
 \For{$t \gets 0$ \KwTo $k$}{
    $\mathrm{DP}[s+1][k] \gets 0$
  }
  \For{$t \gets 1$ \KwTo $k$}{
    \For{$a \gets 1$ \KwTo $s$}{
      $\mathrm{best}\gets\infty$;\ $\mathrm{arg}\gets\varnothing$;\ $\mathrm{graph}\gets \varnothing$ \;
      \For{$b \gets a+1$ \KwTo $s+1$}{
        $(G, v) \gets \FuncSty{\textbf{\textsc{\textcolor{myred}{CompletionTime}}}}(a,b-1) + \mathrm{DP}[b][t-1]$ \Comment{\textcolor{gray}{Recurrence (Lemma~\ref{lem:dprecurrence})}}

        \If{$v < \mathrm{best}$}{ $\mathrm{best}\gets v$;\ $\mathrm{arg}\gets b$;\ $\mathrm{graph}\gets G$ }
      }
      $\mathrm{DP}[a][t]\gets \mathrm{best}$;\quad $\mathrm{next}[a][t]\gets \mathrm{arg}$;\ 
      
      $\mathrm{topo}[a][t]\gets graph$\;
    }
  }
  $r \gets []$;\ $a \gets 1$;\ $t \gets k$\;
  \While{$t>0$}{
    $b \gets \mathrm{next}[a][t]$;\ \lIf{$b=\varnothing$}{\textbf{break}}
    $G \gets \mathrm{topo[a][t]}$

    \If{$b \le s$}{append $(G,b)$ to $r$}
    $a \gets b$;\ $t \gets t-1$\;
  }
  \Return{$\big(\mathrm{DP}[1][k], r\big)$} \Comment{\textcolor{gray}{(Completion time w/o reconfiguration delays, schedule)}}
}

\textcolor{gray}{\textit{To find the optimal reconfigurations across all number of reconfigurations $k$.}}

$\mathrm{bestCost}\gets\infty$;\ $\mathrm{bestR}\gets []$\;
\For{$k \gets 0$ \KwTo $s$}{
  $(C, R) \gets \Compute(s,k)$ \Comment{\textcolor{gray}{Schedule with exactly $k$ reconfigurations}}

  $C \gets C + k\cdot \alpha_r$ 
  % \Comment{\textcolor{gray}{Completion time including reconfiguration delays}}

  \If{$C < \mathrm{bestCost}$}{ $\mathrm{bestCost}\gets C$;\ $\mathrm{bestR}\gets R$ }
}

\Return{$\mathrm{bestR}$} \Comment{\textcolor{gray}{Optimal schedule (Theorem~\ref{thm:optimality})}}

\caption{\name}
\label{alg:algorithm}
\end{algorithm}

\subsection{Recurrence and Dynamic Programming}
\label{sec:recurrence}

At a high level, our dynamic programming approach partitions the collective communication algorithm's steps into contiguous intervals, separated by $k$ reconfiguration events. Central to our approach are three variables:
\begin{itemize}[leftmargin=*,label=\small{\textcolor{black}{$\blacksquare$}}]
 \item $\dpv[a][t]$ denotes the optimal completion time with $t$ reconfigurations for the sequence of steps starting at $a$, until the end. We do not include the reconfiguration delays at this point. Our dynamic program takes a fixed number of reconfigurations as input, and finds a corresponding optimal schedule. We later find the best number of reconfigurations (\S\ref{sec:synthesis}).
 \item $\nextv[a][t]$ stores the next step after $a$ at which a reconfiguration occurs in the optimal schedule.
 \item $\mathrm{topo}[a][t]$ stores the optimal topology for the steps~$a$ through $\nextv[a][t]$. 
\end{itemize}

\begin{Lemma}[Recurrence]
\label{lem:dprecurrence}
For any starting step $a\in\{1,\dots,s\}$ and number of reconfigurations $k \ge 1$, the optimal completion time is
\begin{equation}\label{eq:dpat}
\dpv[a][k] = \min_{\,a < b \le s+1}\ \big(t_c(a,\ b-1) + \dpv[b][k-1]\big)
\end{equation}
where $t_c(\cdot,\cdot)$ is the completion time for steps $a-b$, given by Eq~\ref{eq:tc-subproblem}.
\end{Lemma}
% \behnaz{Just a thought: given this, can you prove something about the optimality of Liangyu's schedule? there seems to be some very intuitive relationship between the way he derives the schedule and the lemma above --- just curious, not an actionable comment.}

\begin{proof}
We prove the claim by induction on the number of reconfigurations $k \ge 0$.
For the base case $k=0$, no reconfigurations are allowed and the schedule consists of a single contiguous interval covering steps $a$ through $s$.
The collective completion time is therefore $t_c(a,s)$ by definition, and hence $\dpv[a][0]=t_c(a,s)$.
Now assume the claim holds for all numbers of reconfigurations $t<k$ and for all starting steps.
Fix $k\ge 1$ and a starting step $a$.
Consider any feasible schedule with exactly $k$ reconfigurations over steps $a$ through $s$.
Let $b$ denote the step at which the first reconfiguration occurs, where $a<b\le s$.
This choice partitions the schedule into two segments: steps $a$ through $b-1$, executed without reconfiguration, and steps $b$ through $s$, executed with the remaining $k-1$ reconfigurations.
The completion time of the first segment is $t_c(a,b-1)$.
By the inductive hypothesis, the minimum completion time achievable for steps $b$ through $s$ with $k-1$ reconfigurations is $\dpv[b][k-1]$.
Therefore, the total completion time of any such schedule is at least
$t_c(a,b-1) + \dpv[b][k-1]$.
Minimizing over all valid choices of $b$ yields the recurrence defining $\dpv[a][k]$, which completes the proof.
% By Observation~\ref{obs:opt-range}, the optimal topology for $[a,b)$ is the direct-connect pattern of step $a$, yielding
% \[
% cost(a,b) = \delta \sum_{i=a}^{b-1}2^{i-a} + \beta m \sum_{i=a}^{b-1}\frac{2^{i-a}}{2^i}
% \]
Thus any $b$ yields total cost $t_c(a,b-1)+\dpv[b][k-1]$, and minimizing over $b$ gives Equation~\ref{eq:dpat}.
\end{proof}

The optimal schedule with exactly $k$ reconfigurations is $\dpv[1][k]$. Algorithm~\ref{alg:algorithm} unrolls Equation~\ref{eq:dpat} for $\dpv[1][k]$ within a for loop. At each partition, we record $\nextv[a][t]$, which then allows us to reconstruct the optimal schedule.

\subsection{Synthesizing Switching Schedules}\label{sec:synthesis}

Lemma~\ref{lem:dprecurrence} gives the optimal topology sequence for a fixed number of reconfigurations $k$, excluding reconfiguration delays. It remains to find the number of reconfigurations for which the reconfiguration schedule minimizes the overall completion time, including reconfiguration delays.
Taking the minimum over $k=0,\dots,\ s$ and adding the reconfiguration delay $k\cdot\alpha_r$ corresponding to $k$ reconfigurations, yields the delay-aware schedule that minimizes overall completion time.

\begin{theorem}[Optimality of the schedule]
\label{thm:optimality}
Fix the number of steps $s$. For any $k\in\{0,\dots,s\}$, the schedule reconstructed from $\nextv[\cdot][\cdot]$ that attains $\dpv[1][k]$ is optimal among all schedules with exactly $k$ reconfigurations. Moreover,
$
\argmin_{k\in\{0,\dots,s\}}\ \big(\dpv[1][k] + k\ \alpha_r\big)
$
is optimal among all schedules that account for reconfiguration delay.
\end{theorem}

\begin{proof}
By Lemma~\ref{lem:dprecurrence} and induction on $k$, $\dpv[a][k]$ equals the optimal completion time for steps $a$ through $s$ with exactly $k$ reconfigurations; in particular $\dpv[1][k]$ is optimal for the full instance, and $\nextv$ reconstructs an optimal switching schedule. Since reconfiguration delay is additive and depends only on $k$, minimizing $\dpv[1][k]+k\,\alpha_r$ over $k\in\{0,\dots,s\}$ yields the delay-aware optimum.
\end{proof}

% Our objective is to minimize the total completion time of the collective communication algorithm, which consists of four components: ($\delta\cdot \ell_i$) the propagation delay as function of path lengths, ($\alpha)$ the fixed latency factor, ($\alpha_r$) the total reconfiguration delay incurred by the interconnect, and ($\frac{1}{\theta}$) the congestion factor across all steps. The congestion and propagation delay depend on whether we choose to reconfigure the interconnect to match the communication pattern $\mathcal{M}_i$ or maintain the base topology $G$. The constraints ensure that $z_i$ correctly captures whether a reconfiguration occurs between steps, and all variables are binary.

% Overall, this formulation is an integer linear program ( ILP), which is NP-hard in the general case~\cite{karp2009reducibility}. Interestingly, our model has a special sequential structure: the variables $x_i$ (interconnect state) and $z_i$ (reconfiguration event) depend only on the previous step. This structure admits an efficient dynamic programming solution and is polynomial-time solvable due to the principle of optimality~\cite{bertsekas2012dynamic}.
Overall, our framework captures the fundamental trade-off between reconfiguration delay and congestion in adaptive photonic interconnects. It provides a systematic way to synthesize circuit switching schedules for collective communication, balancing the benefits of reconfiguration against its costs. Notably, the synthesis is aware of data volume in each step, reconfiguration delay, propagation delay, and the underlying network throughput.

\subsection{Discussion}
\label{sec:roadmap}
Our synthesis framework combines a dynamic program over step intervals with a topology optimization subproblem to compute the optimal reconfiguration. The dynamic program has polynomial complexity in the number of collective steps, while the dominant computational cost arises from solving the topology subproblem via a MISOCP. The overall complexity is $O(s^4 \cdot g)$, where $s$ denotes the number of communication steps and $g$ captures the cost of solving a single MISOCP instance.

In practice, the number of steps $s$ is modest for many widely used collective algorithms. For example, recursive doubling and Swing have $s = \Theta(\log_2 n)$ steps. We can synthesize the schedule offline and cache it to avoid computations at runtime when messages arrive. As a result, the synthesis cost does not lie on the critical path of collective execution. 
We discuss the practical aspects of the computation cost in~\S\ref{sec:evaluation}.

The primary scalability challenge lies in the topology optimization subproblem. We can solve our MISOCP for interconnect sizes of up to $64$ GPUS, but we need to reduce the effective topology search space further to scale to larger interconnects. This observation motivates us to exploit the structure in collective communication patterns and to restrict our attention to a small set of candidate topologies that are likely to be optimal over contiguous intervals of steps.

It is interesting to understand the precise conditions under which restricted topology classes suffice for arbitrary collectives in future work.
In the next section, we address the following questions, which guide the design of practical synthesis algorithms that balance optimality and efficiency:

\smallskip
\noindent
\textit{\textbf{(Q1)} To what extent can we reduce the topology search space without compromising optimality?}

\smallskip
\noindent
\textit{\textbf{(Q2)} Can we synthesize optimal schedules within polynomial-time complexity in the number of nodes?}

\section{Optimal Photonic Switching Schedules for Recursive Doubling AllReduce}
Building on our observations in \S\ref{sec:modeling} and the synthesis technique in \S\ref{sec:harvest}, our goal is to efficiently synthesize an \emph{optimal} schedule. Our design centers on the recursive doubling algorithm for AllReduce. We make two new observations about recursive doubling that, as we later show in this section, enable synthesis of schedules within polylogarithmic-time complexity.

% We first give a brief overview and present novel observations on recursive doubling algorithm (\S\ref{sec:overview}). We present \name (\S\ref{sec:oroborus}), outline the key properties of \name schedules (\S\ref{sec:properties}) and discuss its practicality (\S\ref{sec:practicality}).

\subsection{Observations on Recursive Doubling}
\label{sec:overview}

We make two simple yet powerful observations about the recursive doubling algorithm, which directly guide the synthesis. We consider cyclic version of recursive doubling~\cite{kolmakov2020generalization} that retains the same properties as the pairwise exchange version. A node $u$ communicates with node $u+2^{i-1}$, and transmits $\frac{m}{2^{i}}$ chunk size in step $i$, during the reduce-scatter phase. The communication pattern reverses in the AllGather phase. 

We first check if we need additional reconfigurations to preserve reachability when we establish direct links based on the communication step $i$. This, in turn, helps find topologies that can serve multiple steps without frequent reconfiguration.

\begin{graybox}
\begin{observation}[Connectivity]\label{obs:connectivity}
The topology that establishes direct links between GPUs according to the communication pattern of step $i$ also preserves connectivity for all subsequent steps $j \geq i$ in recursive doubling.
\end{observation}
\end{graybox}

In recursive doubling, step $i$ requires each node $a$ to communicate with $a+2^{i-1}$ (with steps indexed from $1$). Establishing these direct links does not break connectivity for any later step $j \geq i$. In step $j$, node $a$ must communicate with $a+2^{j-1}$, which lies at distance $2^{j-i}$ in this topology. 
This is because $a$ connects to $a+1$, which in turn connects through a chain of nodes $a+2, a+4, \ldots, a+2^{j-i}$, which ensures connectivity. The proof follows.

\begin{figure*}[!t]
     \centering
     % --- Row 1 ---
     \begin{subfigure}[b]{0.24\textwidth}
         \centering
         \includegraphics[width=\textwidth]{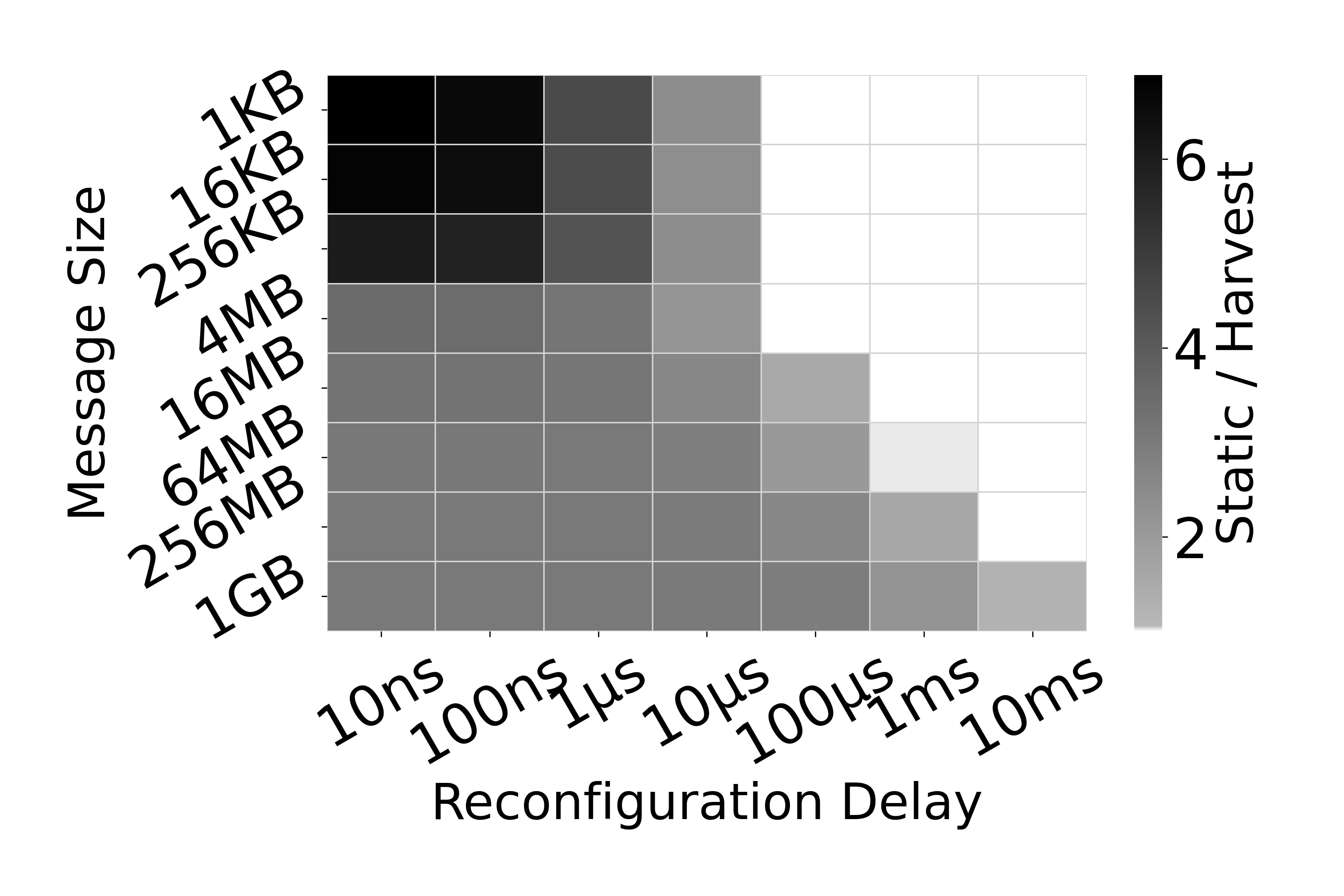}
         \subcaption{\small Recursive doubling}
         \label{fig:sim-hd-static}
     \end{subfigure}
     \hfill
     \begin{subfigure}[b]{0.24\textwidth}
         \centering
         \includegraphics[width=\textwidth]{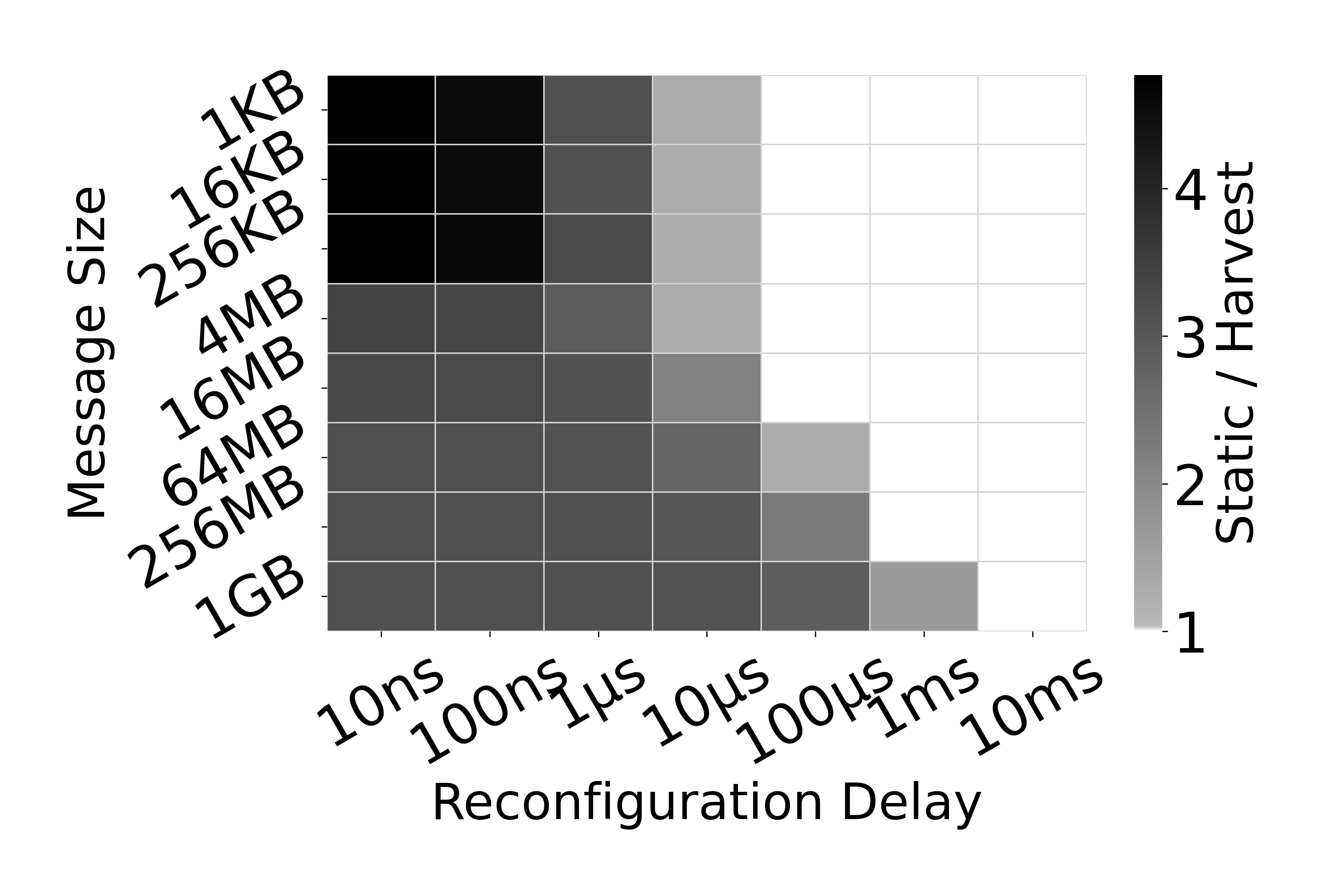}
         \subcaption{\small Swing}
         \label{fig:sim-swing-static}
     \end{subfigure}
     \hfill
     \begin{subfigure}[b]{0.24\textwidth}
         \centering
         \includegraphics[width=\textwidth]{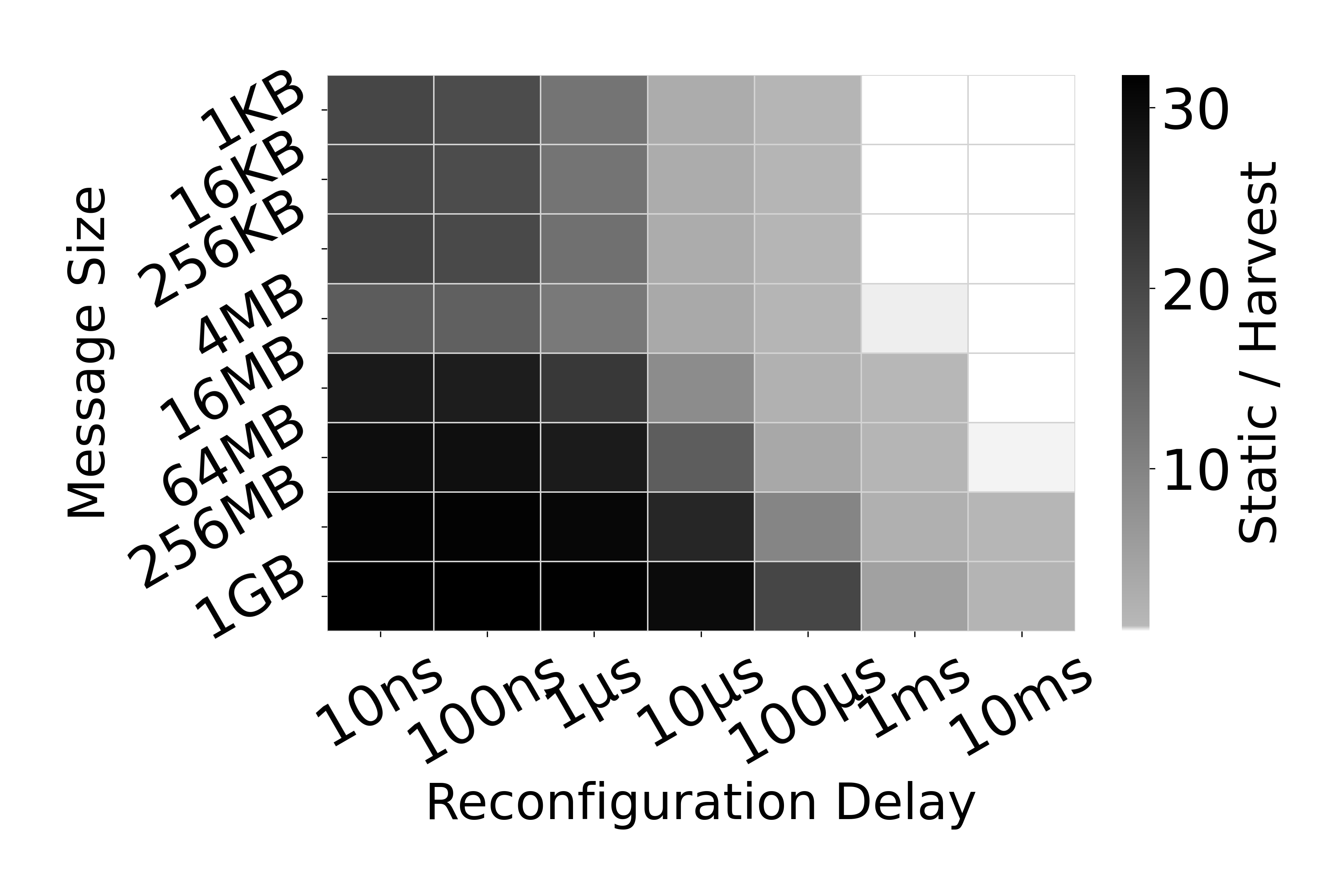}
         \subcaption{\small Direct All-to-All}
         \label{fig:sim-direct-static}
     \end{subfigure}
     \hfill
     \begin{subfigure}[b]{0.24\textwidth}
         \centering
         \includegraphics[width=\textwidth]{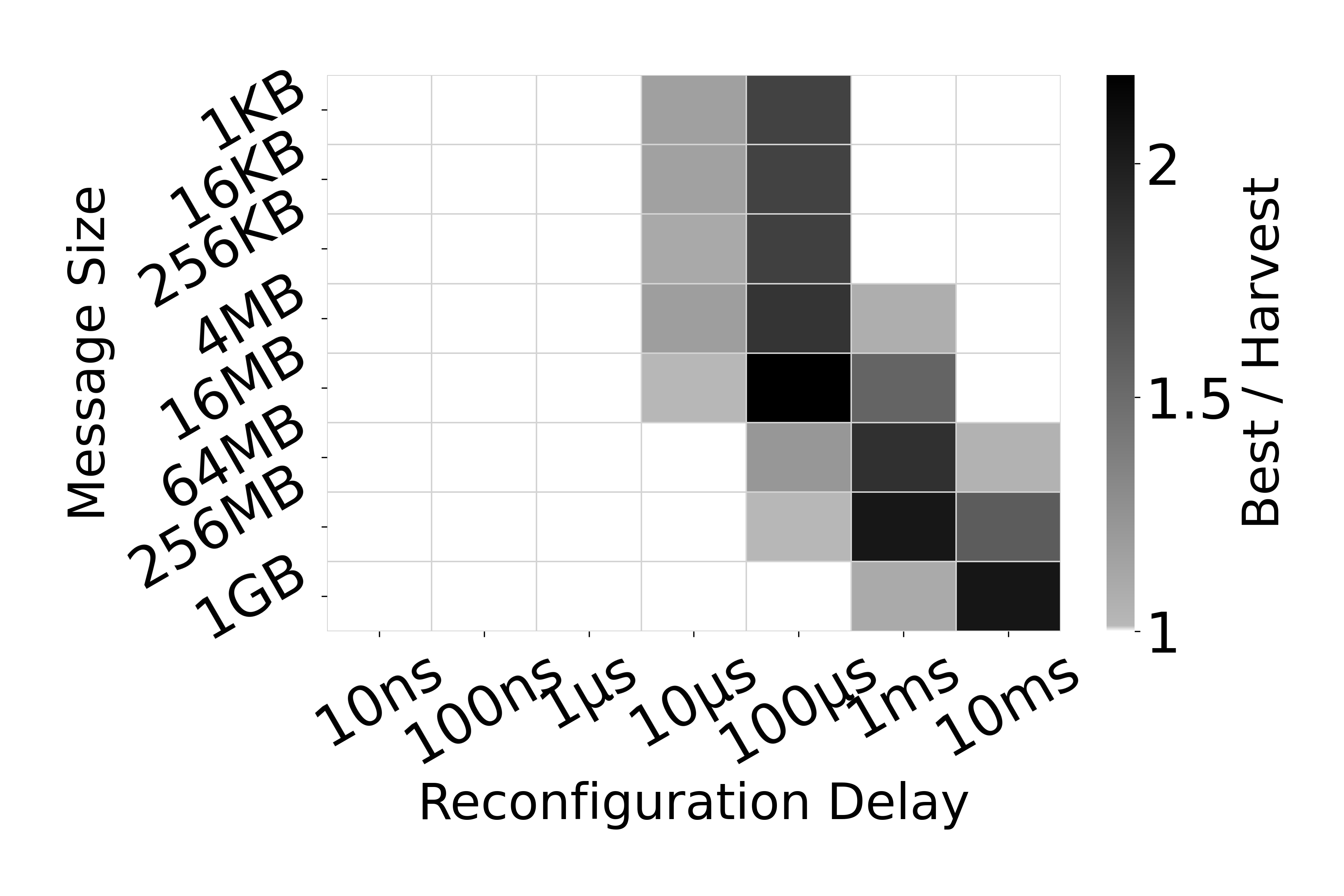}
         \subcaption{\small Direct All-to-All}
         \label{fig:sim-direct-best}
     \end{subfigure}

     \vspace{2pt} 

     % --- Row 2 ---
     \begin{subfigure}[b]{0.24\textwidth}
         \centering
         \includegraphics[width=\textwidth]{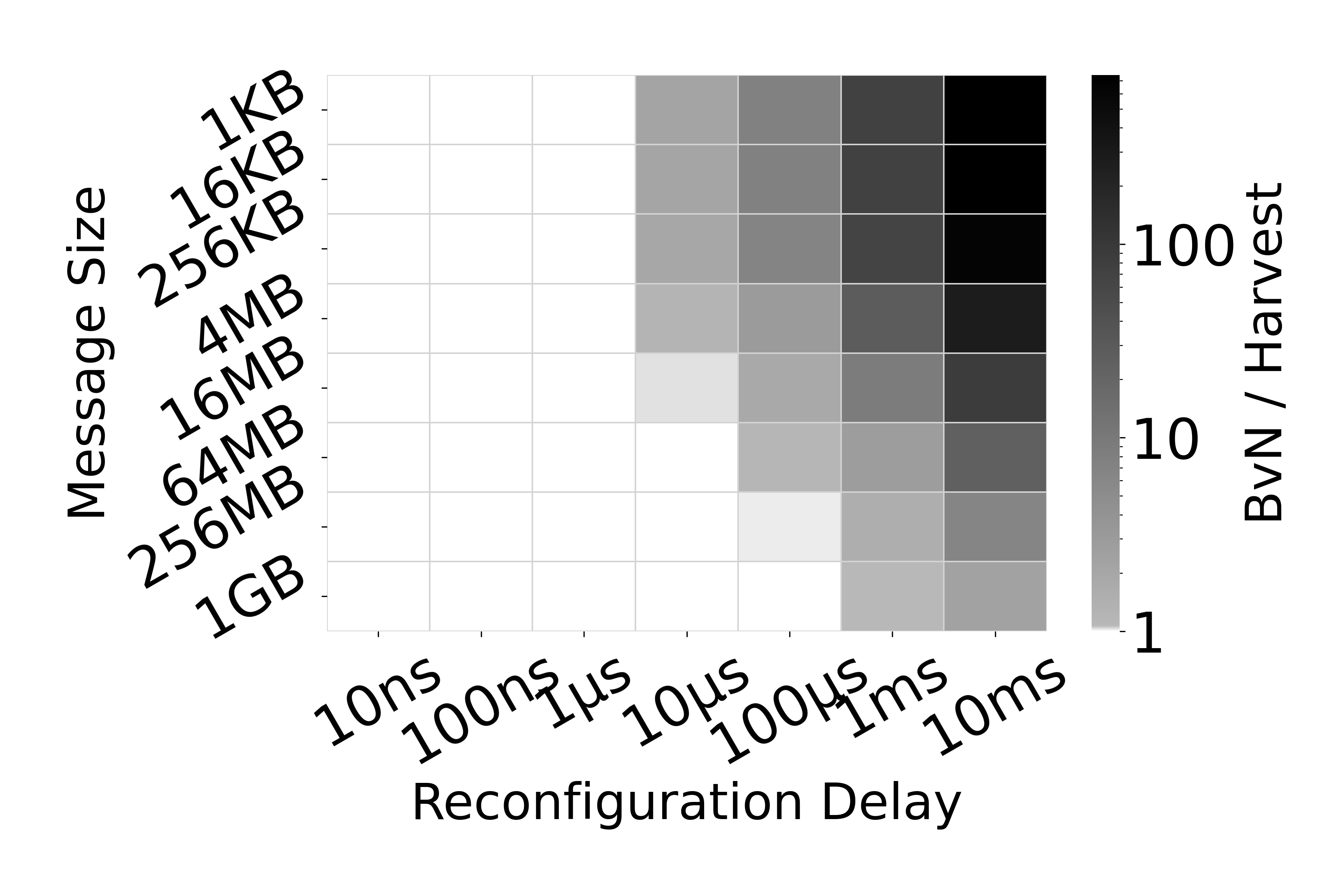}
         \subcaption{\small Recursive doubling}
         \label{fig:sim-hd-bvn}
     \end{subfigure}
     \hfill
     \begin{subfigure}[b]{0.24\textwidth}
         \centering
         \includegraphics[width=\textwidth]{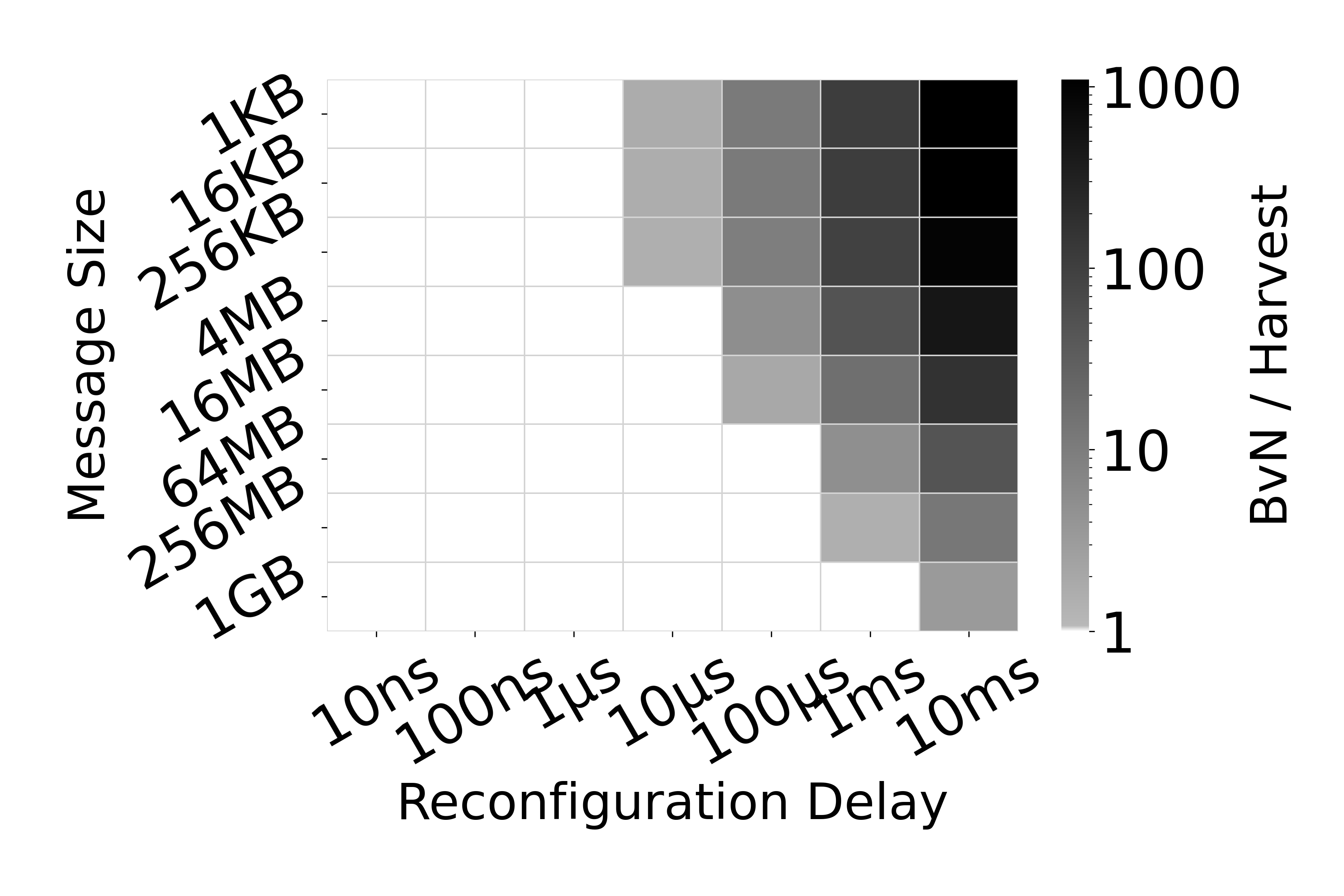}
         \subcaption{\small Swing}
         \label{fig:sim-swing-bvn}
     \end{subfigure}
     \hfill
     \begin{subfigure}[b]{0.24\textwidth}
         \centering
         \includegraphics[width=\textwidth]{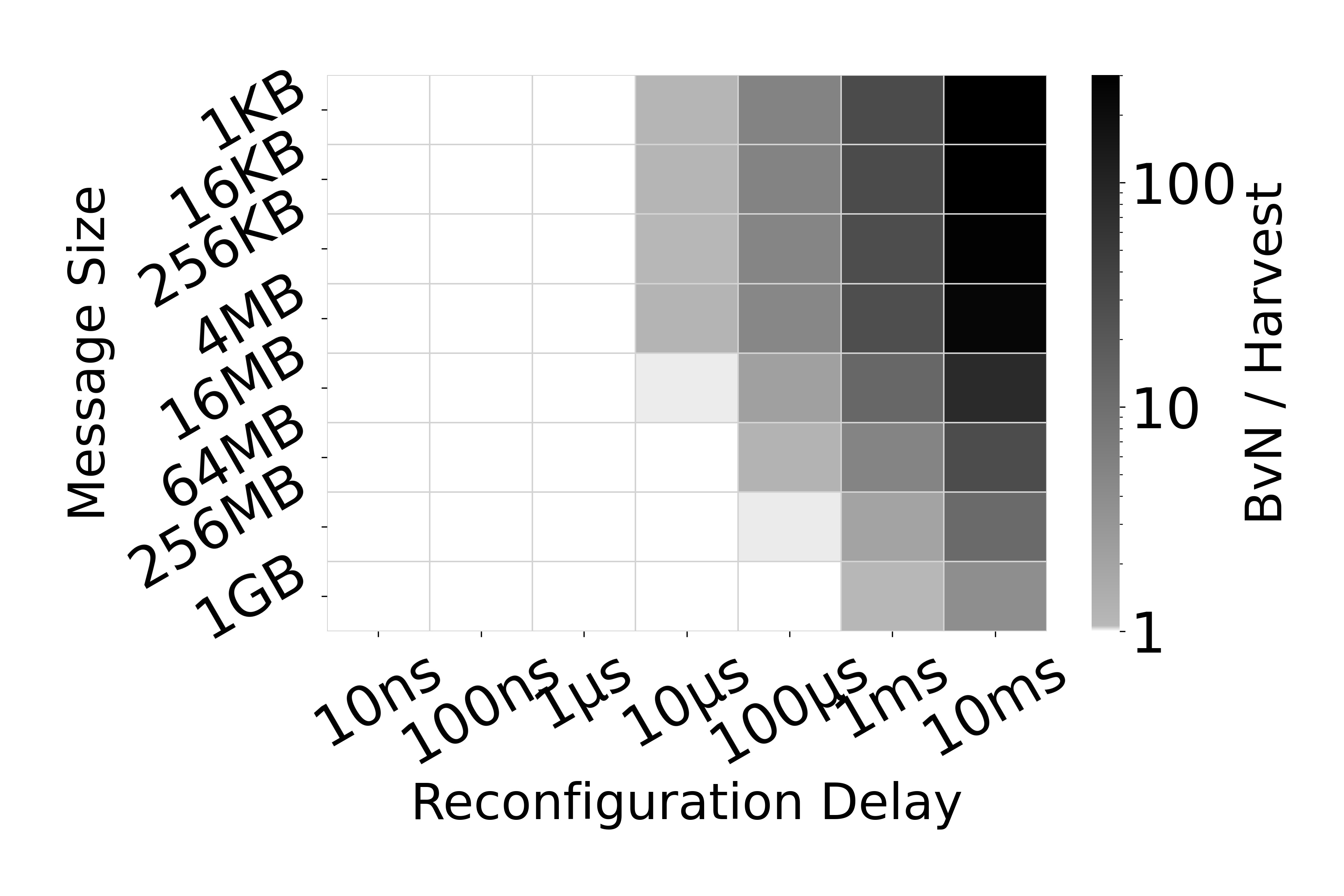}
         \subcaption{\small Direct All-to-All}
         \label{fig:sim-direct-bvn}
     \end{subfigure}
     \hfill
     \begin{subfigure}[b]{0.24\textwidth}
         \centering
         \includegraphics[width=\textwidth]{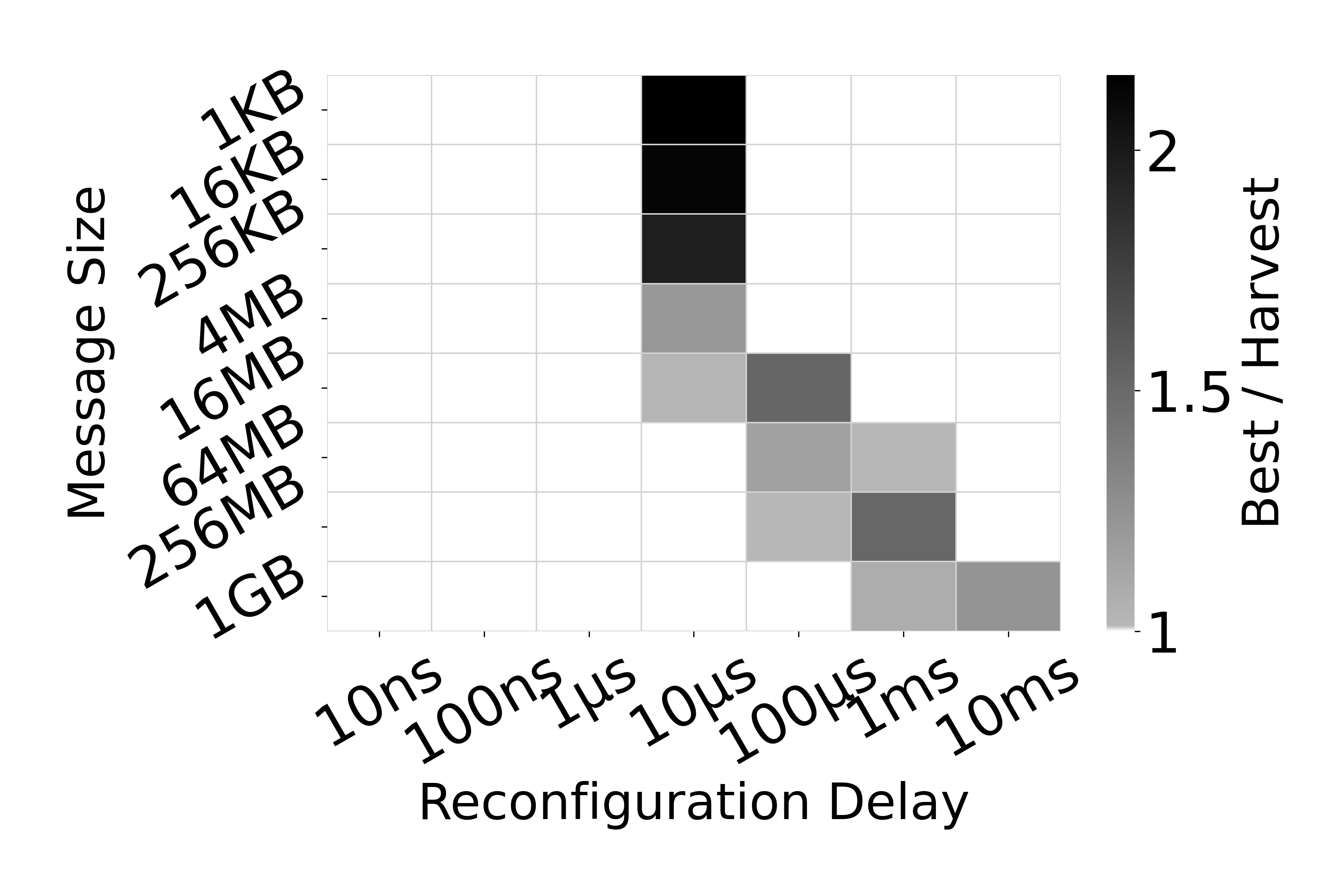}
         \subcaption{\small Recursive doubling}
         \label{fig:sim-hd-best}
     \end{subfigure}
     \caption{[Simulations] We show \name speeds up collective completion time compared to BvN-based schedules and static topologies with link bandwidth=800Gbps, $\alpha$=500ns, and $\delta$=500ns}
     %Heatmaps showing the speedup in collective completion times achieved by \name compared to BvN-based schedules and static topologies with link bandwidth=800Gbps, $\alpha$=500ns, and $\delta$=500ns}
     \label{fig:sim-results}
\end{figure*}

Next, we aim to find a single topology that minimizes the completion time for any sequence of steps $a$ through $b$ in recursive doubling. This enables us to restrict the search space to a specific class of topologies.  

\begin{graybox}
\begin{observation}[Optimal Topology]\label{obs:opt-range}
For any interval of steps $a$ through $b$ in recursive doubling, the topology that minimizes completion time is the one that establishes direct links between GPUs according to the communication pattern of step $a$.
\end{observation}
\end{graybox}

We can characterize the interval of steps $a$ through $b$ as: in step $a$, each node $u$ communicates with $u+2^{a-1}$ (with steps indexed from $1$). The  sequence of minimum path lengths for steps $a$ through $b$ is $\langle 1, \ldots, 2^{b-a} \rangle$. Likewise, the minimum congestion we incur is at least $\langle 1, \ldots, 2^{b-a} \rangle$, which corresponds to each step from $a$ through $b$. The topology that establishes direct connections according to step $a$'s communication pattern, i.e., a direct link between $u$ and $u+2^{a-1}$, achieves exactly this minimum sum of path lengths and congestion. Specifically, for step $a$, the communication distance is reduced to $1$. For step $a+1$, node $u$ communicates with $u+2^{a}$, which is at distance $2$: in our chosen topology, $u$ connects to $u+2^{a-1}$, which in turn connects to $u+2^{a-1}+2^{a-1} = u+2^{a}$. The proof follows for both path lengths and congestion.

We can now express the completion time of steps $a-b$ where the communication pattern matches that of step $a$ as:
\vspace{-2mm}
\begin{equation}
t_c(a,\ b) =  \sum_{i=a}^{b} DCT(m_i \cdot \mathcal{M}_i,\ G_{a,b}) \nonumber 
\end{equation}
\vspace{-4mm}
\begin{align}
&= \alpha\cdot(b-a+1) + \delta\cdot\sum_{i=a}^{b}2^{i-a} + \beta\cdot m \cdot \sum_{i=a}^{b} \frac{1}{2^{i}}\cdot 2^{i-a} \nonumber \\
&= \alpha\cdot(b-a+1) + \delta \cdot \left(2^{\, (b-a+1)} - 1\right) + \beta\cdot m \cdot \frac{b-a+1}{2^{a}}
\end{align}
% \Fn{\Prop{$a,b$}}{
%   \Return{$\delta \cdot \sum_{i=a}^{b} 2^{\,i-a} \;=\; \delta \cdot \left(2^{\, (b-a+1)} - 1\right)$}
% }

% \Fn{\Comm{$a,b$}}{
%   \Return{$\beta\,m \cdot \sum_{i=a}^{b} \frac{1}{2^{i}}\cdot 2^{\,i-a} \;=\; \beta\,m \cdot \frac{b-a+1
%   }{2^{a}}$}

Observation~\ref{obs:connectivity} shows certain topologies can preserve connectivity without further reconfiguration. This removes forward dependencies in reconfiguration decisions. Observation~\ref{obs:opt-range} finds the optimal topology for any range of steps $a$ through $b$ in $O(1)$ time, without additional work. Together, they hugely simplify the topology search space to just $O(1)$, leaving the key question: when should we reconfigure? 

Building on Observations~\ref{obs:connectivity} and~\ref{obs:opt-range}, we synthesize optimal schedules for recursive doubling using the dynamic programming approach we described in \S\ref{sec:harvest} (Algorithm~\ref{alg:algorithm}). In particular, the function \texttt{\textcolor{takeawaycolor}{CompletionTime(a,b)}} simplifies i.e., finding $G_{a,b}$ is $O(1)$. The rest of the procedure remains the same as earlier: \first the function \texttt{\textcolor{takeawaycolor}{SynthesizeSchedule(s,k)}} synthesizes the optimal schedule corresponding to a given $k$ number of reconfigurations; \second we iterate from $0$ to $s=\log_2(n)$ number of reconfigurations, synthesize the schedule for each, and return the best global schedule. Given the logarithmic number of steps in recursive doubling AllReduce, synthesizing optimal switching schedules reduces to polylogarithmic complexity of $O((\log n)^4)$.

\begin{figure*}
\centering
\begin{subfigure}{0.24\linewidth}
\centering
\includegraphics[width=1\linewidth]{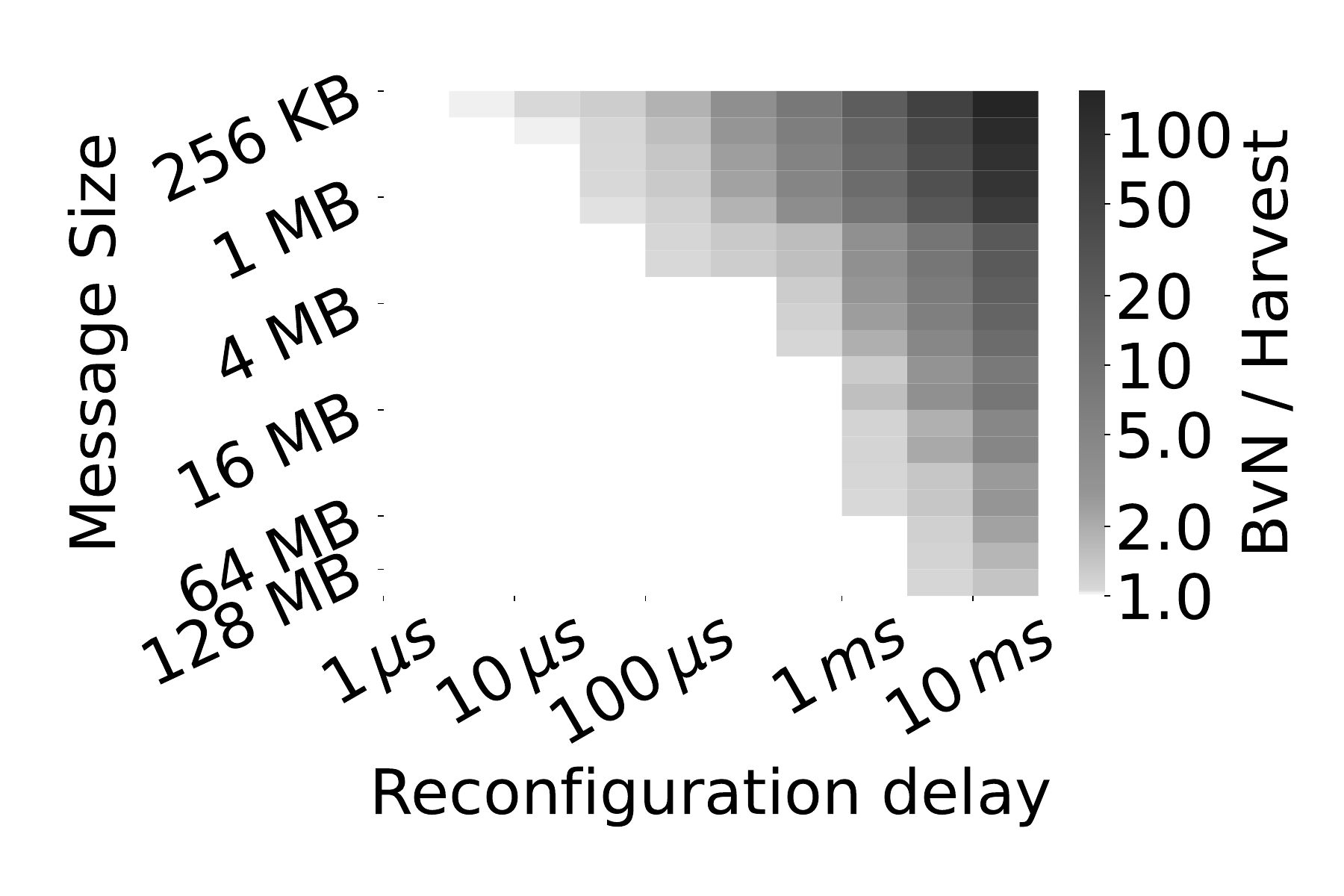}
\caption{}
\label{fig:emu-bvn}
\end{subfigure}
\begin{subfigure}{0.24\linewidth}
\centering
\includegraphics[width=1\linewidth]{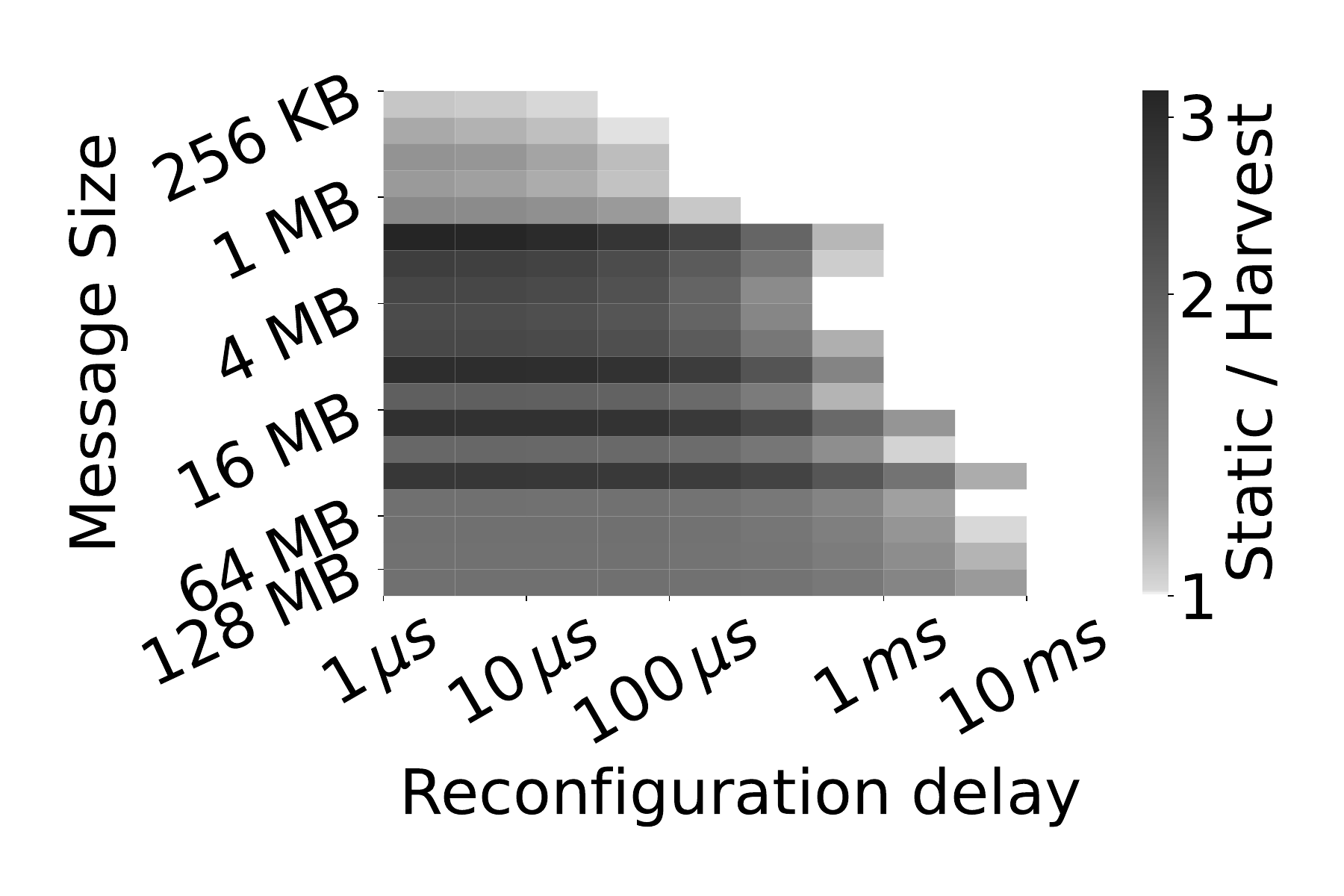}
\caption{}
\label{fig:emu-static}
\end{subfigure}
\begin{subfigure}{0.24\linewidth}
\centering
\includegraphics[width=1\linewidth]{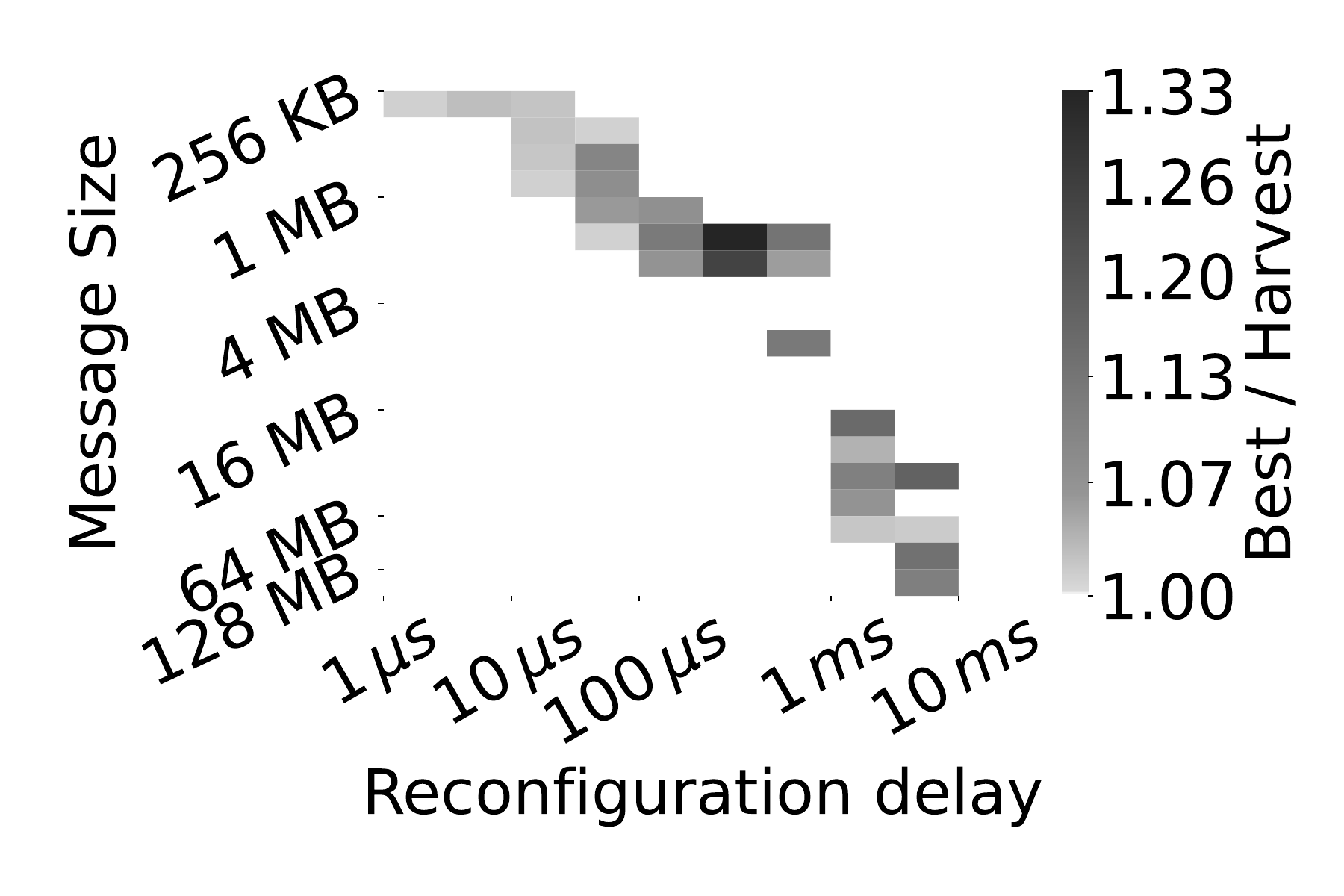}
\caption{}
\label{fig:emu-best}
\end{subfigure}
\begin{subfigure}{0.24\linewidth}
\centering
\includegraphics[width=1\linewidth]{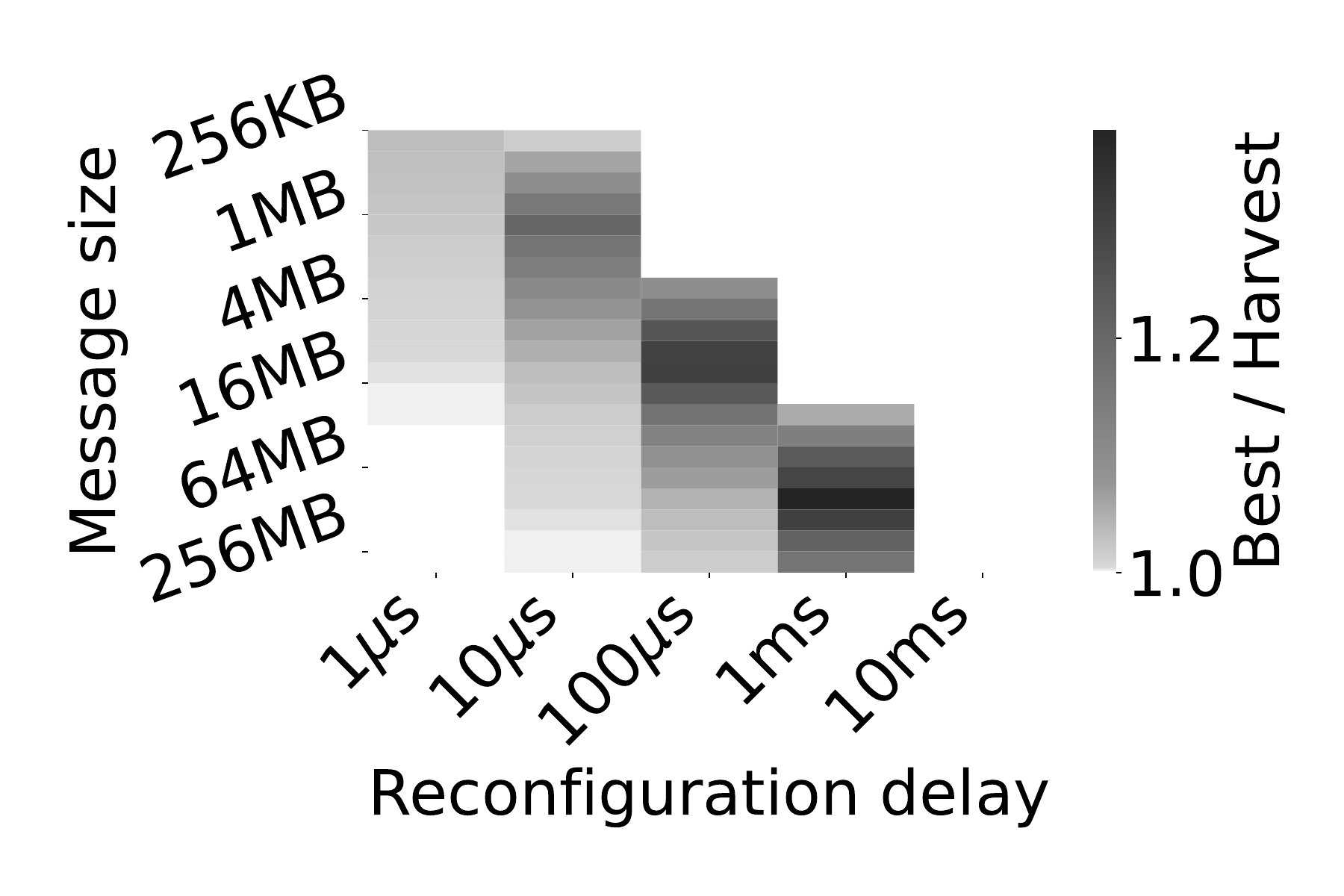}
\caption{}
\label{fig:synth-emu-params}
\end{subfigure}
\caption{[Hardware emulation] Heatmaps showing the speedup in collective completion time achieved by \name relative to (a) BvN-based schedules, (b) a static ring, (c) the best of BvN and static, and (d) the speedup estimated by synthesis, which closely matches the measured speedups observed in hardware emulation.}
\label{fig:emu}
\end{figure*}

\vspace{-3mm}
\section{Evaluation}
\label{sec:evaluation}
% We evaluate the performance of \name by measuring collective completion time for widely studied collective communication algorithms. In particular, we look at All-Reduce and All-to-All workloads on static topologies ranging from torus networks to expander graphs. We supply collective algorithms optimized for the underlying static topology to \name and observe \textit{when} and \textit{how} it smartly reconfigures these topologies based on reconfiguration delay and demand at each step. Our evaluations make use of both simulators and a real testbed to validate the performance of \name's generated schedules. 
We evaluate the schedules \name synthesizes and their completion time across several collective communication algorithms. We compare their performance against BvN schedules and static  topologies. 
Our evaluation spans simulation, hardware emulation, and numerical optimization, capturing both performance and system-level effects.

\vspace{-3mm}
\subsection{Setup}

\textbf{Network:}
We consider networks with $8$ to $64$ GPUs, representative of typical scale-up domains.
Each GPU has $d$ ports, where $d$ ranges from $2$ (e.g., a 1-D ring) to $6$ (e.g., a 3-D torus).
Unless stated otherwise, we set the per-port bandwidth to $800$Gbps for simulations and $100$Gbps for our hardware emulation experiments.
We vary the setup latency $\alpha$, link propagation delay $\delta$, and the interconnect reconfiguration delay $\alpha_r$ over a wide range, from $10$ns (e.g., tunable lasers~\cite{10.1145/3387514.3406221}) to $10$ms (e.g., 3-D MEMS~\cite{polatis}).
This range captures diverse photonic switching technologies and allows us to identify regimes in which reconfiguration is beneficial.

% For our evaluations, we sweep a range of configurations to capture diverse operating conditions. We evaluate a rich set of base topologies, including N-D tori, Kautz graphs, N-D shifted rings, expander graphs, and De Bruijn graphs, among others. Unless otherwise stated, experiments are conducted on a
% network with $N = 64$ nodes and message sizes ranging from $1{,}024$ bytes to
% $1$~GB are evaluated. We vary the propagation delay across $\{10, 500, 10{,}000\}$~ns, with
% corresponding $\alpha$ delays of $\{10{,}000, 500, 10{,}000\}$~ns to understand how the trade-off between the $\alpha, \delta, and \beta$ terms are handled by \name. Reconfiguration delays range from $0$~ns to $1$~ms, while the link bandwidth is fixed at $800$~Gbps. These parameters ensures that our analysis captures the performance trends across a wide spectrum
% of realistic scenarios.

\begin{figure}[t]
\centering
\includegraphics[width=0.9\linewidth]{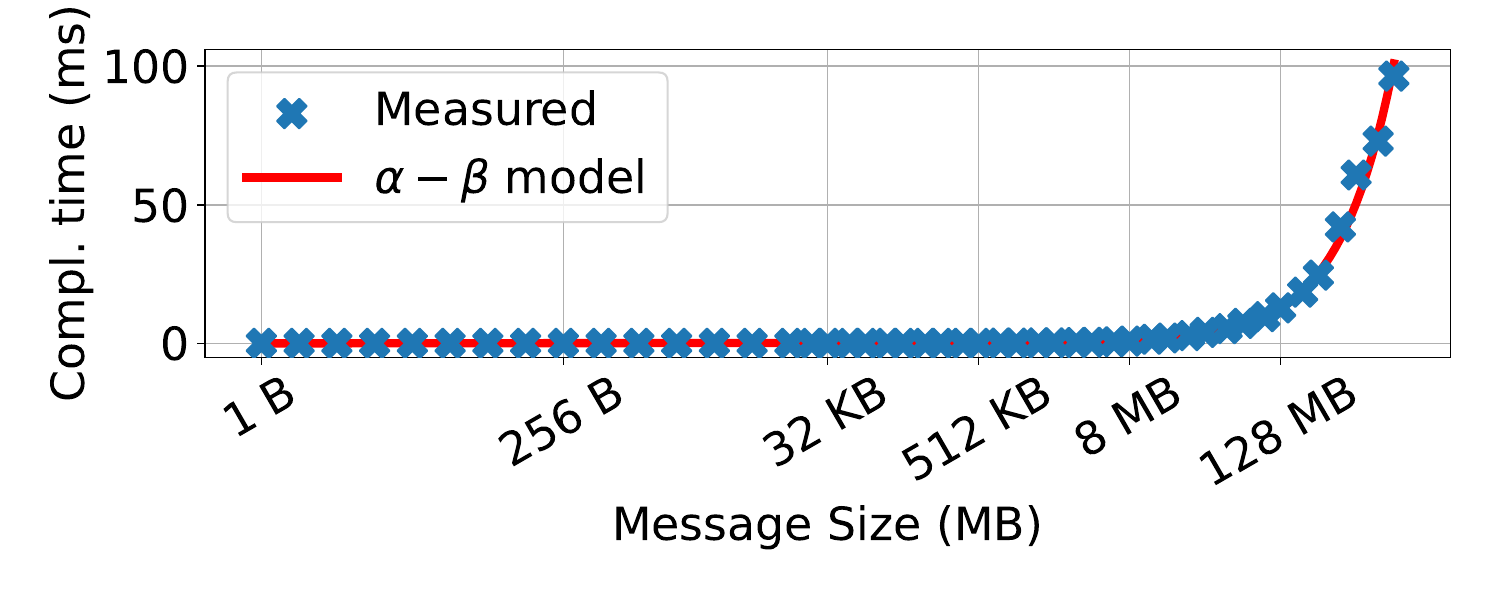}
\vspace{-4mm}
\caption{[Hardware emulation] Correlation between the $\alpha$--$\beta$ cost and measured completion times across different message sizes for one-hop communication.}
\label{fig:alpha-beta-model}
\vspace{-5mm}
\end{figure}

\myitem{Baselines:}
We compare \name against two representative baselines: \first a \emph{static} topology that remains fixed throughout the collective and \second \emph{BvN schedules} that reconfigure the topology at every communication step to directly connect the communicating GPU pairs~\cite{10.1145/3748273.3749210,10.1145/3748273.3749203}.
For static topologies, we consider rings, $2$D and $3$D tori, and generalized Kautz graphs~\cite{305352}.
Together, these baselines capture the two extremes of the design space: interconnects that never reconfigure, and those that reconfigure at every step.
% Comparing against them allows us to evaluate \name's ability to balance reconfiguration overhead against congestion and propagation delay.

% \textbf{Baselines:} There are mainly two baselines that fit our problem formulation: static interconnects and reconfigurable interconnects. In static interconnects, we have a pre-selected topology that remains fixed for the duration of the collective. For reconfigurable interconnects, we consider ones that reconfigure based on the BvN decomposition of aggregated demand matrices of a collective. In simple terms, it would reconfigure after each step of the collective to form direct connect topologies between nodes that will be communicating in the next step.

\myitem{Collective algorithms:} The input to \name and other baselines is the workload which consists of standard collective algorithms.
For AllReduce, we evaluate recursive doubling~\cite{10.1007/978-3-540-24685-5-1}, Swing~\cite{295653} (which is equivalent to Bine Butterfly~\cite{10.1145/3712285.3759835} in terms of communication pattern), and Bruck's concatenation algorithm~\cite{642949}.
For All-to-All, we evaluate the total exchange or transpose operation (direct All-to-All) as well as Bruck's All-to-All (index) algorithm~\cite{642949}.
For broadcast, we evaluate the binomial tree and binary tree algorithms.

% Once we select a topology, we run collective algorithms that are known to be optimized for the given topology. For instance, Rabenseifer's Recursive Doubling algorithm is run on a ring topology. Similarly, we evaluate Swing, Bruck's All-to-All and All-Gather, Bine, etc. \\
% \textcolor{blue}{[mr: Maybe mention search space??]}

\myitem{Simulations:}
We perform packet-level simulations using Astra-Sim~\cite{9238637,10158106}, which we extend to support circuit-switched interconnects. The simulator now accepts a topology reconfiguration schedule as input and dynamically reconfigures the topology according to it during collective operations.

% Our evaluation covers network simulations primarily based on Astra-sim, a simulator built on top of NS3 \cite{ns3}, that helps model various collective communication algorithms. We extended the existing framework to support circuit switched network topologies, where we can reconfigure after each step of the collective communication algorithm.

\myitem{Hardware emulation:} We would need specialized hardware to directly validate \name across a broad range of photonic switching technologies that have varying reconfiguration delays~---~this is costly.
%Directly validating \name across a broad range of photonic switching technologies with varying reconfiguration delays would require specialized hardware and incur substantial cost.
Instead, we emulate a reconfigurable photonic interconnect on an $8$-GPU testbed. The GPUs are connected in a ring topology via $8$ BlueField-3 NICs equipped with $100$Gbps optical transceivers.
GPU--NIC communication uses GPUDirect RDMA over PCIe, and we ensure that the available PCIe bandwidth exceeds the NIC I/O bandwidth to eliminate PCIe bottlenecks. We implement NIC--NIC communication and routing via the BlueField-3 eSwitch, with flow steering offloaded to hardware.

We use NCCL to execute collective operations step by step, and measure the runtime of each step until a reconfiguration is required.
At that point, we pause execution, update the interconnect configuration, and resume the remaining steps.
We compute the total completion time as the sum of the step-wise runtimes plus a fixed physical reconfiguration penalty. This methodology is practical, cost-effective, and yields representative results.

\myitem{Numerical evaluation:}
We implement the synthesis component of \name in C++ and use Gurobi to solve the subproblems in each dynamic program.
We describe the optimization formulation in more detail in \S~\ref{app:misocp}.
For a range of network sizes and topologies, primarily multi-dimensional topologies, we synthesize reconfiguration schedules and report the corresponding collective completion times produced by the optimization. Moreover, we compare the trends in the synthesized schedules with the results obtained from our test-bed experiments.

% For multiport and multidimensional setups, we perform flow level simulation
% to demonstrate performance of \name across various state-of-the-art algorithms applied on a set of more complex topologies like 3D Toris and Kaust Graphs.

\begin{figure*}[t]
     \centering
     % --- Row 1 ---
     \begin{subfigure}[b]{0.24\textwidth}
         \centering
         \includegraphics[width=\textwidth]{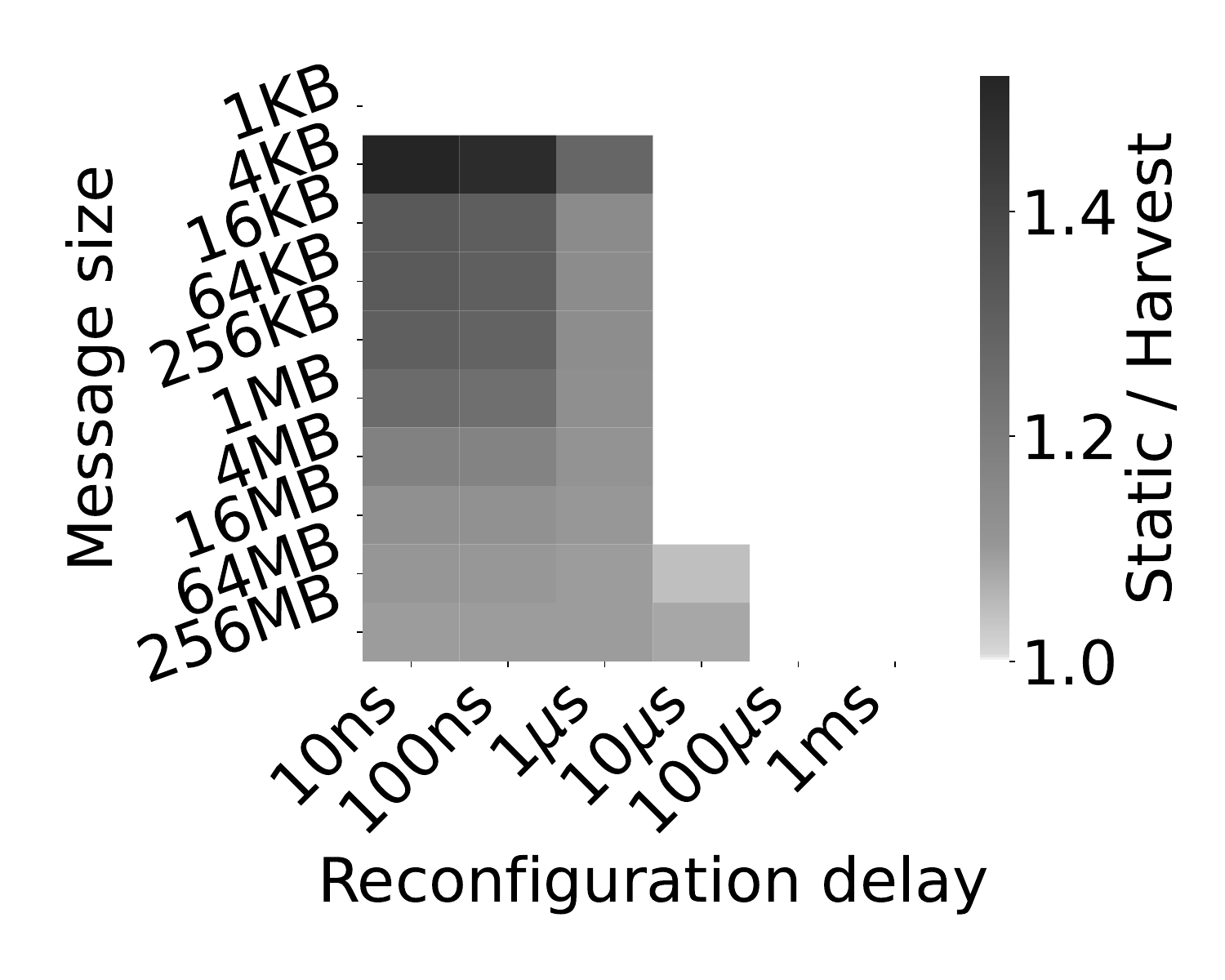}
         \caption{8x8}
         \label{fig:hd_800_row1}
     \end{subfigure}
     \hfill
     \begin{subfigure}[b]{0.24\textwidth}
         \centering
         \includegraphics[width=\textwidth]{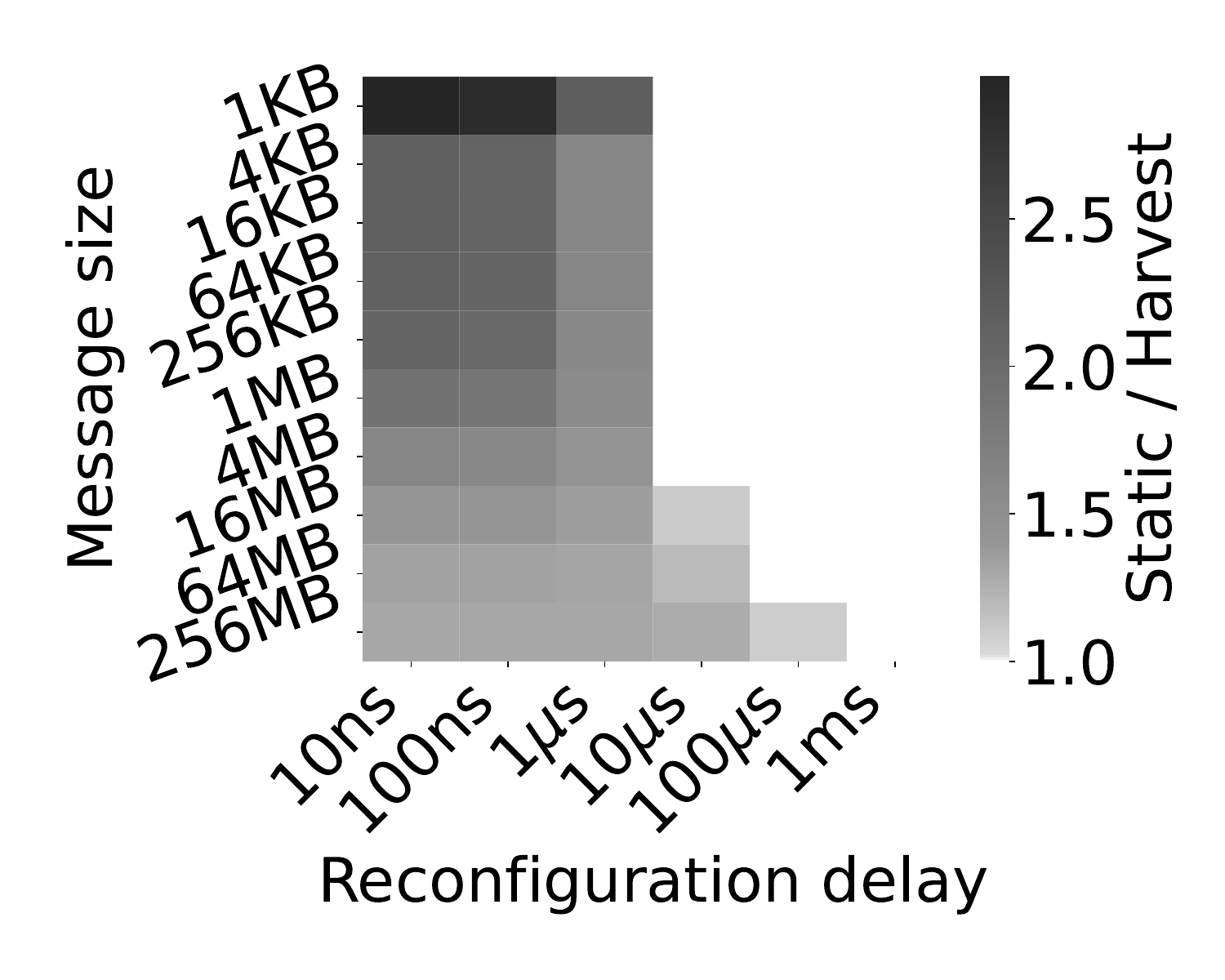}
         \caption{16x4}
         \label{fig:direct_best_row1}
         % \subcaption[]{}
     \end{subfigure}
     \hfill
     \begin{subfigure}[b]{0.24\textwidth}
         \centering
         \includegraphics[width=\textwidth]{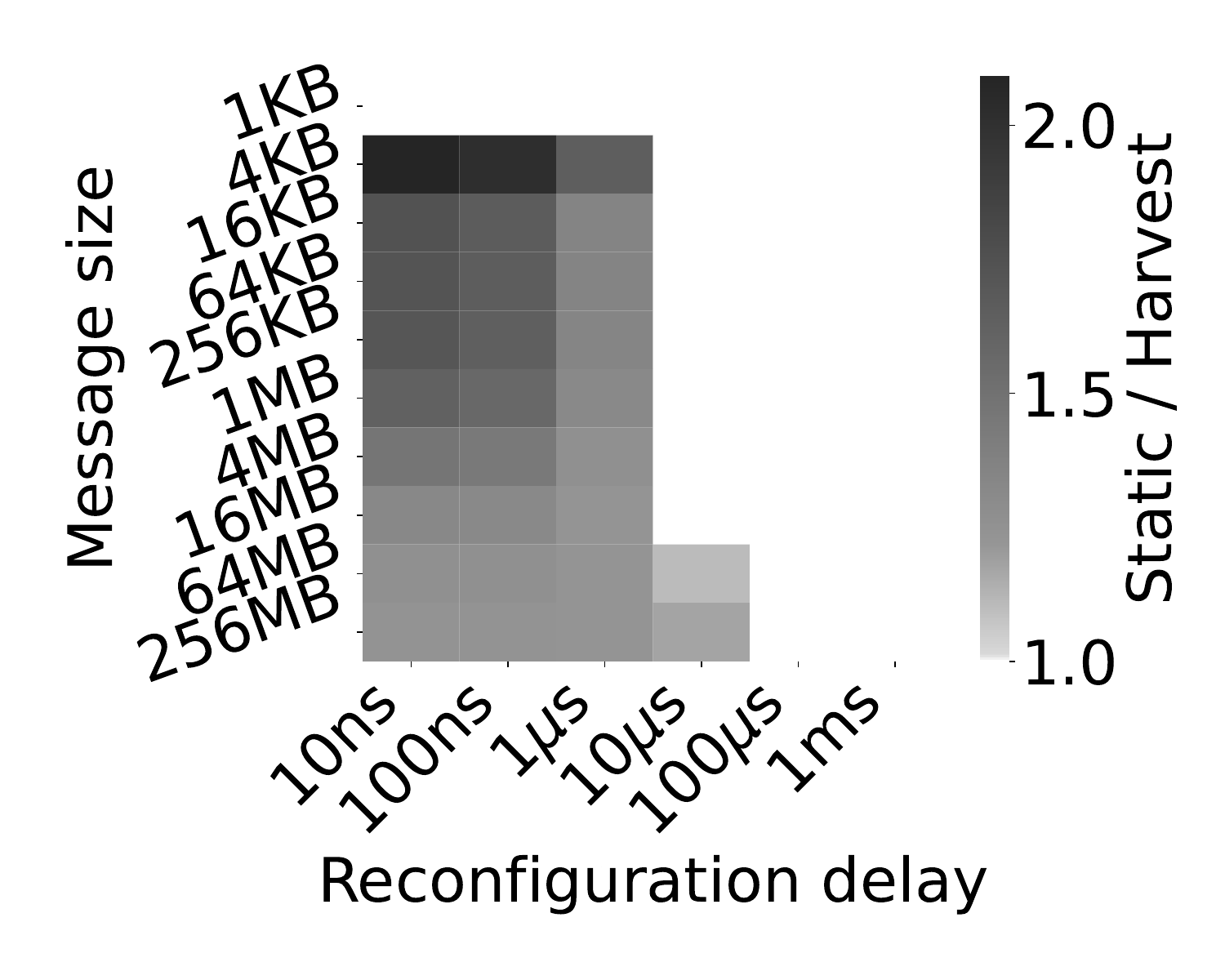}
         \caption{8x4x2}
         \label{fig:swing_800_row1}
         % \subcaption[]{}
     \end{subfigure}
     \hfill
     \begin{subfigure}[b]{0.24\textwidth}
         \centering
         \includegraphics[width=\textwidth]{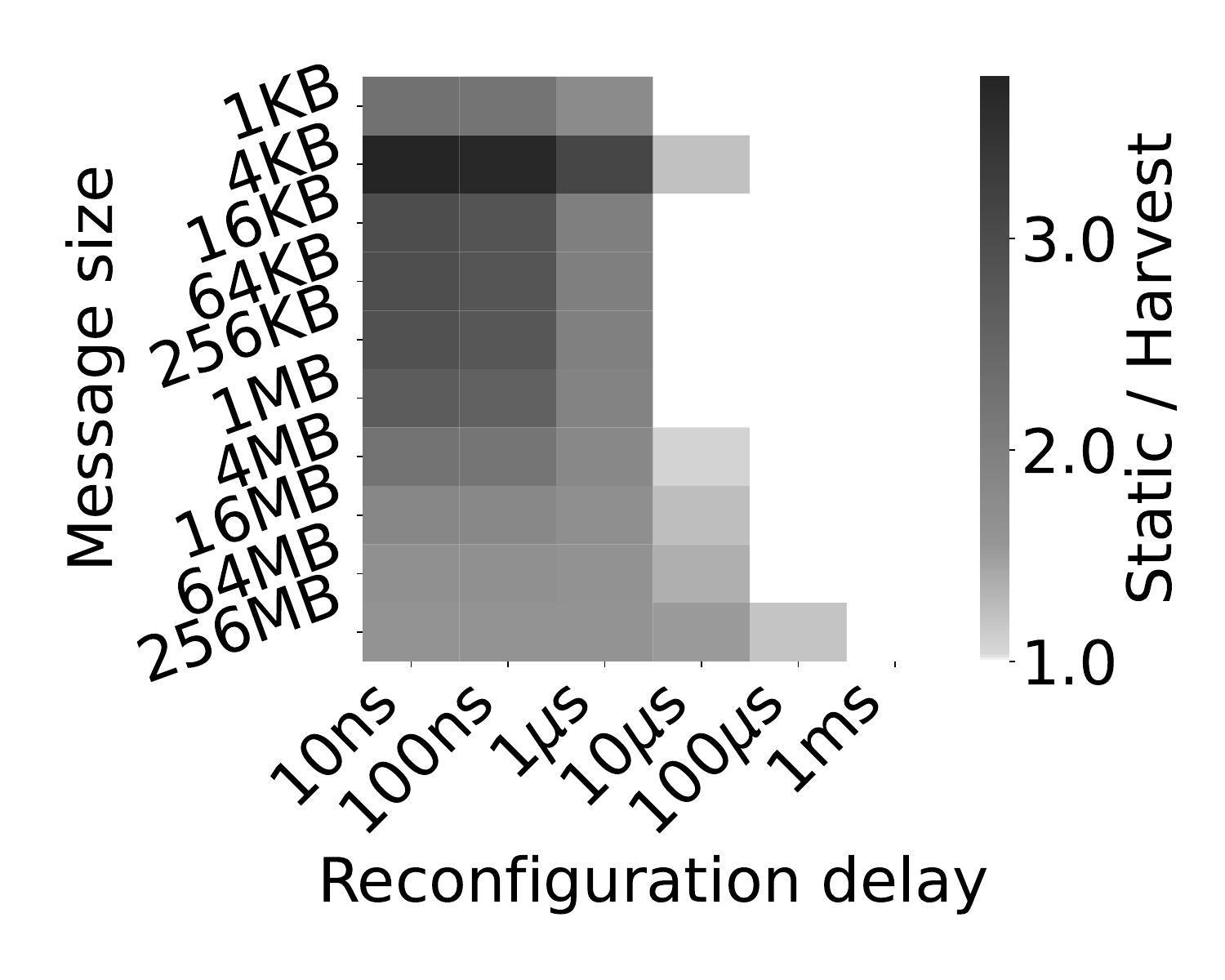}
         \caption{16x2x2}
         \label{fig:direct_800_row1}
         % \subcaption[]{}
     \end{subfigure}\hfill

     % --- Row 2 ---
     \begin{subfigure}[b]{0.24\textwidth}
         \centering
         \includegraphics[width=\textwidth]{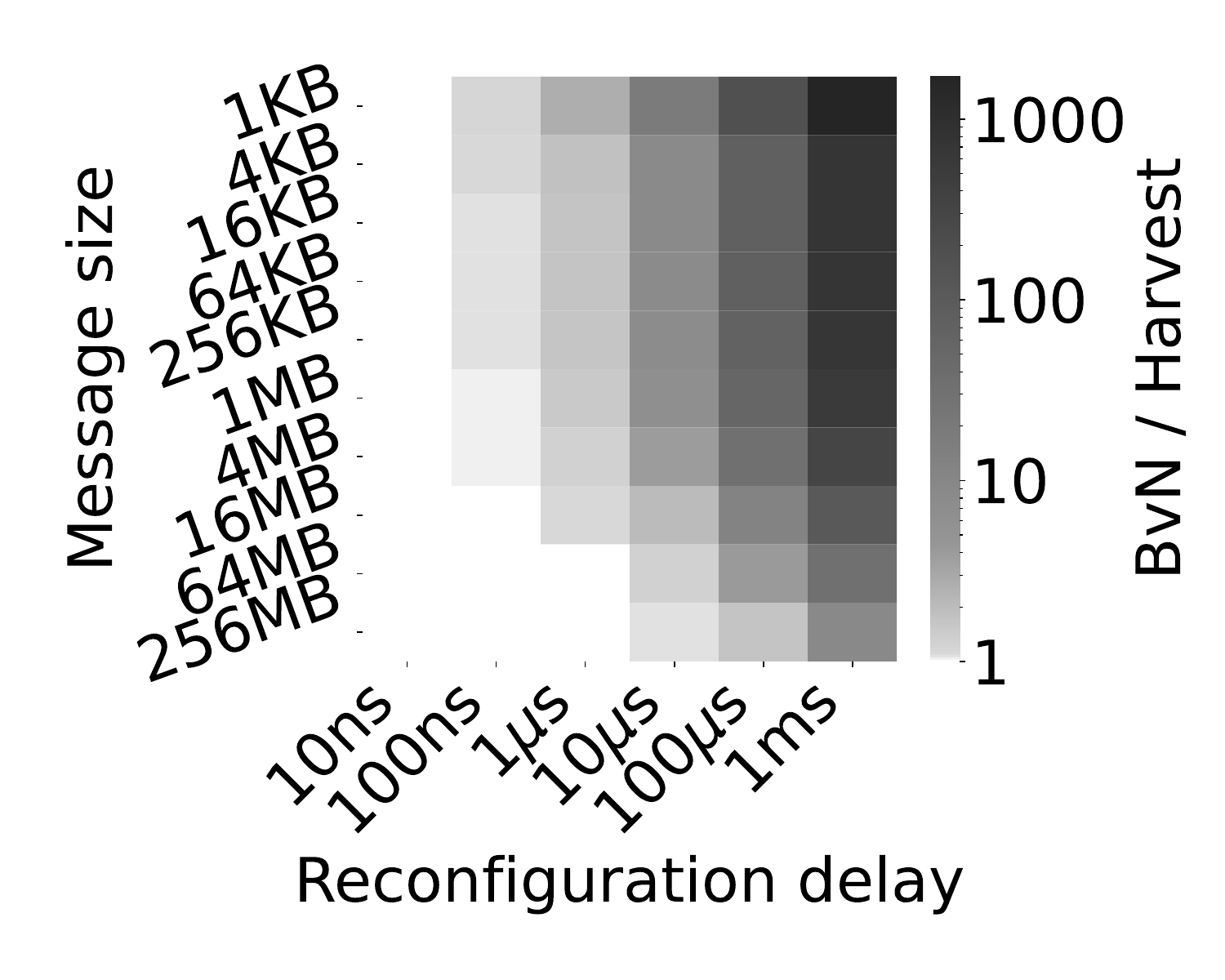}
         \caption{8x8}
         \label{fig:hd_800_row1}
     \end{subfigure}
     \hfill
     \begin{subfigure}[b]{0.24\textwidth}
         \centering
         \includegraphics[width=\textwidth]{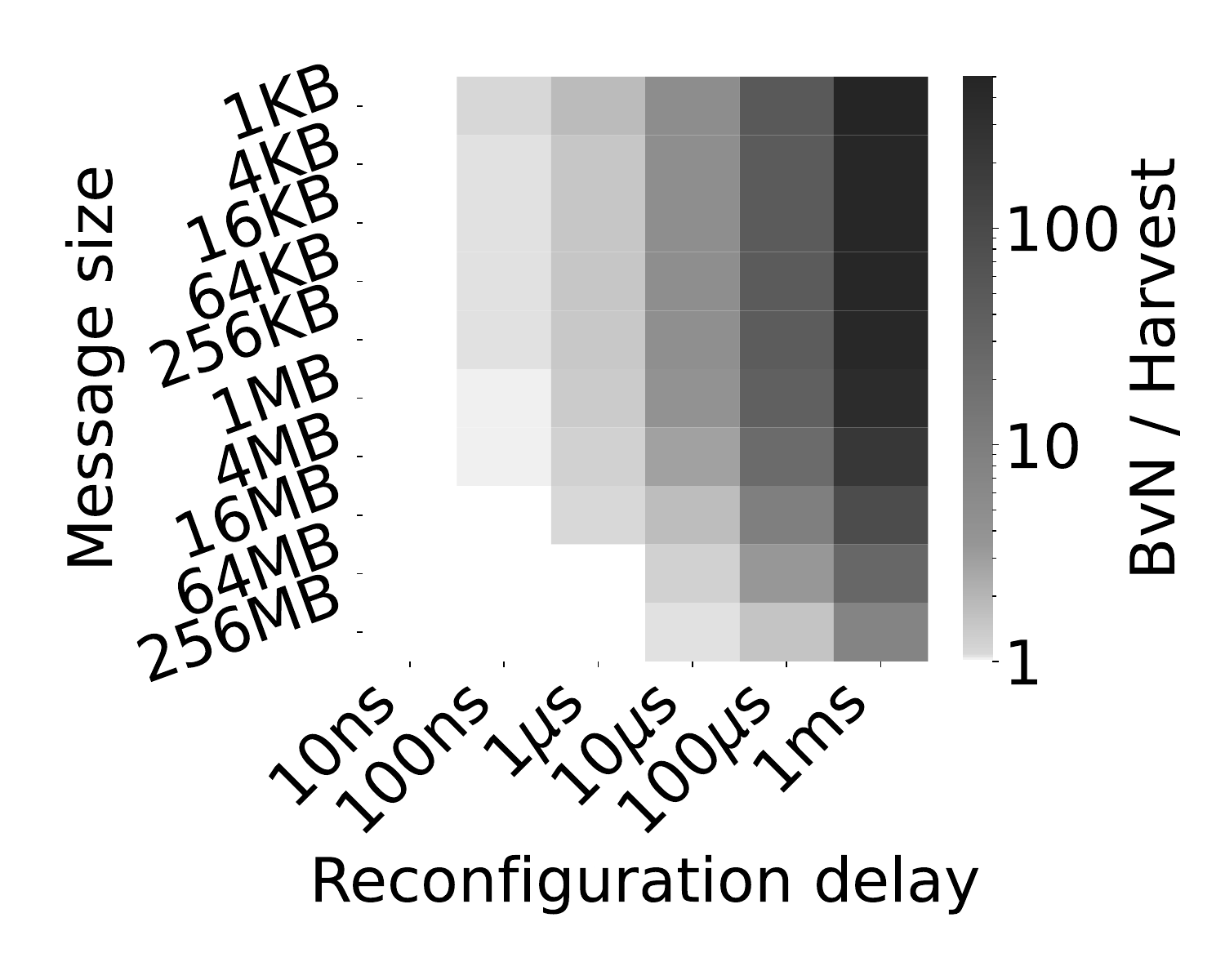}
         \caption{16x4}
         \label{fig:direct_best_row1}
         % \subcaption[]{}
     \end{subfigure}
     \hfill
     \begin{subfigure}[b]{0.24\textwidth}
         \centering
         \includegraphics[width=\textwidth]{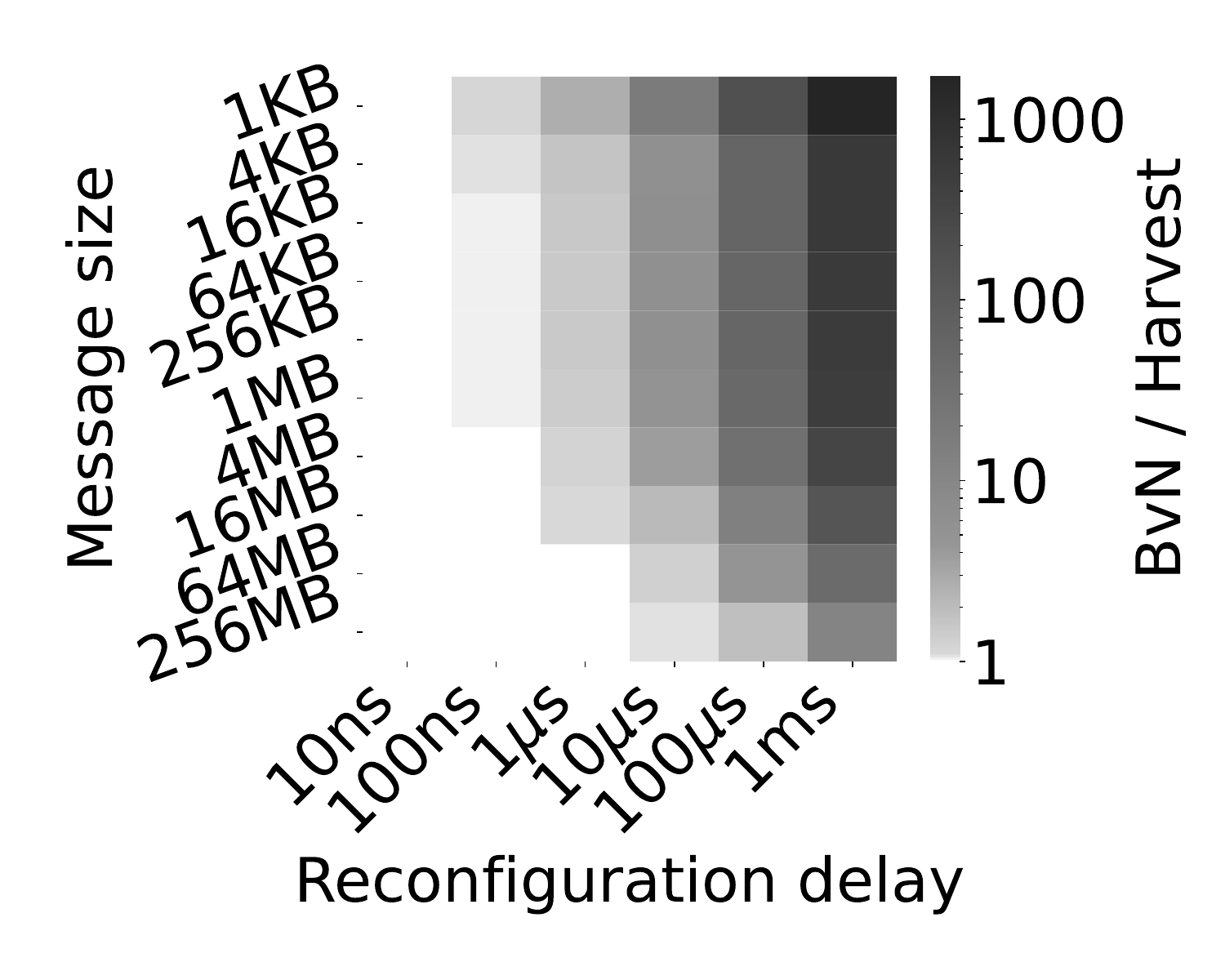}
         \caption{8x4x2}
         \label{fig:swing_800_row1}
         % \subcaption[]{}
     \end{subfigure}
     \hfill
     \begin{subfigure}[b]{0.24\textwidth}
         \centering
         \includegraphics[width=\textwidth]{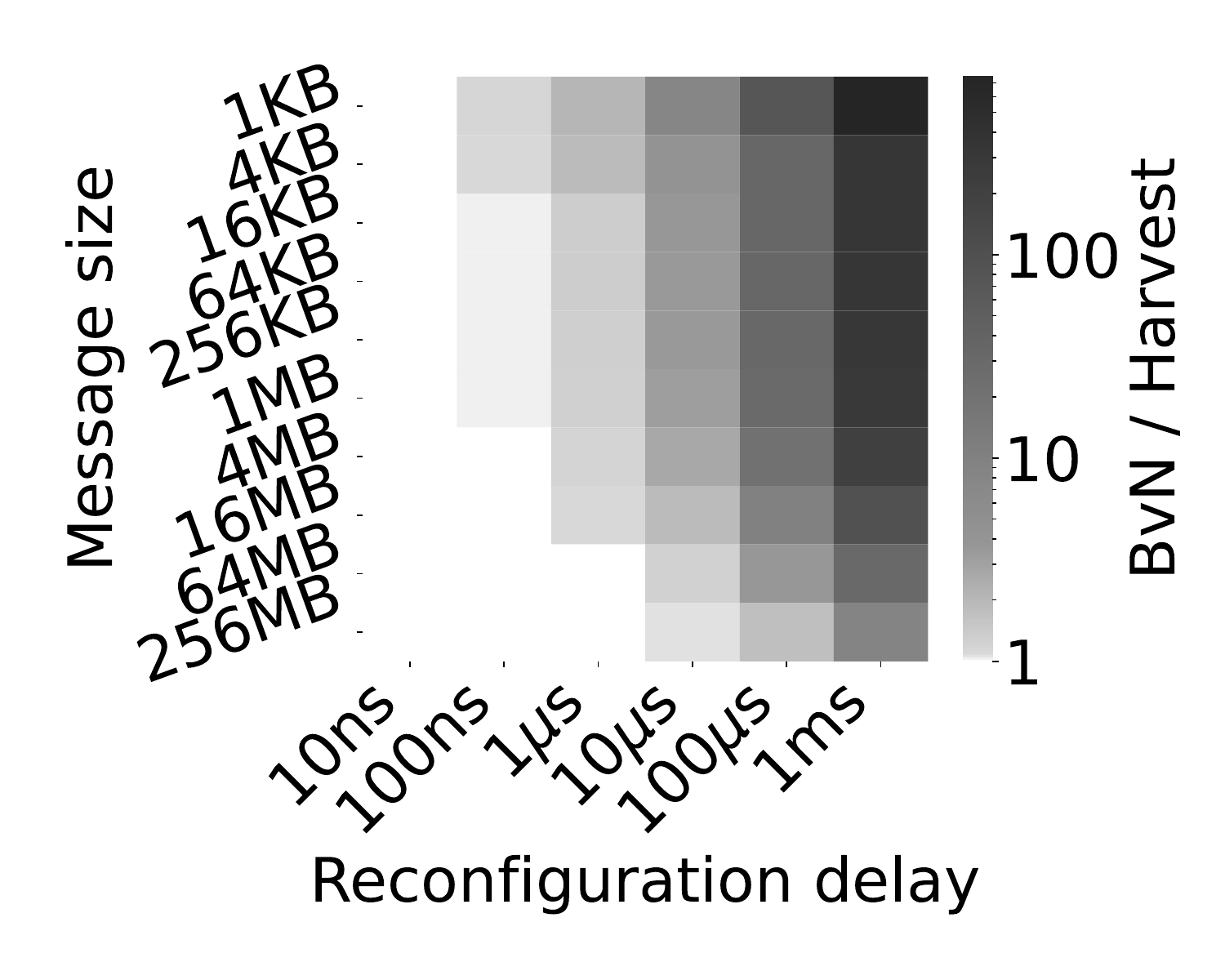}
         \caption{16x2x2}
         \label{fig:direct_800_row1}
         % \subcaption[]{}
     \end{subfigure}
      % --- Row 3 ---
     \begin{subfigure}[b]{0.24\textwidth}
         \centering
         \includegraphics[width=\textwidth]{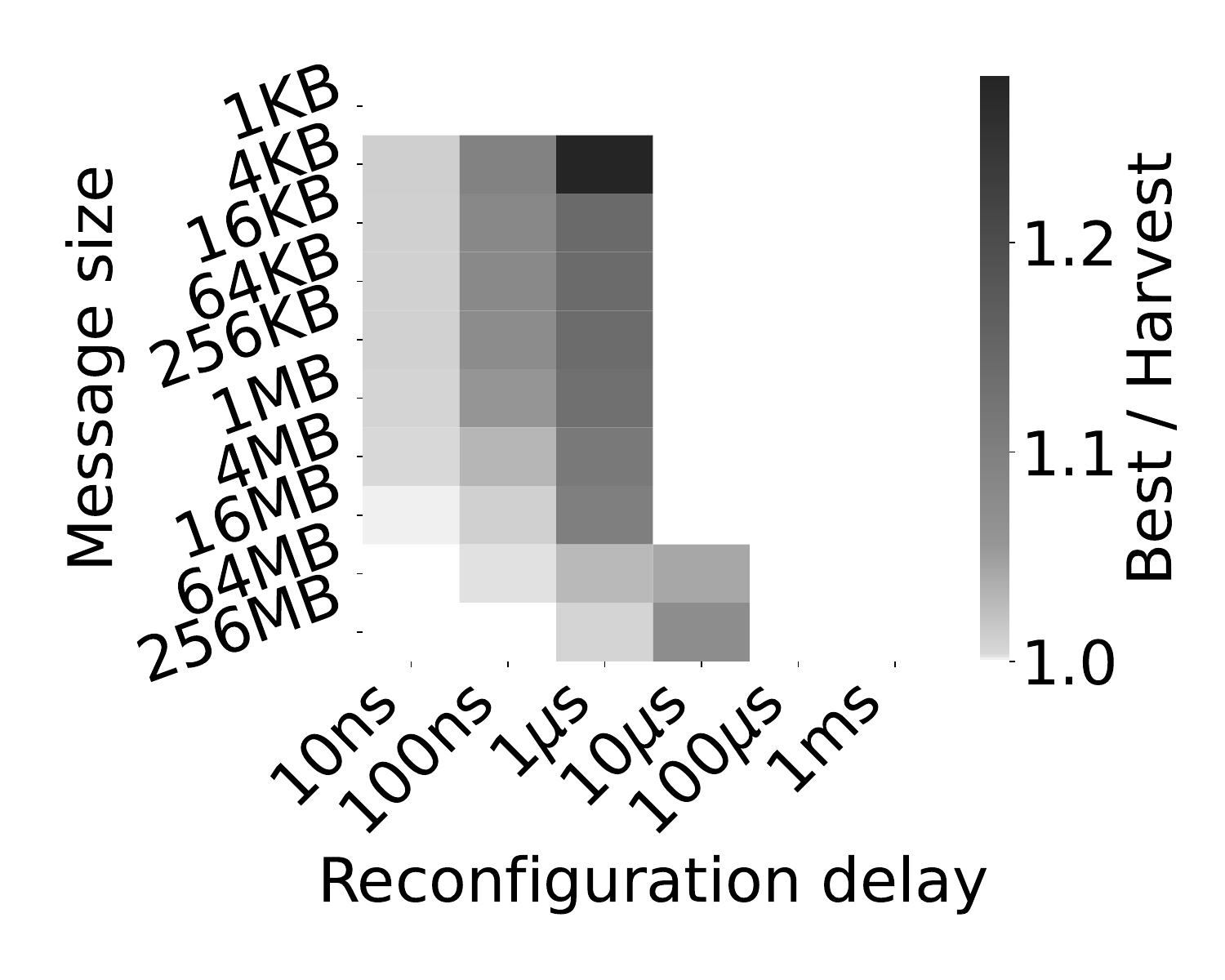}
         \caption{8x8}
         \label{fig:hd_800_row1}
     \end{subfigure}
     \hfill
     \begin{subfigure}[b]{0.24\textwidth}
         \centering
         \includegraphics[width=\textwidth]{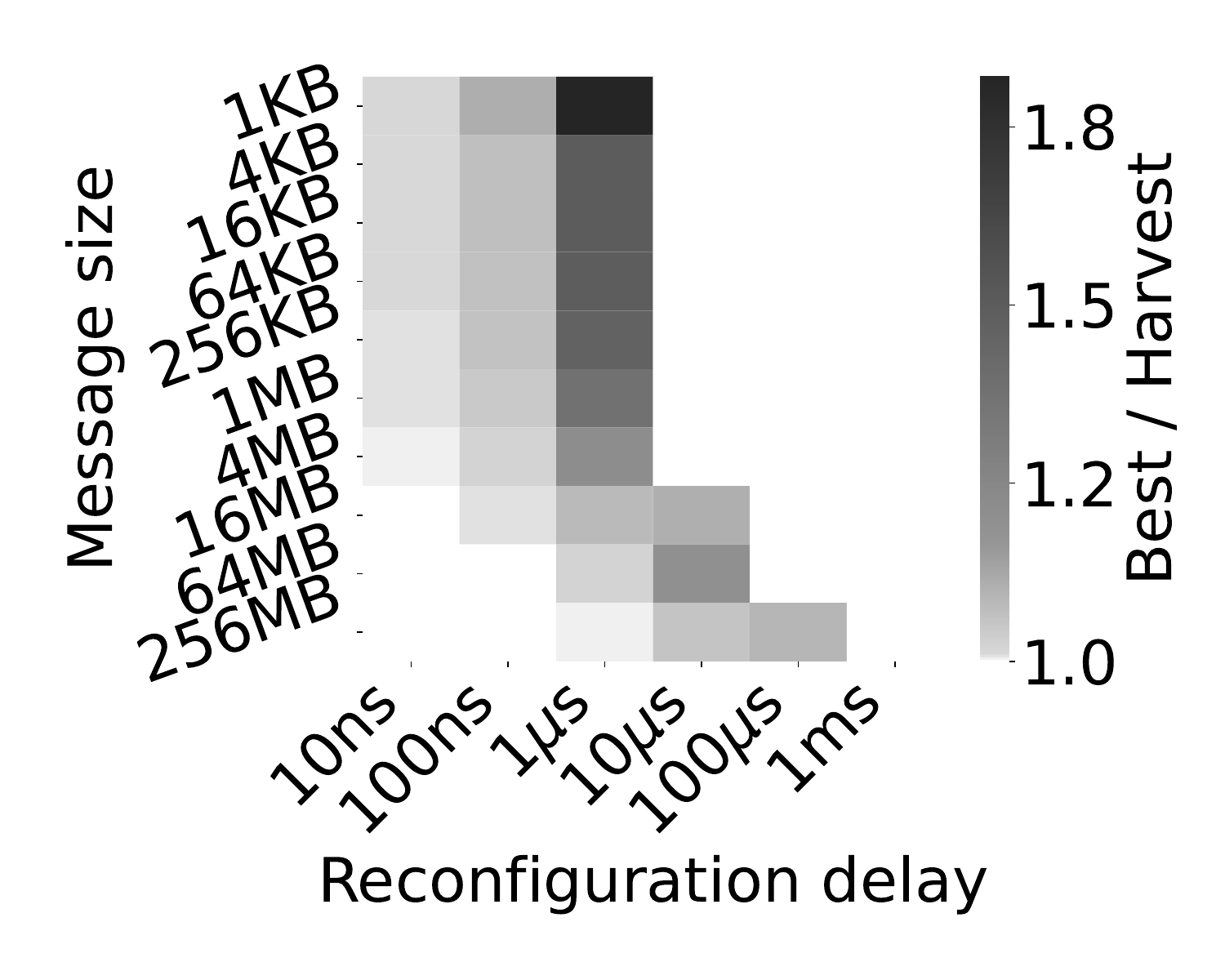}
         \caption{16x4}
         \label{fig:direct_best_row1}
         % \subcaption[]{}
     \end{subfigure}
     \hfill
     \begin{subfigure}[b]{0.24\textwidth}
         \centering
         \includegraphics[width=\textwidth]{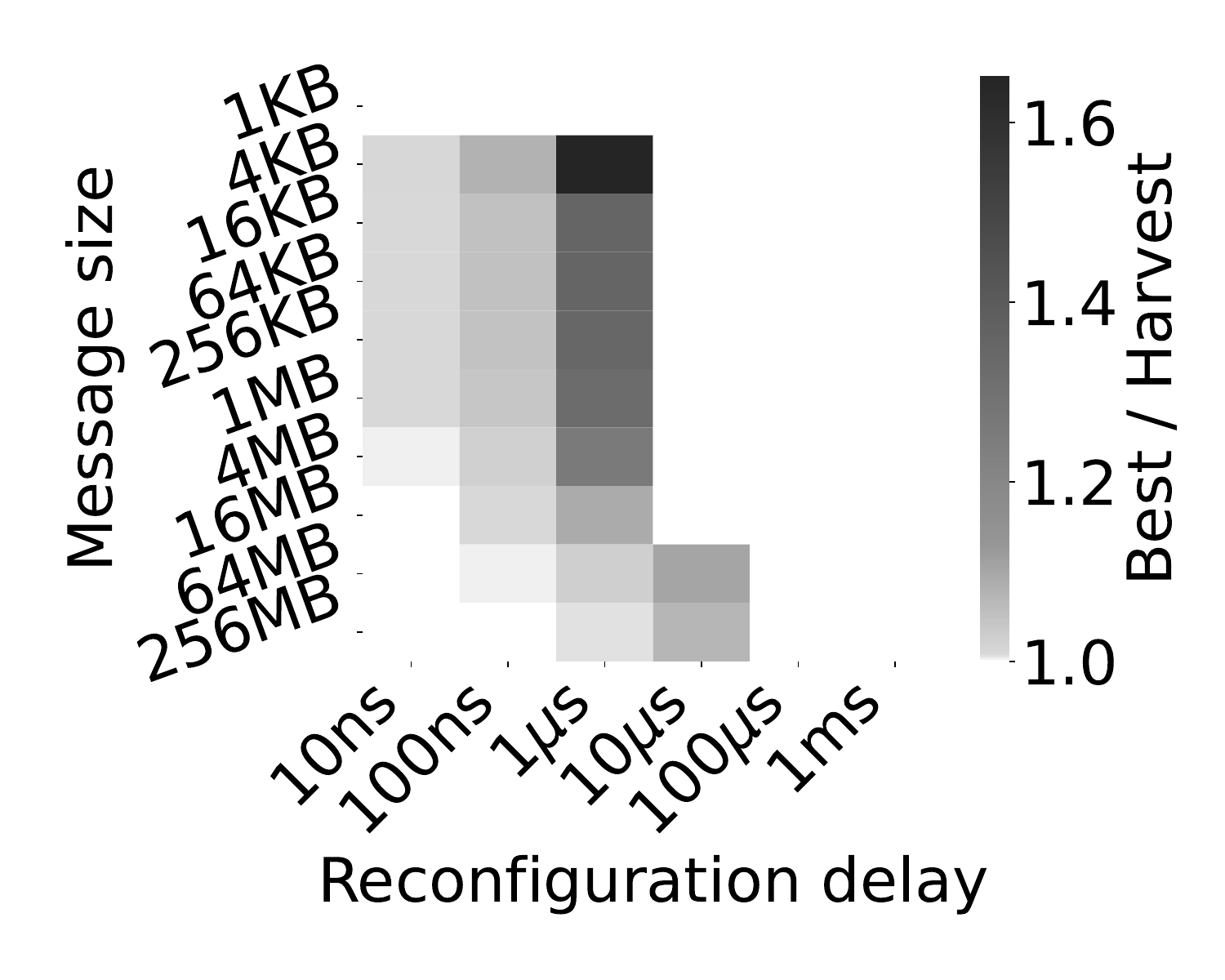}
         \caption{8x4x2}
         \label{fig:swing_800_row1}
         % \subcaption[]{}
     \end{subfigure}
     \hfill
     \begin{subfigure}[b]{0.24\textwidth}
         \centering
         \includegraphics[width=\textwidth]{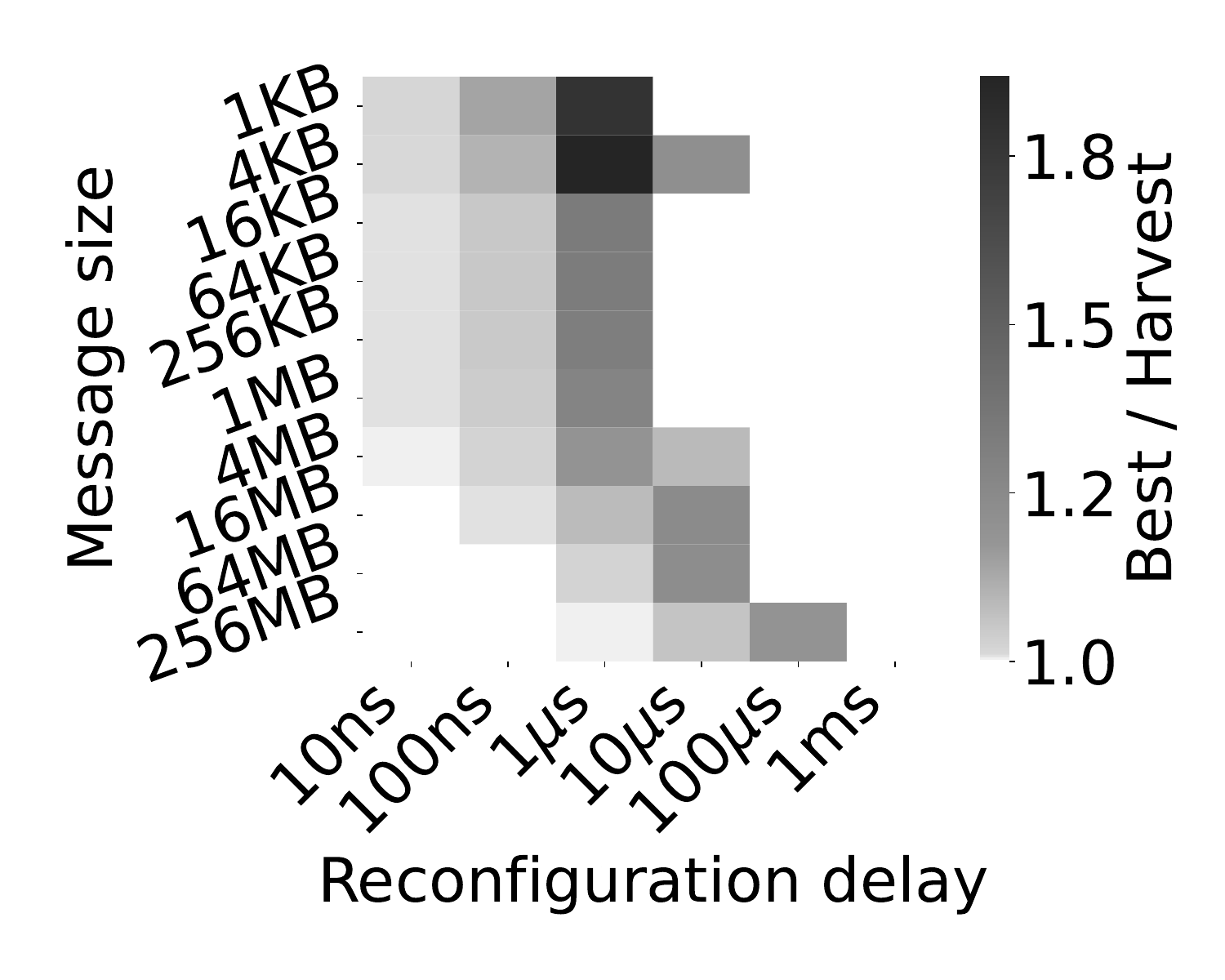}
         \caption{16x2x2}
         \label{fig:direct_800_row1}
         % \subcaption[]{}
     \end{subfigure}

     \caption{[Numerical optimization] Heatmaps showing the speedup in collective completion time achieved by \name relative to static Torus, BvN (reconfiguring at every step), and best among both, for Swing AllReduce for various Torus configurations. We use multi-port Swing, along with mirroring for all multi-dimensional topologies.}
     \label{fig:swing-multidim}
\end{figure*}

\begin{figure*}[t]
     \centering
     % --- Row 1 ---
     \begin{subfigure}[b]{0.24\textwidth}
         \centering
         \includegraphics[width=\textwidth]{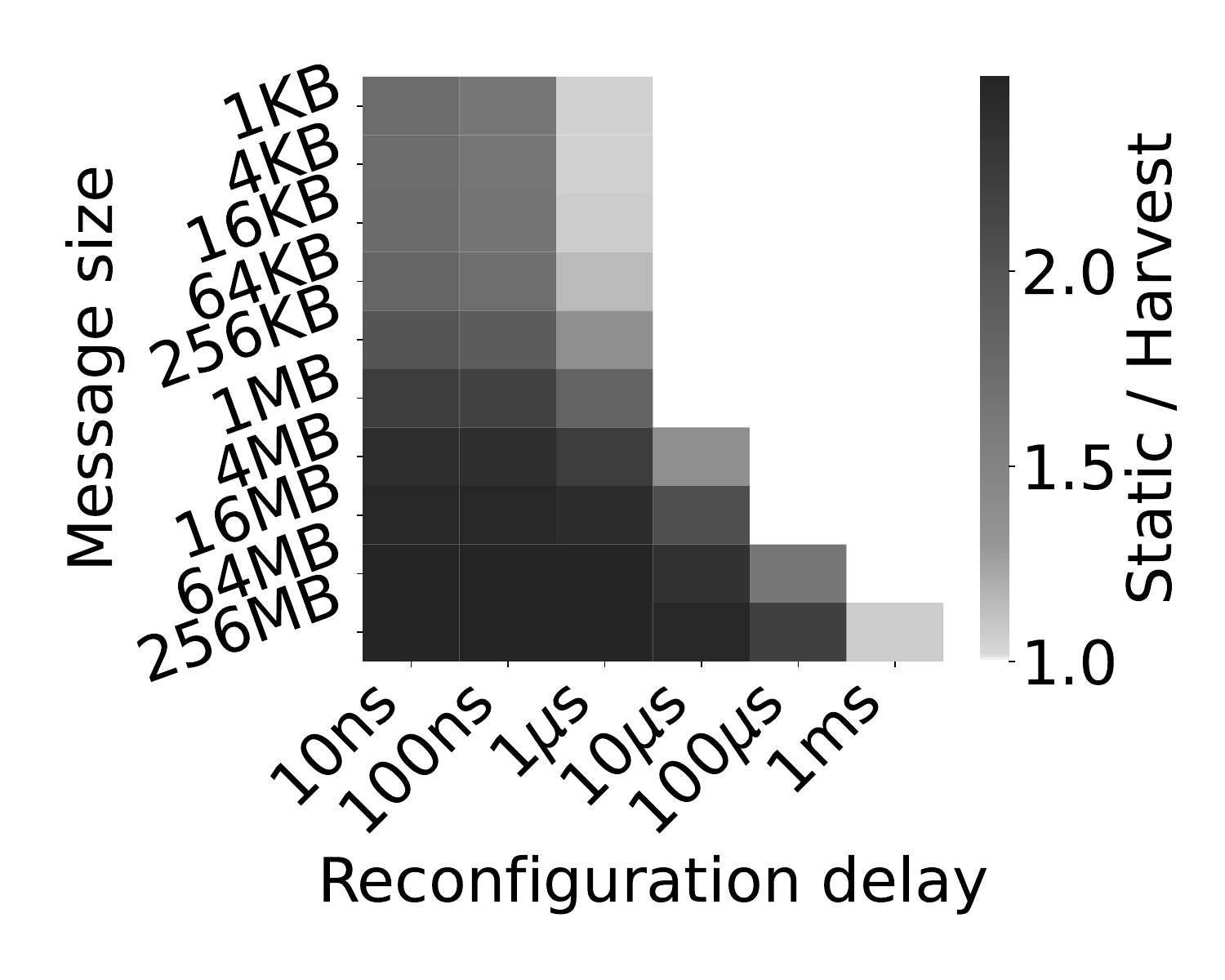}
         \caption{Bruck $r=4$ All-to-All}
         \label{fig:hd_800_row1}
     \end{subfigure}
     \hfill
     \begin{subfigure}[b]{0.24\textwidth}
         \centering
         \includegraphics[width=\textwidth]{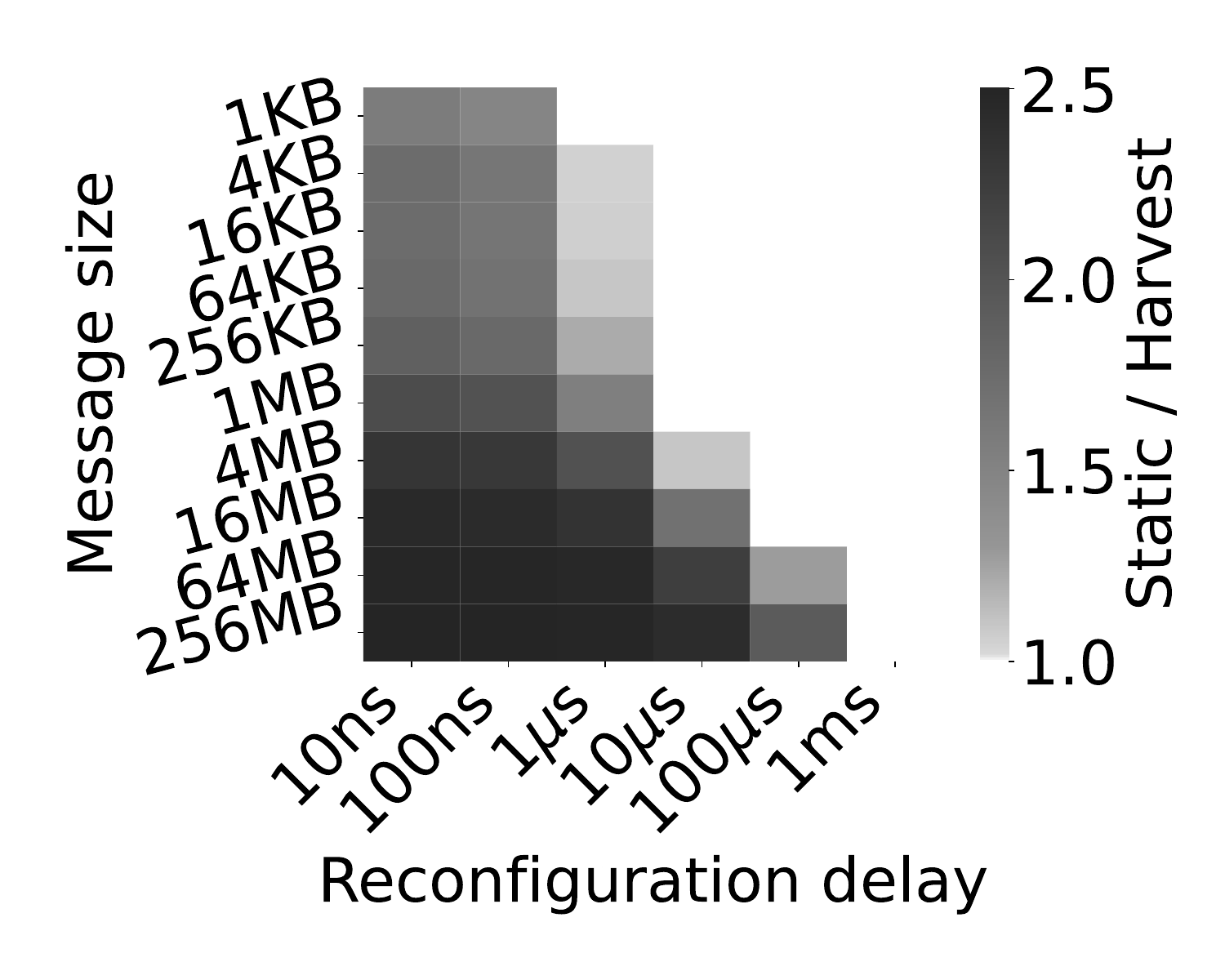}
         \caption{Bruck $r=4$ AllGather}
         \label{fig:swing_800_row1}
         % \subcaption[]{}
     \end{subfigure}
     \hfill
     \begin{subfigure}[b]{0.24\textwidth}
         \centering
         \includegraphics[width=\textwidth]{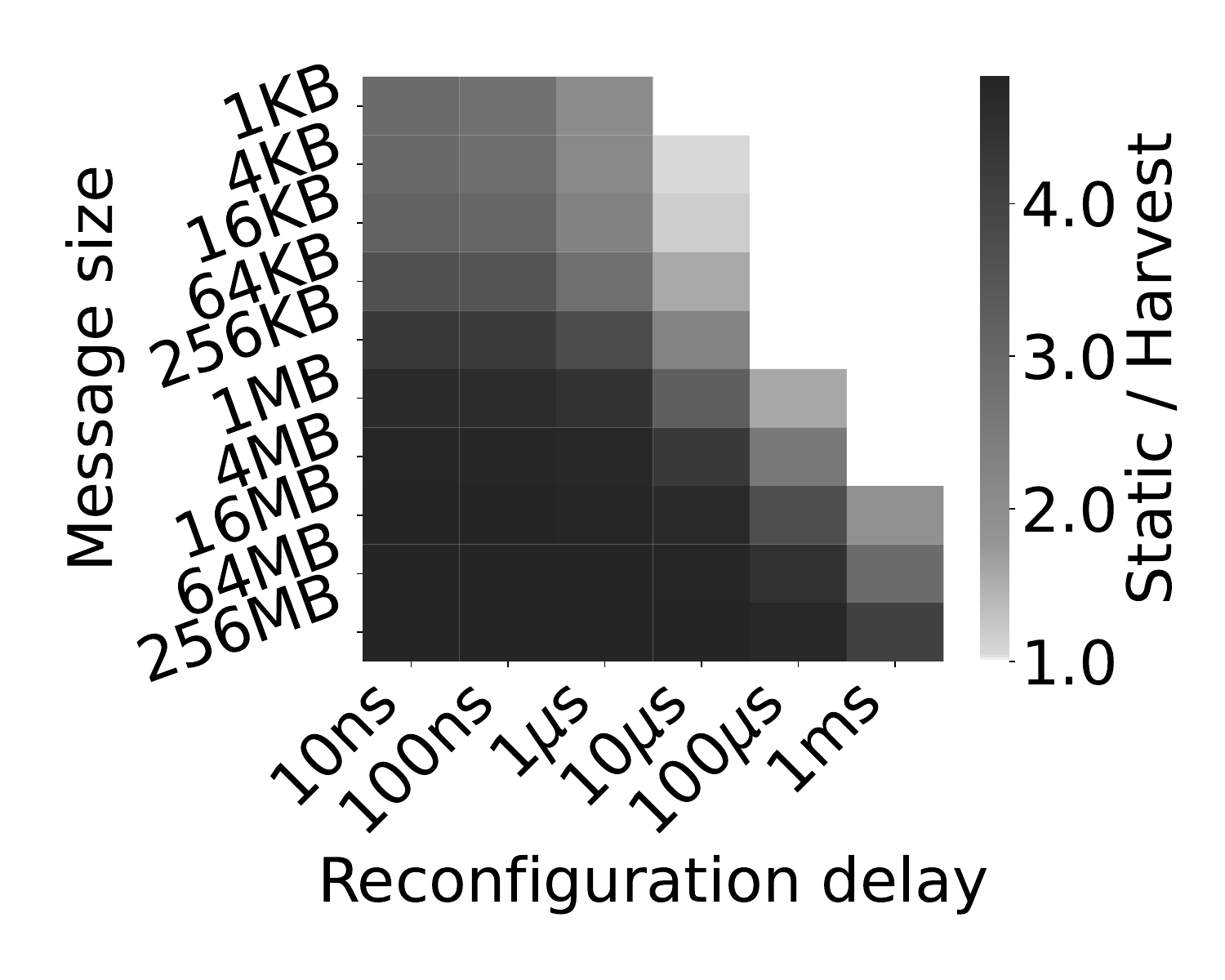}
         \caption{Binomial tree}
         \label{fig:direct_800_row1}
         % \subcaption[]{}
     \end{subfigure}\hfill
     \hfill
     \begin{subfigure}[b]{0.24\textwidth}
         \centering
         \includegraphics[width=\textwidth]{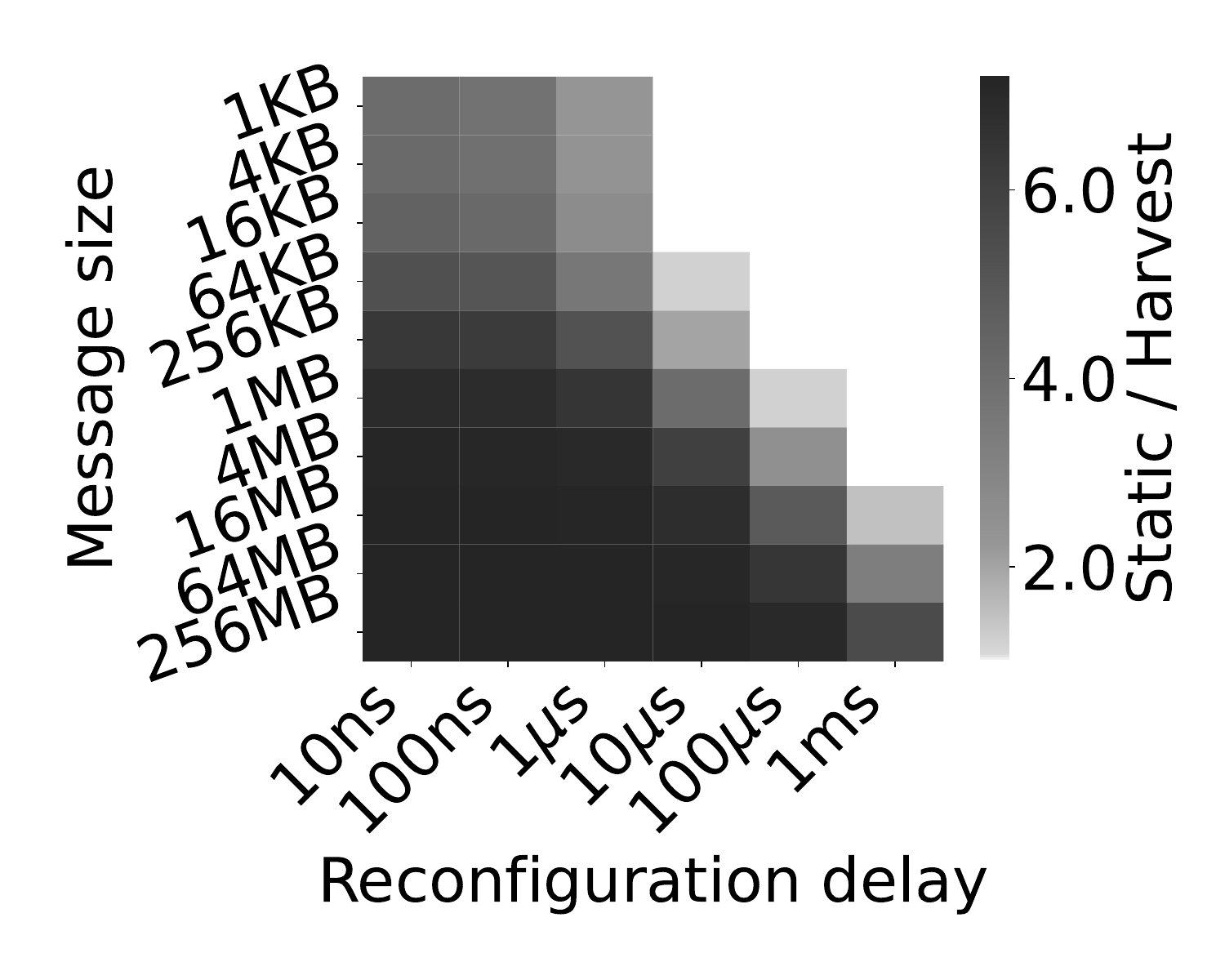}
         \caption{Binary tree}
         \label{fig:direct_best_row1}
         % \subcaption[]{}
     \end{subfigure}

     % --- Row 2 ---
     \begin{subfigure}[b]{0.24\textwidth}
         \centering
         \includegraphics[width=\textwidth]{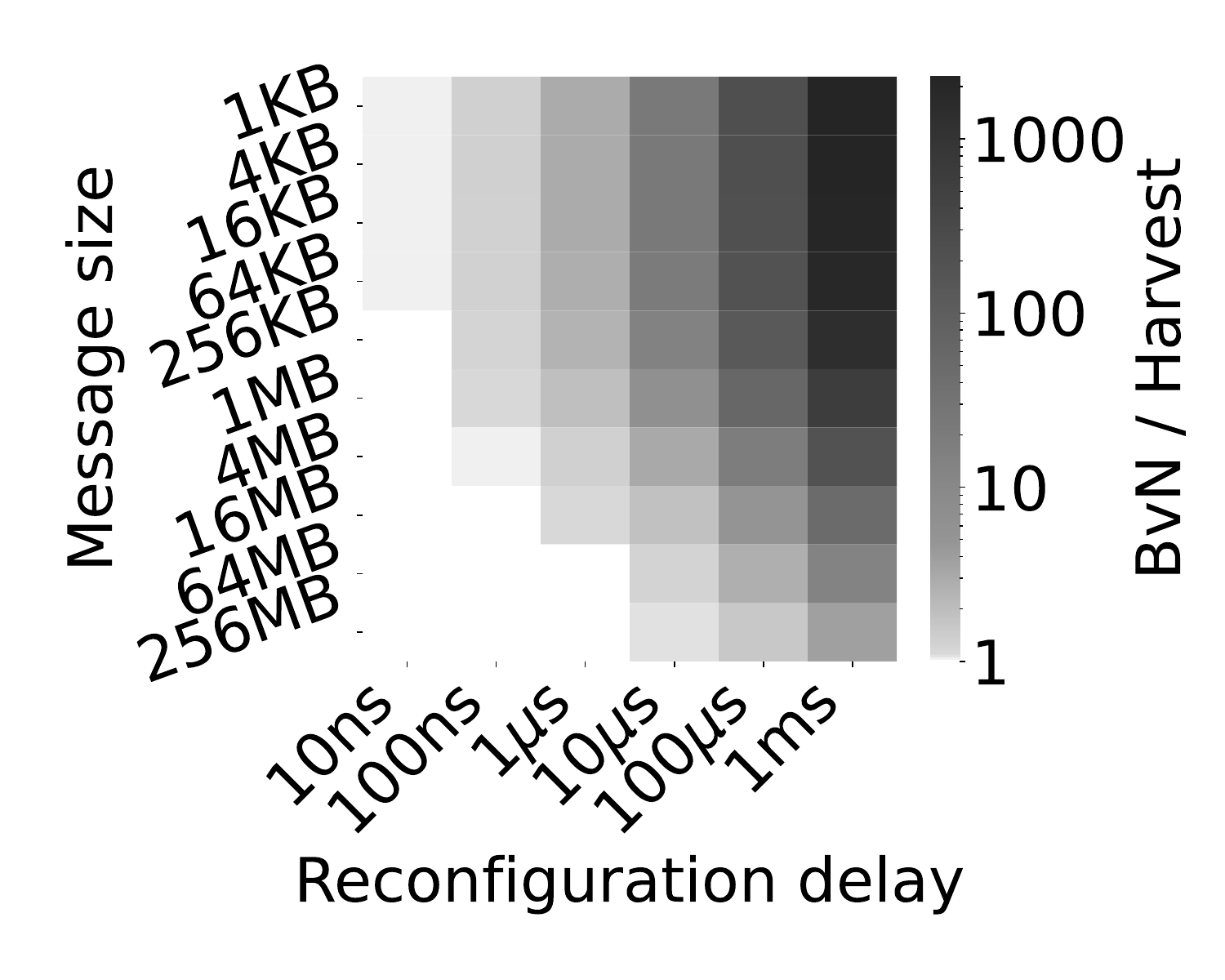}
         \caption{Bruck $r=4$ All-to-All}
         \label{fig:hd_800_row1}
     \end{subfigure}
     \hfill
     \begin{subfigure}[b]{0.24\textwidth}
         \centering
         \includegraphics[width=\textwidth]{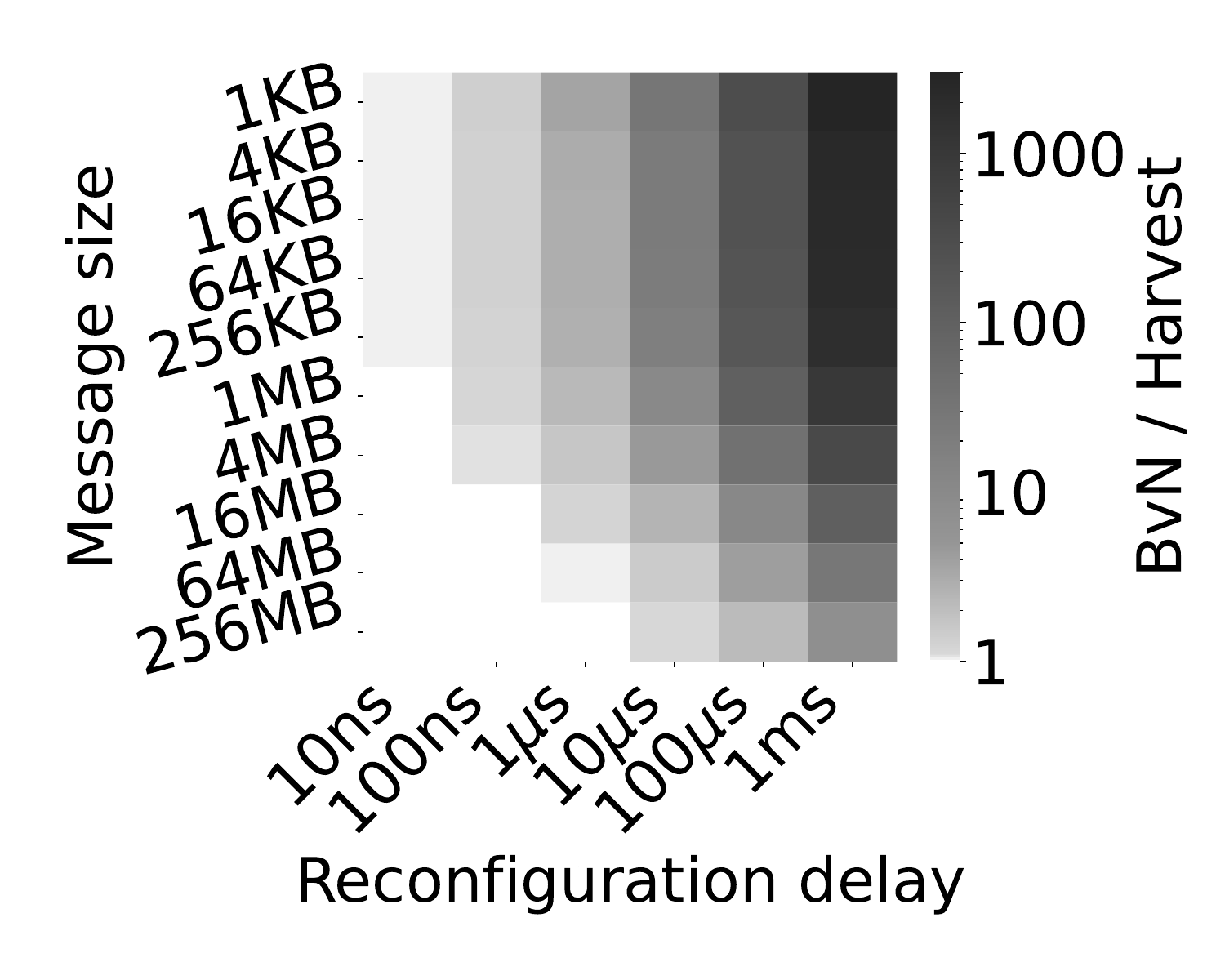}
         \caption{Bruck $r=4$ AllGather}
         \label{fig:swing_800_row1}
         % \subcaption[]{}
     \end{subfigure}
     \hfill
     \begin{subfigure}[b]{0.24\textwidth}
         \centering
         \includegraphics[width=\textwidth]{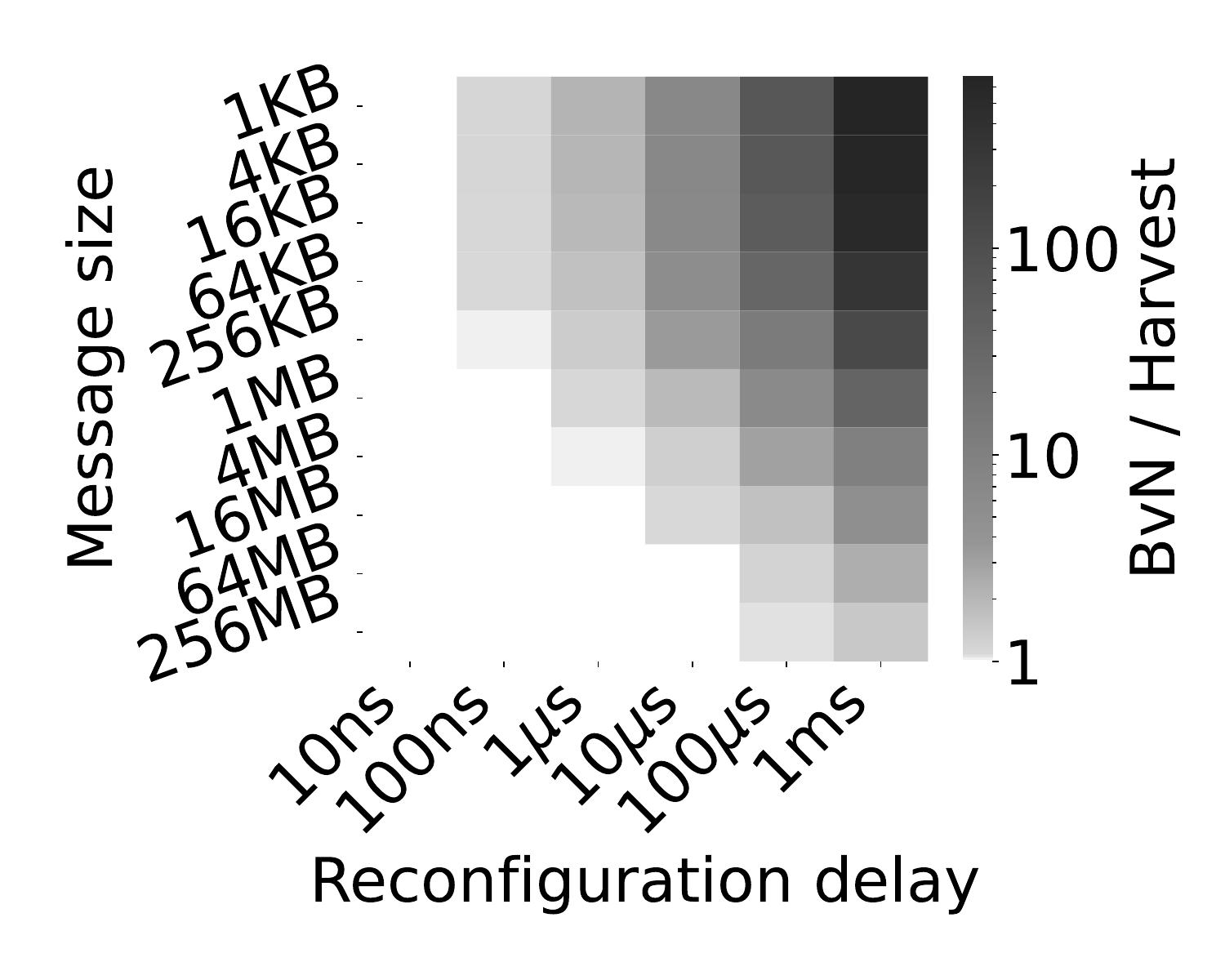}
         \caption{Binomial tree}
         \label{fig:direct_800_row1}
         % \subcaption[]{}
     \end{subfigure}\hfill
     \hfill
     \begin{subfigure}[b]{0.24\textwidth}
         \centering
         \includegraphics[width=\textwidth]{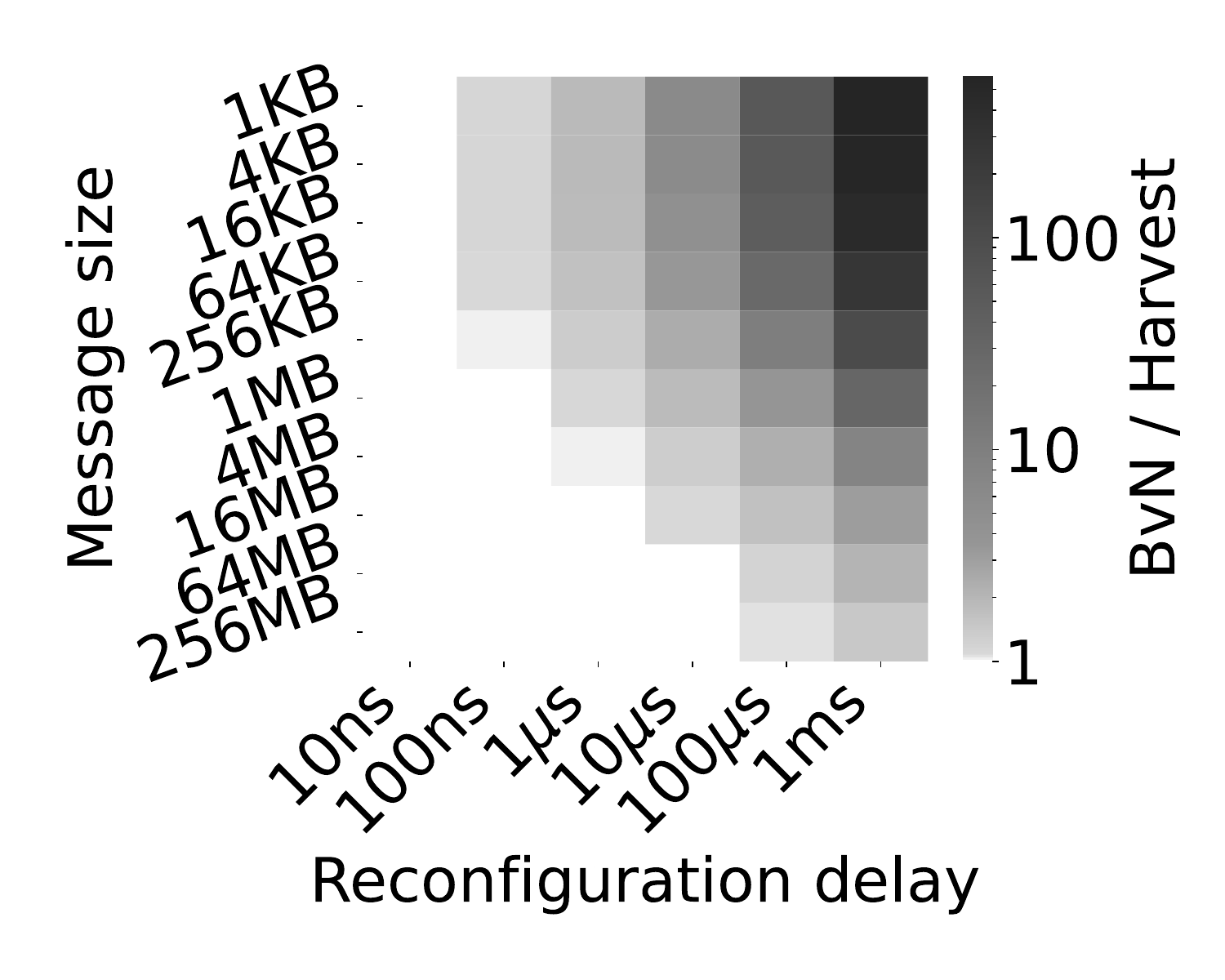}
         \caption{Binary tree}
         \label{fig:direct_best_row1}
         % \subcaption[]{}
     \end{subfigure}
      % --- Row 3 ---
     \begin{subfigure}[b]{0.24\textwidth}
         \centering
         \includegraphics[width=\textwidth]{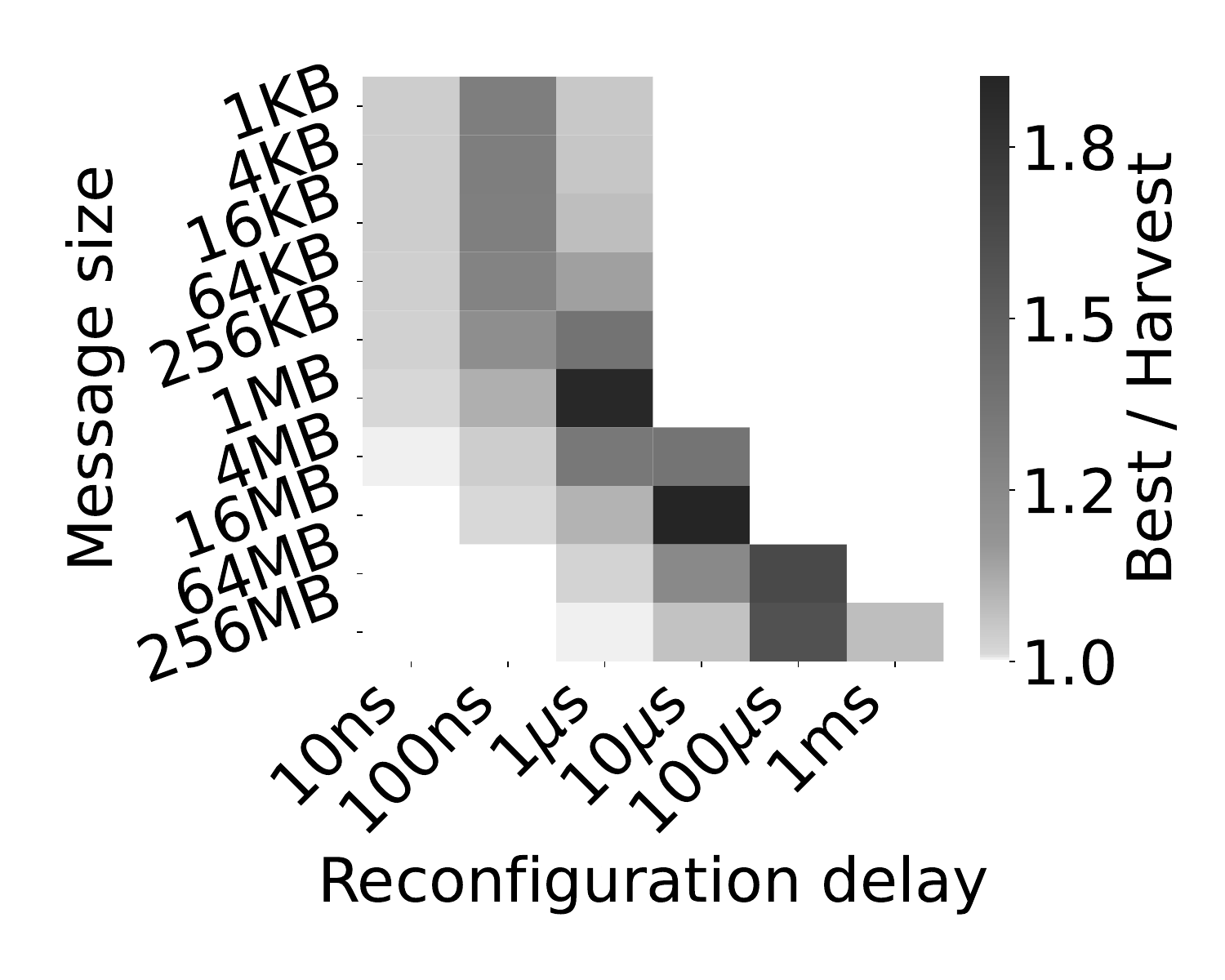}
         \caption{Bruck $r=4$ All-to-All}
         \label{fig:hd_800_row1}
     \end{subfigure}
     \hfill
     \begin{subfigure}[b]{0.24\textwidth}
         \centering
         \includegraphics[width=\textwidth]{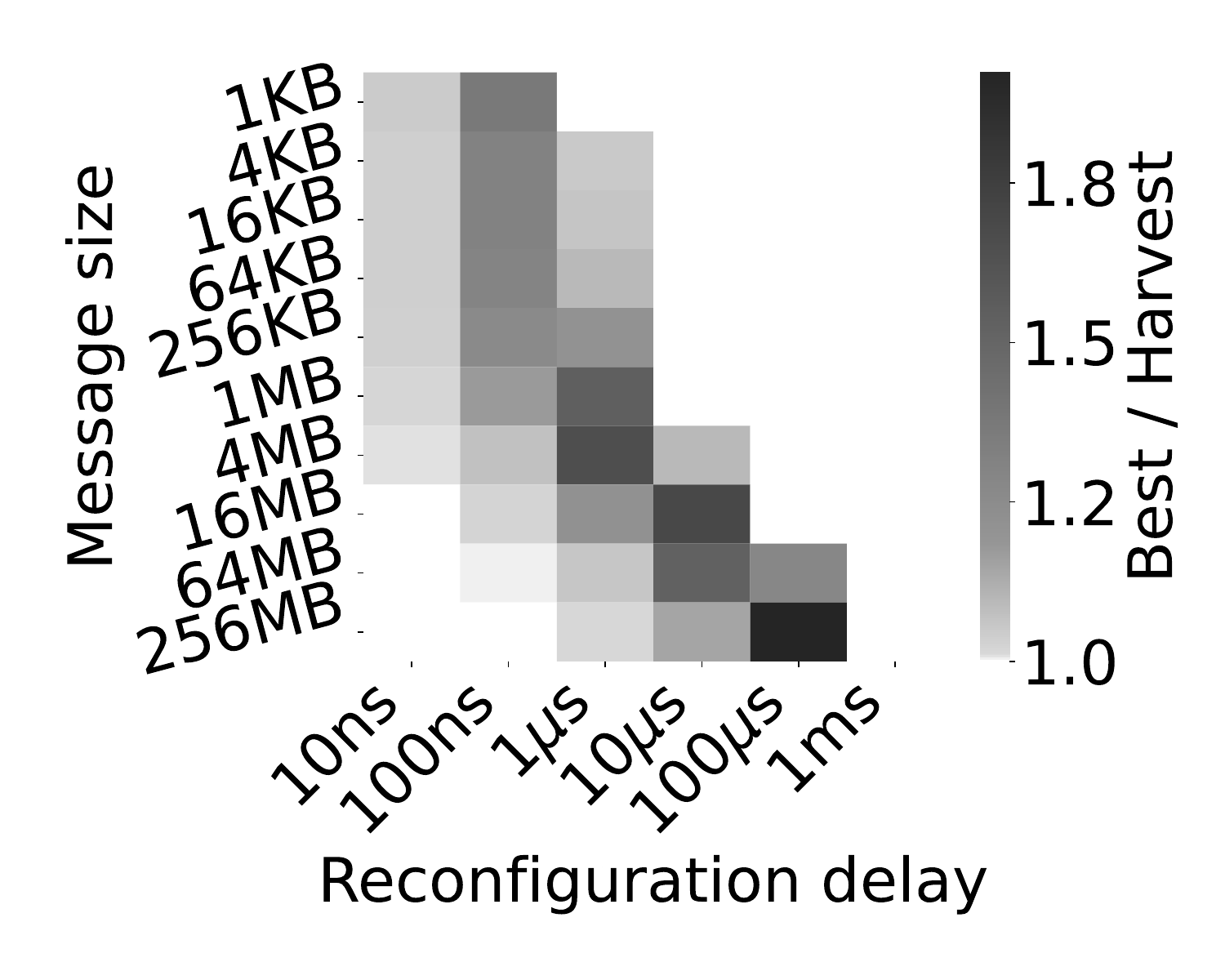}
         \caption{Bruck $r=4$ AllGather}
         \label{fig:swing_800_row1}
         % \subcaption[]{}
     \end{subfigure}
     \hfill
     \begin{subfigure}[b]{0.24\textwidth}
         \centering
         \includegraphics[width=\textwidth]{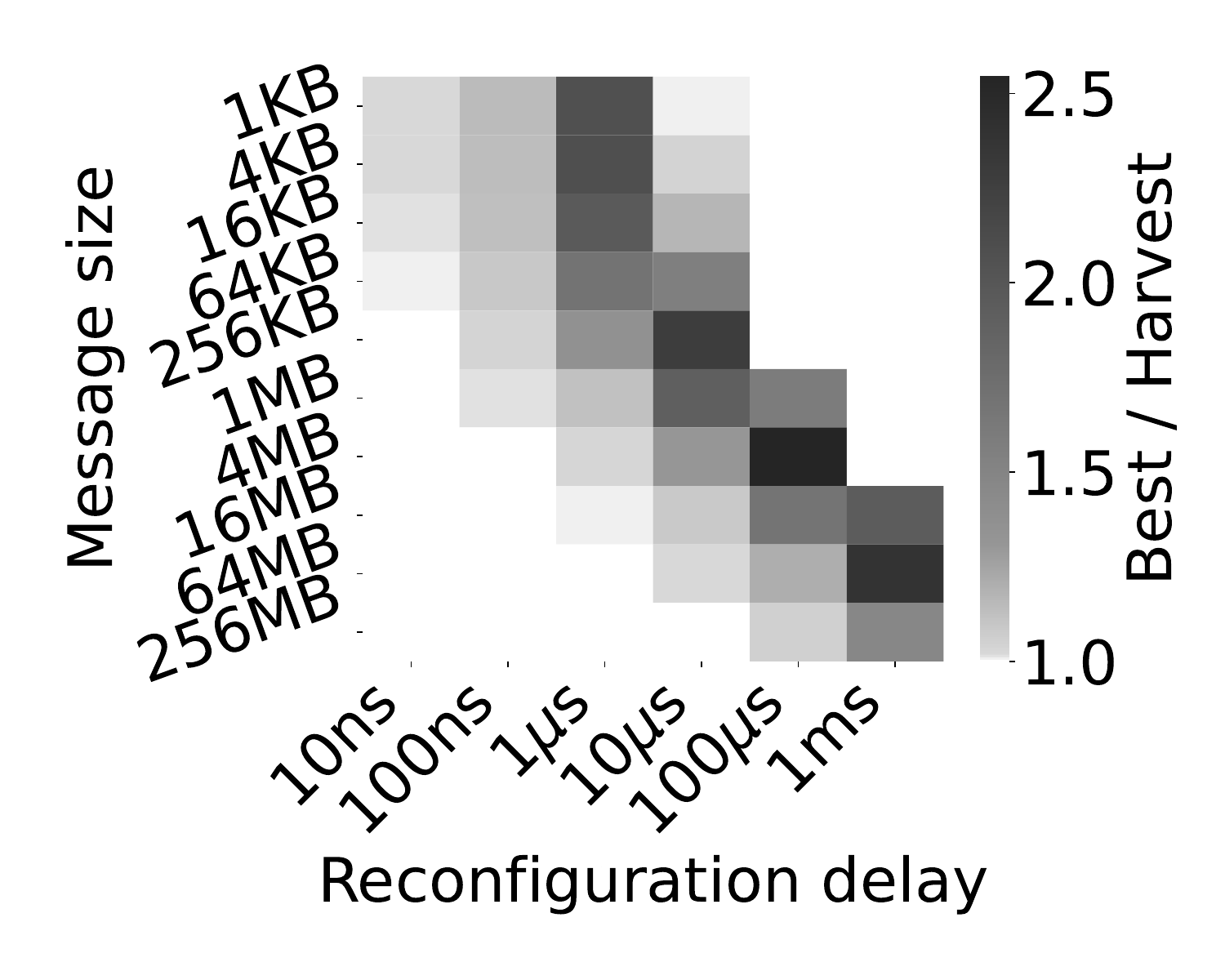}
         \caption{Binomial tree}
         \label{fig:direct_800_row1}
         % \subcaption[]{}
     \end{subfigure}\hfill
     \hfill
     \begin{subfigure}[b]{0.24\textwidth}
         \centering
         \includegraphics[width=\textwidth]{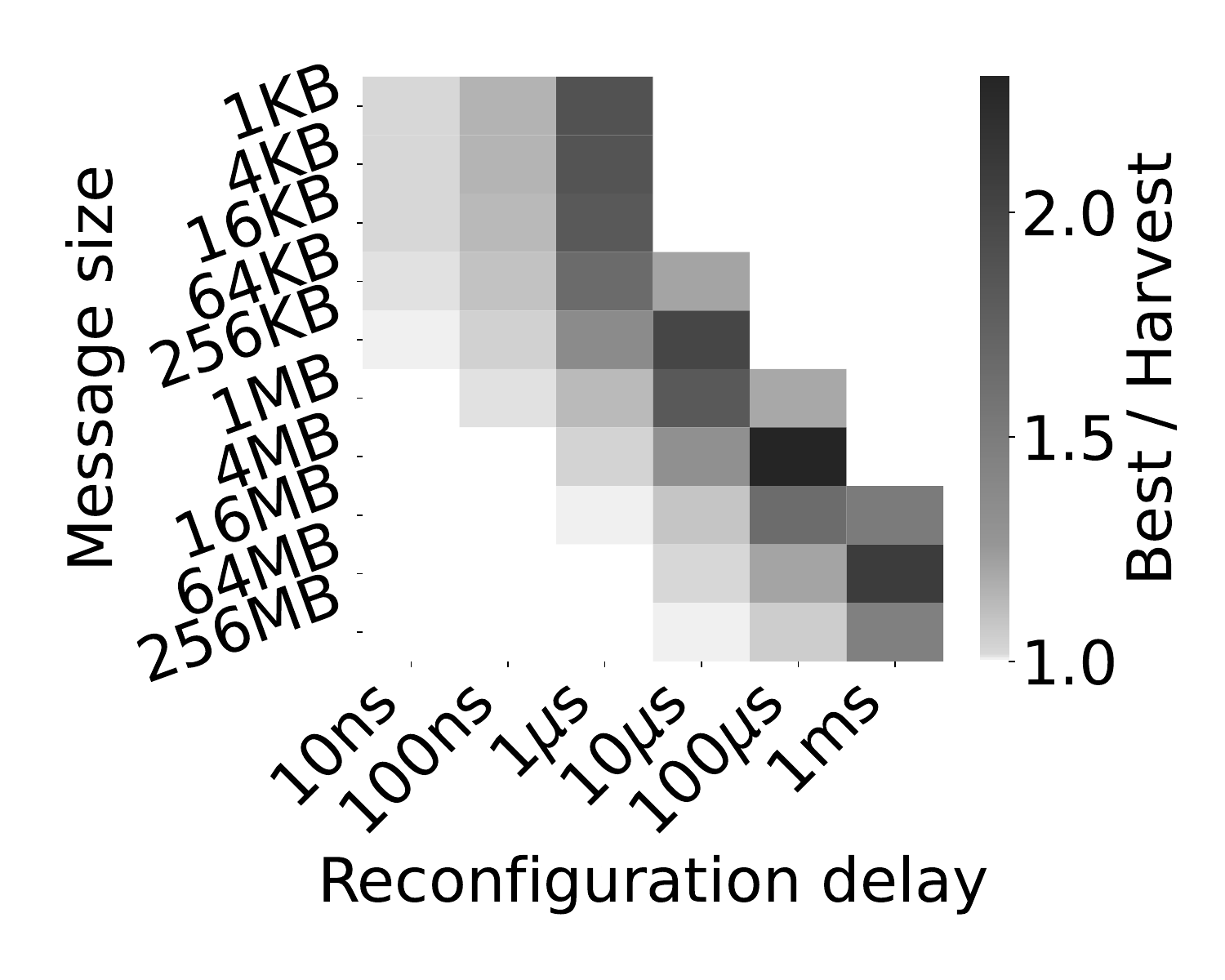}
         \caption{Binary tree}
         \label{fig:direct_best_row1}
         % \subcaption[]{}
     \end{subfigure}

	\caption{[Numerical optimization] Heatmaps showing the speedup in collective completion time achieved by \name relative to static topologies, BvN schedules (reconfiguring at every step), and the best of the two; for Bruck's algorithm for All-to-All and AllGather, as well as binomial tree and binary tree algorithms for broadcast. For Bruck's algorithm with $r=4$, each node performs four send/receive operations per step.}
     \label{fig:bruck-broadcast}
\end{figure*}

\subsection{Results}

\myitem{When does \name outperform static topologies?}

We observe \name consistently outperforms static topologies when reconfiguration delay is low (Figures~\ref{fig:sim-hd-static},~\ref{fig:sim-swing-static}, and~\ref{fig:sim-direct-static}). We use packet-level simulations and a one-dimensional topology.
%Figures~\ref{fig:sim-hd-static},~\ref{fig:sim-swing-static}, and~\ref{fig:sim-direct-static} present packet-level simulation results on a one-dimensional topology across multiple collective algorithms. We observe that \name consistently outperforms static topologies when the reconfiguration delay is low. 
\name speeds up Recursive Doubling by $6.4\times$, Swing by $4.7\times$, and All-to-All by $20\times$ for small message ($1-256$KB) and reconfiguration delays $< 1\mu$s. This is because \name reconfigures the topology to shorten the long-distance steps which in turn reduces both congestion and propagation delays (these dominate small transfers).
%For small message sizes between $1$\,KB and $256$\,KB, and reconfiguration delays below $1\,\mu$s, \name achieves average speedups of $6.4\times$ over Recursive Doubling, $4.7\times$ over Swing, and $20\times$ over direct all-to-all. 
%This improvement arises because \name reconfigures the topology to shorten long-distance communication steps, thereby reducing both congestion and propagation delays that dominate for smaller message sizes.

As reconfiguration delays increase, the benefits of reconfiguration diminish for smaller message sizes, where \name naturally falls back to static schedules. In contrast, for larger messages the gains from reconfiguration persist even at higher delays. For example, with a $1$GB message size, \name achieves up to $3.0\times$ speedup over Recursive Doubling, $3.1\times$ over Swing, and up to $30\times$ over direct All-to-All, even with a $10\mu$s reconfiguration delay. These results show that \name effectively balances reconfiguration overheads against congestion costs, selectively reconfiguring only when the performance benefits outweigh the overhead.

\myitem{When does \name outperform BvN-based schedules?}

Even though it is useful to reconfigure the topology, if we do so in each step we may inflate the completion time (when reconfiguration delays are non-negligible).~\name consistently outperforms BvN-based schedules at higher reconfiguration delays (Figures~\ref{fig:sim-hd-bvn},~\ref{fig:sim-swing-bvn}, and~\ref{fig:sim-direct-bvn}).

The performance gap is most pronounced for small message sizes, where \name strategically limits the number of reconfigurations to reduce overhead. For instance, with a $100\mu$s reconfiguration delay and message sizes between $1$KB and $256$KB, \name achieves on average around $7.3\times$ speedup over BvN schedules for Recursive Doubling, $10\times$ for Swing, and $5.3\times$ for personalized All-to-All. As message sizes increase, the benefits of reconfiguration become more pronounced, and \name adapts toward BvN-like schedules. Notably, even at moderate reconfiguration delays between $10\mu$s and $100\mu$s, \name outperforms BvN schedules by up to $3.0\times$ for $4$MB messages with Recursive Doubling, and up to $4.8\times$ with personalized All-to-All.

Overall, these results highlight the importance of carefully balancing the benefits of reconfiguration against its costs, rather than reconfiguring indiscriminately at every step.

\myitem{When does \name outperform other schedules?}

The natural trade-off between reconfiguration delay and congestion means there is an intermediate regime in which \name outperforms both static and BvN schedules.~\name achieves the best performance in a transitional regime where it selectively reconfigures only a subset of communication steps  (Figures~\ref{fig:sim-direct-best} and~\ref{fig:sim-hd-best}). For example, in personalized All-to-All (Figure~\ref{fig:sim-direct-best}), \name outperforms both static and BvN schedules for message sizes up to $16$MB when reconfiguration delays is in $[10, 100]${$\mu$}s. For larger reconfiguration delays, this regime moves towards larger messages. We observe a similar trend for Recursive Doubling. These results highlight the optimal strategy is neither to avoid reconfiguration entirely nor to reconfigure at every step, but to carefully choose a subset of reconfigurations that balances their benefits against their costs.

% Can we observe the performance trends we find through our synthesis framework match hardware?
\myitem{Does hardware reflect the performance trends revealed by our synthesis framework and cost model?}

We validate the performance trends our synthesis framework predicts with hardware emulations. The results closely match those we saw in simulations (Figures~\ref{fig:emu-bvn},~\ref{fig:emu-static}, and~\ref{fig:emu-best}).  Figure~\ref{fig:emu-bvn} reports the completion-time ratio between BvN schedules and \name for Recursive Doubling~---~it confirms~\name outperforms BvN schedules when reconfiguration delays are high.~\name also outperforms a static ring (Figure~\ref{fig:emu-static}) topology across a wide range of the space ($3\times$ speedup). We also show there  exists a transitional regime in which~\name outperforms both (Figure~\ref{fig:emu-best}).

We parameterize our cost model based on measurements from the testbed to further validate our results. We first estimate the $\alpha$--$\beta$ parameters: we measure the single-hop GPU-to-GPU communication as a function of message size (Figure~\ref{fig:alpha-beta-model}). We find $\alpha = 30.32\mu$s and $\beta = 85.11$Gbps. We then run our synthesizer to compute the completion time it estimates across message sizes and reconfiguration delays. We find the transitional regime and performance trends; as well as the speedup values closely match hardware (Figure~\ref{fig:emu-best}).

% Is \name multi-port, multi-dimension compatible?
\myitem{Is \name multi-port, multi-dimension compatible?}

We compare Swing ($2$-D and $3$-D torus topologies) with static, BvN and \name schedules (Figure~\ref{fig:swing-multidim}). We consider a multidimensional extension of Swing with mirroring, which exploits all available ports simultaneously to maximize link utilization. Numerical evaluations using our synthesizer reveal trends that are consistent with those observed in one-dimensional simulations and hardware emulation.
Across a range of $2$-D and $3$-D torus configurations with $64$ GPUs, \name consistently outperforms all schedules. The regimes in which \name provides the largest gains are dictated by the trade-off between reconfiguration delay and congestion overhead as a function of message size. We observe similar trends for other multi-port collectives, including Bruck’s index and concatenation algorithms, as well as Binomial Tree and Binary Tree broadcast (Figure~\ref{fig:bruck-broadcast}).
All-to-All communication in multi-port settings is different (Figure~\ref{fig:all2all-3d}): the performance advantage over static topologies diminishes as the number of ports increases. This is expected, since All-to-All communication benefits from low average shortest-path lengths, which decrease significantly as the topology degree increases and the network diameter shrinks, reducing the need for reconfiguration. 

\begin{figure}[h]
\centering
\includegraphics[width=0.7\linewidth]{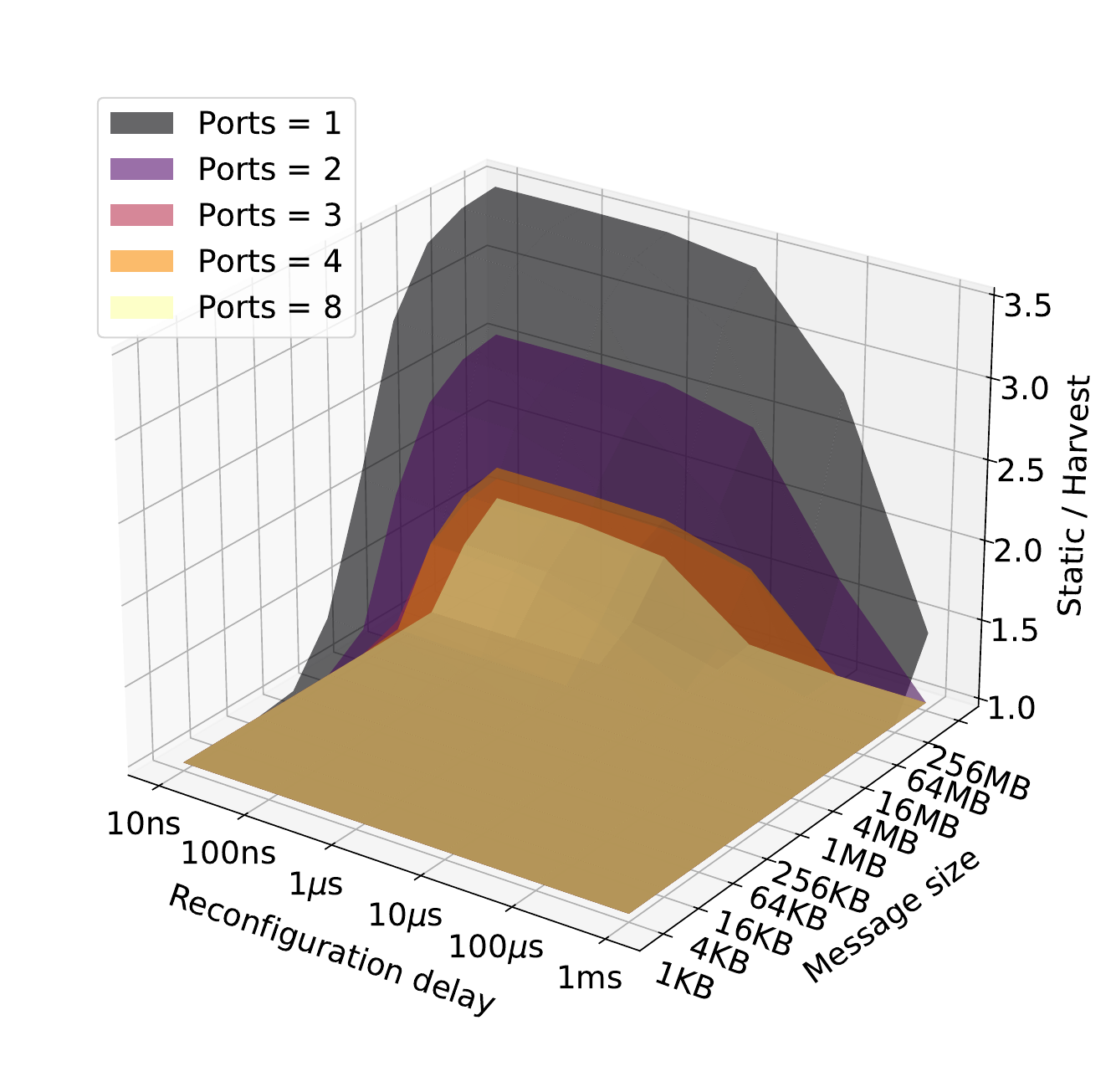}
\caption{[Numerical optimization] All-to-All with Harvest compared to static low-diameter topology with optimal flow schedule~\cite{305352,642949}.}
\label{fig:all2all-3d}
\end{figure}

\begin{figure}[t]
    \centering
    \includegraphics[width=0.8\linewidth]{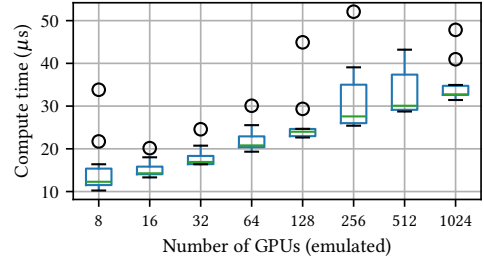}
    \caption{[Hardware emulation] DP Compute time}
    \label{fig:dp-compute}
\end{figure}

\begin{figure}[t]
    \centering
    \includegraphics[width=0.8\linewidth]{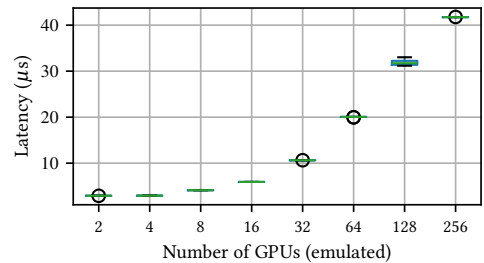}
    \caption{[Hardware emulation] Synchronization time}
    \label{fig:sync-lat}
\end{figure}

\myitem{Is \name practical?}
%Are the overheads of \name practical for real-world use?}

While we do not focus on a full system implementation in this paper, we quantify the practical compute and synchronization overhead introduced by \name. Figure~\ref{fig:dp-compute} reports the time required to solve the dynamic program for Recursive Doubling across a range of network sizes. For typical scale-up deployments of up to $64$ GPUs, the solver completes within $20\mu$s. Even for larger configurations of up to $1024$ nodes, the average runtime remains under $35\mu$s. These schedules can also be computed and cached for future use.

Furthermore, Figure~\ref{fig:sync-lat} shows the synchronization overhead measured among GPUs accessing an array in a shared memory space. Each GPU is programmed to flip a bit at its assigned index in the array and we measure the time it takes for all GPUS to complete this operation. To emulate larger systems, we increase the size of the array and initialize multiple threads on remote GPUs that each flip corresponding entries in the shared array. For an $8$ GPU network, the average synchronization latency is approximately $4\mu$s while a $256$ one shows under $45\mu s$. 

Our results indicate that both the compute and synchronization overheads of the synthesis framework are modest relative to the performance gains it enables.
These overheads are not fundamental~---~we can reduce them further through targeted hardware support e.g., through on-chip schedule synthesis, system-level optimizations where we overlap reconfiguration with computation when the step-wise communication is known \emph{a priori}.

\section{Related Work}
We briefly discuss the significant research efforts in the past in improving collective communication performance.
% We briefly discuss prior work on collective communication and reconfigurable interconnects.

\myitem{Topology-aware collectives:}
A substantial body of work designs collective algorithms specialized for fixed network topologies~\cite{10.1145/3437801.3441620,285084,10.1145/3651890.3672249}.
Algorithms such as Bruck~\cite{642949}, Sack and Gropp~\cite{10.1145/2686882}, Swing~\cite{295653}, and BineTrees~\cite{10.1145/3712285.3759835} optimize communication for multiported static interconnects.
Bruck's algorithm further generalizes AllReduce to $\log_d(n)$ steps for $d$-port networks, but does not model network congestion.
More broadly, collective synthesis approaches~\cite{10.1145/3437801.3441620,285084,10.1145/3651890.3672249,305352} assume a static topology throughout execution.
In contrast, our work focuses on \emph{topology synthesis}, enabling dynamic reconfiguration during a collective while explicitly balancing reconfiguration delay and congestion.

\myitem{Reconfiguration-aware circuit-switching:}
Reconfigurable circuit-switched network topologies have been widely studied in datacenter settings~\cite{10.1145/3387514.3406221,10.1145/3098822.3098838,10.1145/3651890.3672273,10.1145/3579312,10.1145/3651890.3672248,10.1145/3519935.3520020,10.1145/2619239.2626328,10.1145/2896377.2901479,10.1145/2486001.2486007,1230204,10.1145/3409964.3461786,10.1145/2934872.2934911,10.1145/3651890.3672222,10.1145/2716281.2836126}.
Early systems often assumed negligible reconfiguration delays~\cite{10.1145/3387514.3406221,10.1145/3098822.3098838}, while later work incorporated reconfiguration cost into optimization objectives~\cite{10.1145/2716281.2836126,1230204,10.1145/2896377.2901479}.
Many of these approaches operate over discrete topology choices and model reconfiguration as a switching cost~\cite{10.1145/2896377.2901479,10.1145/2716281.2836126}, which limits their ability to capture congestion and routing flexibility.
Opera~\cite{opera} allows multi-hop forwarding within a topology but does not adapt to reconfiguration delays.
Our approach bridges reconfiguration cost and network throughput via maximum concurrent flow, enabling multi-hop routing while explicitly accounting for reconfiguration delay.

\myitem{Circuit-switching for collectives:}
Recent work has explored reconfigurable interconnects tailored to collective communication~\cite{10.1145/3748273.3749203,10.1145/3696348.3696856}.
Chronos~\cite{10.1145/3748273.3749210} preschedules circuits using step-wise collective structure, but does not explicitly consider reconfiguration delays.
Actina~\cite{10.1145/3712285.3759842} supports dynamic reconfiguration for ML workloads, but the topology remains static within collective execution.

\myitem{High-throughput topologies:}
Datacenter network topologies have been extensively studied~\cite{10.1145/1402958.1402967,10.1145/1592568.1592576,10.1145/2999572.2999580,180604,7013016,227667}.
Clos-based networks achieve full throughput at high cost~\cite{10.1145/3651890.3672265,10.1145/3651890.3672233,10.1145/1402958.1402967}.
Expander-based topologies reduce hardware cost but sacrifice throughput~\cite{10.1145/2999572.2999580,10.1145/3452296.3472913}.
Torus-based networks align well with modern parallel workloads despite low bisection bandwidth~\cite{10.1145/3579371.3589350,295551}, but can incur substantial congestion.
Our work targets such scale-up domains, where adaptive reconfiguration can reduce congestion and improve collective performance.

\section{Conclusion}

\name is a systematic approach that optimizes reconfigurable topologies for collective communication. It explicitly balances propagation delay, congestion, and reconfiguration overhead. \name synthesizes reconfiguration schedules. We show it is general and finds provably optimal schedules for recursive doubling AllReduce. Our schedules significantly reduce collective completion time. We show the benefits of reconfiguration depend critically on when and how we apply it.

Our framework opens several directions for future work which include joint synthesis of collective communication and interconnect topologies. This work lays the foundation for adaptive photonic scale-up domains where collectives and topologies co-evolve.

\label{bodyLastPage}

\bibliographystyle{plainurl}
\bibliography{references}

\appendix

\section{MISOCP Formulation for the Subproblem}\label{app:misocp}

We formulate a Mixed-Integer Second-Order Conic Program (MISOCP) to compute
the optimal topology $G_{a,b}$ that minimizes the completion time of steps
$a$ through $b$ without reconfiguration. Recall that this minimization objective is our subproblem in the schedule synthesis~\S\ref{sec:subproblem}.

\paragraph{Variables:}
\begin{itemize}
    \item $x_{u,v} \in \mathbb{Z}_{\ge 0}$: number of directed links from node $u$ to node $v$.
    \item $f^{(i)}_{s,t}(u,v) \ge 0$: flow routed on directed edge $(u,v)$ for demand $(s,t)$
    in step $i$.
    \item $T_i \ge 0$: transmission time of step $i$.
    \item $\theta_i \ge 0$: throughput scaling factor for step $i$.
\end{itemize}

\paragraph{Parameters:}
\begin{itemize}
    \item $d$: maximum in-degree and out-degree per node.
    \item $c$: capacity of a single directed edge.
    \item $D^{(i)}_{s,t} = m_i \cdot \mathcal{M}_i(s,t)$: demand from node $s$ to $t$ in step $i$.
    \item $\beta$: inverse of the link bandwidth.
\end{itemize}

Our goal is to minimize the sum of tranmission times $T_i$ for the sequence of steps $a$ through $b$ i.e., minimizing the total transmission time for these steps together. 

\paragraph{Objective:}
\begin{align}
\min \quad \sum_{i=a}^{b} T_i
\end{align}

Edges in the topology are binary variables in our formulation. To capture the maximum number of links available at each node, we impose node degree constraints as follows.

\paragraph{Degree constraints:}

\begin{align}
\sum_{v \in V} x_{u,v} &\le d \quad &&\forall u \in V \label{eq:deg-out} \\
\sum_{u \in V} x_{u,v} &\le d \quad &&\forall v \in V \label{eq:deg-in}
\end{align}

The rest of the formulation follows standard maximum concurrent flow formulation. In particular, we consider flow variables $f_{s,t}^(i)(u,v)$ sent on edge $(u,v)$, corresponding to the demand between $(s,t)$ in step $i$. To satisfy flow conservation, we incorporate the following constraints at source, destination, and intermediate nodes.

\paragraph{Flow conservation for each step:}
\begin{align}
\sum_{v} f^{(i)}_{s,t}(u,v) - \sum_{v} f^{(i)}_{s,t}(v,u)
=
\begin{cases}
\theta_i \cdot D^{(i)}_{s,t}, & u=s \\
- \theta_i \cdot D^{(i)}_{s,t}, & u=t \\
0, & \text{otherwise}
\end{cases}
\quad
\forall i,s,t,u
\label{eq:flow-cons}
\end{align}

Further, the total flow between $(u,v)$ must satisfy the available capacity between the two nodes given by $c\cdot x_{u,v}$, where $x_{u,v}$ is the variable indicating the number of edges between $u,v$.

\paragraph{Edge capacity constraints:}
\begin{align}
\sum_{s,t} f^{(i)}_{s,t}(u,v)
\;\le\;
c \cdot x_{u,v}
\quad
\forall i,(u,v)
\label{eq:capacity}
\end{align}

Finally, since our objective is to minimize transmission time, we express the following constraint as an inequality.
\paragraph{Transmission time:}
\begin{align}
T_i \;\ge\; \frac{\beta\, m_i}{\theta_i},
\quad \theta_i \ge 0,\; T_i \ge 0
\qquad \forall i .
\label{eq:tx-time}
\end{align}
The transmission time for step $i$ is exactly $\frac{\beta m_i}{ \theta_i}$. However, writing this as an inequality does not relax the solution: the constraint is tight at optimality, since any strictly larger value of $T_i$ would increase the objective.

The above constraint is equivalent to the bilinear constraint
\begin{align}
\theta_i \cdot T_i \;\ge\; \beta\, m_i,
\quad \forall i.
\label{eq:tx-time-bilinear}
\end{align}

\noindent
Constraint~\eqref{eq:tx-time-bilinear} admits a second-order conic
representation. In particular, it is equivalent to the following
second-order cone (SOC) constraint:
\begin{align}
\left\|
\begin{bmatrix}
2\sqrt{\beta m_i} \\
\theta_i - T_i
\end{bmatrix}
\right\|_2
\;\le\;
\theta_i + T_i,
\quad \forall i.
\label{eq:tx-time-soc}
\end{align}

\paragraph{Variable type:}
\begin{align*}
f^{(i)}_{s,t}(u,v) &\ge 0 \\
x_{u,v} &\in \mathbb{Z}_{\ge 0} \\
\theta_i & \in (0, 1] \\
T_i &\in \mathbb{R}_{\ge 0}
\end{align*}

The solution to this formulation yields the optimal topology $G_{a,b}$, which is constructed from the decision variables $x_{u,v}$ indicating the number of directed edges between every node pairs. The formulation also returns the per-step completion times $T_i$, and the total completion time $\sum_{i=a}^b T_i$, which are used by the outer dynamic program.

We restrict the search space in our evaluations to a collection of topologies: \first shifted rings, \second an expander, \third shifted torus topologies, \fourth topologies that match the communication pattern for each step of the collective, and the synthesizer selects the best topology for each step of the collective. This restriction converts the MISOCP into SOCP without integer variables for finding optimal topologies. As discussed in \S\ref{sec:subproblem}, the sequential dependencies between steps induce the conic structure of the formulation. In contrast, a standard maximum concurrent flow formulation assumes that all demands are available for transmission simultaneously and does not capture such dependencies, resulting in a simple linear program.

\label{LastPage}

\end{document}